\documentclass[12pt]{report}
\usepackage[onehalfspacing]{setspace}
\usepackage[margin=0.9in]{geometry}

\usepackage[textwidth=1.0in, shadow,backgroundcolor=yellow,textsize=tiny]{todonotes} % add 'disable' to the options to disregard all todos in the output
\newcommand{\itamar}[1]{} %{\todo[color=yellow!20] {Itamar: #1}}
\newcommand{\itamarc}[1]{}

\newcommand{\redtext}[1]{\textcolor{black}{#1}}

\usepackage[toc,page]{appendix}

\usepackage{amsmath, amsthm, tabto, longtable}
\usepackage{algpseudocode, algorithm}
\usepackage{varwidth} % Used for indenting broken long line in PGM's pseudo-code
\usepackage{xspace, lipsum} % lipsum is for page start / clear numbering
\algnewcommand{\AND}{\textbf{and}\xspace}
\algnewcommand{\OR}{\textbf{or}\xspace}
\usepackage{paralist}
\usepackage{hyperref} 
\hypersetup{
    colorlinks,
    citecolor=black,
    filecolor=black,
    linkcolor=black,
    urlcolor=black
}
\usepackage{float, graphicx, color, microtype, pstricks, pst-node, pst-plot, afterpage}
\usepackage[caption=false]{subfig} %\usepackage{subcaption} % Cannot be used in cooperation with subfig.
\usepackage{subfloat} %deprecated
\usepackage{wesa}
\usepackage[T1]{fontenc}
\usepackage[utf8]{inputenc}
\usepackage{amssymb, tabularx, array, tabu}
\usepackage{hhline, multirow, epsf, epstopdf}
\usepackage{color, colortbl}
\usepackage{subfiles}

\usepackage{pgfplots}
\pgfplotsset{compat=1.14}
\usepgfplotslibrary{units}
\usepackage{tikz, tikz-qtree}
\usetikzlibrary{automata,positioning, spy, fit, arrows, shapes, calc, patterns, automata, positioning, quotes, tikzmark}
\usepackage{siunitx}
\usepackage{verbatimbox}

\makeatletter
\newcommand*{\algrule}[1][\algorithmicindent]{\makebox[#1][l]{\hspace*{.5em}\vrule height .75\baselineskip depth .25\baselineskip}}%

\newcount\ALG@printindent@tempcnta
\def\ALG@printindent{%
    \ifnum \theALG@nested>0% is there anything to print
        \ifx\ALG@text\ALG@x@notext% is this an end group without any text?
            \addvspace{-3pt}% FUDGE for cases where no text is shown, to make the rules line up
        \else
            \unskip
            \ALG@printindent@tempcnta=1
            \loop
                \algrule[\csname ALG@ind@\the\ALG@printindent@tempcnta\endcsname]%
                \advance \ALG@printindent@tempcnta 1
            \ifnum \ALG@printindent@tempcnta<\numexpr\theALG@nested+1\relax% can't do <=, so add one to RHS and use < instead
            \repeat
        \fi
    \fi
    }%
\usepackage{etoolbox}
\patchcmd{\ALG@doentity}{\noindent\hskip\ALG@tlm}{\ALG@printindent}{}{\errmessage{failed to patch}}
\makeatother
\newtheorem{theorem}{Theorem}
\newtheorem{corollary}[theorem]{Corollary}
\newtheorem{lemma}[theorem]{Lemma}
\newtheorem{proposition}[theorem]{Proposition}

\newcommand{\abs}[1]{\left\vert#1\right\vert}
\newcommand{\set}[1]{\left\{#1\right\}}

\newcommand{\eps}{\varepsilon}

\newcommand{\floor}[1]{\lfloor #1 \rfloor}

\DeclareMathOperator*{\argmin}{arg\,min}

\newcommand{\vspacebelowcaption}{0.27cm}
\newcommand {\marksize}         {2.5 pt}
\newcommand {\figfontsize}      {\normalsize}

\DeclareMathOperator{\mSAM}{SA} % math-mode Speculatively Admit Multiple values naive alg'
\newcommand{\SAM}{$\mSAM$}
\DeclareMathOperator{\mSAMWP}{SA} % math-mode Speculatively Admit Multiple Work Mutliple Profit naive alg'
\newcommand{\SAMWP}{$\mSAMWP$}

\DeclareMathOperator{\mSAO}{SA^*} % math-mode
\newcommand{\SAO}{$\mSAO$}
\DeclareMathOperator{\mSAOWP}{SA^*} % math-mode
\newcommand{\SAOWP}{$\mSAOWP$}

\newcommand{\alg}{Alg}
\DeclareMathOperator{\malg}{Alg} % math-mode Alg
\newcommand{\opt}{OPT}
\newcommand{\makeroom}{MakeRoom()}

\newcommand{\subopt}{SubOPT}
\DeclareMathOperator{\msubopt}{SubOPT} % math-mode Alg
\DeclareMathOperator{\FILL}{fill}
\DeclareMathOperator{\FLUSH}{flush}

\DeclareMathOperator{\W}{CW} % Cumulative work on packets in the buffer
\newcommand{\Cs}{G} % Selected ("Guaranteed") Class
\newcommand{\CsK}{G^K} % Known pkts from the Selected ("Guaranteed") Class
\newcommand{\CsU}{G^U} % Unknown pkts from the Selected ("Guaranteed") Class
\newcommand{\mFILL}{^{(\FILL)}}
\newcommand{\mFLUSH}{^{(\FLUSH)}}
\renewcommand{\W}{^{(W)}}
\renewcommand{\P}{^{(P)}}
\newcommand{\U}{^{(U)}}
\newcommand{\K}{^{(K)}}

\newcommand{\hfull}{Gfull}
\newcommand{\Autp} {A^{\U}_{(t_p)}}
\newcommand{\Au}{A^{\U}} %number of U-pkts which arrive at cycle t
\newcommand{\Ak}{A^{\K}} %number of K-pkts which arrive at cycle t
\newcommand{\malpha}{\alpha} %An entry in Markov's transition matrix
\newcommand{\deltaW}{\delta_W} %The maximal mult' ratio between the works of two G-pkts
\newcommand{\deltaP}{\delta_V} %The maximal mult' ratio between the profits of two G-pkts
\newcommand{\mW}{\ell_W} % # of work-classes
\newcommand{\mP}{\ell_V} % # of profit-classes
\newcommand{\minWork}{W_0} %minimal work of packets; used in the lower bounds proof
\newcommand{\mbest}{(\minWork,V)} % math-mode base of log for SA_BW
\newcommand{\mworst}{(W,1)} % math-mode base of log for SA_BW
\newcommand{\algc}{\malg_c} % A c-competitive alg'; used in the brief presentation of the CnRS technique.

\newcommand{\SetSs}{N} %Set of data stores
\newcommand{\NumSs}{n} %Number of data stores
\newcommand{\SetSp}{N_x} %Set of data stores with positive indication
\newcommand{\NumSp}{n_x} %Number of positive indications
\newcommand{\mS} {S}   % Set of element in a data store
\newcommand{\ind}{I} %Indication of indicator (either 0 or 1)

\DeclareMathOperator{\fpr}{FP} %false pos. rate

\newcommand{\fprj}{\fpr_j}
\newcommand{\Phit}{p{^h}} %per-DS hit rate
\newcommand{\Phitj}{p_j^h} %per-DS hit rate
\newcommand{\mr}{\rho} % misindication ratio
\newcommand{\mc} {c}   %cost
\newcommand{\mbeta}{\beta} %miss penalty
\newcommand{\mL}{D} %{\mathcal{L}} %Set of data stores
\DeclareMathOperator{\knap}{Knap}
\DeclareMathOperator{\potential}{Pot}
\DeclareMathOperator{\pp}{PP}
\DeclareMathOperator{\dsalg}{DS}
\newcommand{\ppapproxalg}{$\dsalg_{\pp}$}
\newcommand{\fpapproxalg}{$\dsalg_{\knap}$}
\newcommand{\cacheprob}{DSS}
\newcommand{\cpi}{CPI} %Cheapest Positive Indication
\newcommand{\epi}{EPI} %Every Positive Indication
\newcommand{\umb}{$\dsalg_{\knap}$} %Our Umbrella (DS_{Knap}) algorithm
\newcommand{\pot}{$\dsalg_{\potential}$} %Our Potential algorithm
\newcommand{\fpo}{FPO} %Optimal policy
\newcommand{\Pone}{q} % Pr (I_j(x)==1)
\newcommand{\costhomo}{\tilde{\phi}} %cost func' in the homo' case

\newcommand{\mrgline}{\ell} %line in the merge stage in the log(\mbeta) alg'
\DeclareMathOperator{\PGF}{PGF} % Probability Generating Function
\DeclareMathOperator{\dist}{dist} % hop count distance between nodes in simulation
\DeclareMathOperator{\BW}{BW} % bottleneck bandwidth between nodes in simulation

\DeclareMathOperator{\Dset}{O^{*}} % Set of DSs used by Opt; used in the alg's pot, knap
\DeclareMathOperator{\Oset}{O^{*}} % Set of DSs used by Opt
\DeclareMathOperator{\mbase}{2} % The base of the log in the matrices MM and the vectors
\DeclareMathOperator{\pmr}{PGM}
\newcommand{\pgmalg}{$\pmr$}

\newcommand{\OOset}[2]{O_{#2}^{#1}} % Set of DSs used by Opt out of \piset{#2}^{#1}.
\newcommand{\piset}[2]{N_{#2}^{#1}} % Set of DSs with positive ind' where the cost of each DS is in [2^{#2 \cdot 2^{#l}}, 2^{(#2+1) \cdot 2^{#l}})
\newcommand{\mvec}[2]{V_{#2}^{#1}} % vertex with sets of candidate DSs by \pmralbg
\newcommand{\logbeta}{r}

\newcommand{\algtop}{APSR}%Alg' to generate / destroy sched's in a real system

\title{\textbf {Advanced Algorithms in Heterogeneous and Uncertain Networking Environments \\
\vspace{2.5 cm}
\small{
Thesis submitted in partial fulfillment\\
of the requirements for the degree of\\ 
“DOCTOR OF PHILOSOPHY” } \\
\vspace{2.5 cm}
by \\ 
\vspace{1 cm}
Itamar Cohen} \\ 
\vspace{2 cm}
    {\fontfamily{lmr}\selectfont
        \small{
        \textbf{Submitted to the Senate of Ben-Gurion\\
        University of the Negev}\\ 
        \vspace{0.5 cm}
        \textbf{22.10.2019}\\ \textbf{Beer-Sheva, Israel}}
    }    
}
\author{}
\date{}

\begin{document}

\maketitle
% \center{
%     Submitted to the Senate of Ben-Gurion University of the Negev, \\
%     22.10.2019 \\
%     Beer-Sheva, Israel \\
% }

\pagenumbering{gobble}% Remove page numbers (and reset to 1)
\clearpage
% \maketitle
% \newpage
\tableofcontents
\newpage
\listoffigures
\newpage
\listoftables
\newpage
\clearpage
\pagenumbering{arabic}% Arabic page numbers (and reset to 1)
\begin{abstract}
Communication networks are used today everywhere and in every scale: starting from small Internet of Things (IoT) networks at home, via campus and enterprise networks, and up to tier-one networks of Internet providers. 
Accordingly, network devices should support a plethora of tasks 
with highly heterogeneous characteristics in terms of processing time, bandwidth energy consumption, deadlines and so on. 
Evaluating these characteristics and the amount of currently available resources for handling them 
requires analyzing all the arriving inputs, gathering information from numerous remote devices, and integrating all this information. Performing all these tasks in real time is very challenging in  
today's networking environments, which are characterized by tight bounds on the latency, and always-increasing data rates. 
Hence, network algorithms should typically make decisions under uncertainty. 

\redtext{This work addresses optimizing performance in heterogeneous and uncertain networking environments. 
We begin by detailing the sources of heterogeneity and uncertainty and show that uncertainty appears in all layers of network design, including the time required to perform a task; the amount of available resources; and the expected gain from successfully completing a task. 
Next, we survey current solutions and show their limitations. Based on these insights we develop general design concepts to tackle heterogeneity and uncertainty, and then use these concepts to design practical algorithms. For each of our algorithms, we provide rigorous mathematical analysis, thus showing worst-case performance guarantees. Finally, we implement and run the suggested algorithms on various input traces, thus obtaining further insights as to our algorithmic design principles.
We exemplify our approach on three concrete networking environments, namely: (i) packet buffers, (ii) network caching, and (iii) 
placement of virtual machines in the cloud.}

\itamar{Consider adding: Our work highlights the complex interplay between multiple optimization criteria, run-time restraints, and budget constraints in uncertain and heterogeneous networking environments.}
\end{abstract}

\newpage
\chapter{Introduction}

A common saying claims that ``it is difficult to make predictions, especially about the future''. Indeed, the last fifty years saw a tremendous research effort focusing on online problems, where future input is unknown, and input is revealed on-the-fly, during execution. Traditionally, this literature assumed that while predicting the future is very difficult, measuring present and past inputs is straight-forward and affordable.

Unfortunately, this assumption is far too optimistic in modern networking environments. Having complete, up-to-date information about the input requires analyzing huge quantities of data, maintaining data structures, and gathering and integrating information from multiple sources and devices. These tasks don't scale fast enough to meet the demands of today's networking environments, which are characterized by high heterogeneity and always-increasing data rates. 
Furthermore, allocating more resources for performing these tasks leaves fewer resources available for processing and transmitting the data. In this aspect, the problem of decision making under uncertainty resembles  Heisenberg's uncertainty principle, which claims that the act of measuring necessarily impacts the measured data. 

% This work addresses optimizing performance in heterogeneous and uncertain networking environments. 

\redtext{The rest of this introductory chapter is organized as follows. First, we explain the challenges faced when coming to research uncertain  environments. Next, we explain how our research addresses these challenges. Further on, we discuss the sources of heterogeneity and uncertainty in networking environments. Finally, we overview our contribution, focusing on three concrete networking environments: (i) packet buffers, (ii) network caching, and (iii) placement of virtual machines in the cloud.}
\footnote{Advanced algorithms in heterogeneous and uncertain networking environments where studied also in the M.Sc. work~\cite{NoCs_jrnal}.}
\redtext{\section{The challenge of Researching Uncertain  Environments}\label{sec:uncertain_challenge}}
\redtext{Uncertainty appears in all layers of network design, including the {\em time} required to complete a task, such as processing a data packet, or computing a network function; the amount of currently available {\em resources} (e.g., CPU, and memory); and the expected {\em profit} gained upon successfully completing a task.}

\redtext{Introducing uncertainty to networking research sets new challenges. For instance, it is necessary to formulate new system models which well capture the possibility that some of the system's behaviour is unknown and to design algorithms which operate under uncertainty. However, developing such algorithms is very hard, and it is harder still to provide performance guarantees for such algorithms. For instance, two common approaches in network design are to plan either for the average-case, or for the worst-case. However, in an uncertain environment, the average-case is typically unknown, and hence cannot be accurately evaluated. On the other hand, planning for the worst-case may be a too pessimistic approach, 
because the worst-case may occur only very rarely, if at all. Other common approaches to evaluate suggested algorithms is to use simulations, or a real-world implementation. However, in the face of uncertainty, it is especially important to carefully pick input traces that well reflect a wide range of possible  real-world behaviours. For instance, the widely-used Poisson process for input generation may be insufficient  for imitating the behaviour of highly dynamic and uncertain systems.}

\redtext{In the next section we detail our research methodology for tackling these challenges.}

\section{Research Methodology}\label{sec:rsrch_methods}
We now describe the research methodology used in our research, starting from a system model, through rigorous mathematical analysis, and ending with the evaluation of the proposed solutions.

\subsection{System Model and Optimization Criteria}
The most fundamental step when coming to study a new problem is the development of a mathematical model, which should be both well established by the needs and possibilities of real-world network equipment, and mathematically coherent and consistent. The model should also well define the optimization criteria.

Our focus on heterogeneous and uncertain environments translates to using models with the following characteristics:
\begin{inparaenum}[(i)]
\newline\item \emph{The cost of measuring.} Exploring the inputs in an uncertain environment incurs high overhead in terms of processing power, bandwidth, etc. Hence, the model should associate some cost for each analysis or measurement of the coming traffic. 
\newline\item \emph{Budget constraint, or penalty function.} The main goal of a networking algorithm is commonly to maximize throughput, captured by the total gain from completed \redtext{tasks} per time unit. However, the goal of maximizing throughput typically conflicts with other design goals, such as minimizing latency or communication overhead. This conflict may be modelled by either a \emph{penalty function}, which the algorithm ``pays'' upon consumption of a resource (e.g., a cost paid for every message sent); or by forcing a \emph{budget constraint} on the total resources an algorithm may use.
% \newline\item \emph{Strong timing demands} An online algorithm must take decisions efficiently. Indeed, the term ``efficiently'' may be interpreted in versatile ways.
\newline\item \emph{Unbounded inputs.} By its definition, in an uncertain environment it is hard to predict the future inputs, or even to assume it follows a certain pattern, based on the current input - because the current input is also not fully known. Hence, our models aim at minimizing the assumptions taken with respect to the characteristics of incoming traffic.
% \redtext{Consequently, the algorithms which solve the problems should be scalable, 
% namely, to apply for every size of the inputs, amount of available resources etc.}
\newline\item \redtext{\emph{Multi-dimensional Heterogeneity.} Traditionally, it was common to assume that the inputs for the problem differ from each other in only one aspect (e.g., all tasks have uniform run-time, but some tasks are more valuable). Using such a single-dimensional heterogeneity assumption, it is easy to order, or prioritize, the inputs. However, this assumption is not valid anymore in nowadays networking environments, where servicing an incoming request requires multiple  resources (e.g., CPU, memory and storage), and incurs multi-dimensional costs (e.g., latency, bandwidth, and energy). Our models address the most general cases, of multi-dimensional heterogeneity.}
\end{inparaenum}

\subsection{Competitive Analysis}\label{sec:compeititve_analysis}
We evaluate the performance of online algorithms by means of competitive analysis \cite{Amortized, El-Yaniv}. The main advantage of competitive analysis is that it provides \emph{worst-case guarantees} on the performance of the online algorithm, while making no assumptions with regard to the input data, which is frequently unpredictable in heterogeneous and uncertain environments. 

An algorithm \alg\ is said to be $c$-competitive for a maximization objective if for every finite input sequence $\sigma$, the performance
of {\em any} algorithm for this sequence is at most $c$ times the performance of \alg\ ($c \geq 1$). As a consequence, proving an {\em upper
bound} of $c_u$ on the competitive ratio of a specific online algorithm \alg\ guarantees, that for every finite input
sequence $\sigma$, the performance of \alg\ for this sequence is at least $\frac{1}{c_u}$ times the performance
achievable by an optimal offline algorithm, \opt. Proving a {\em lower bound} of $c_{\ell}$ on the competitive ratio of
online algorithms means that no online algorithm can guarantee to gain in every input sequence $\sigma$ a performance
which is more than $\frac{1}{c_{\ell}}$ times of the performance which \opt\ can obtain for the same input sequence $\sigma$.
% \redtext{Note that in an uncertain environment}

\subsection{Approximation Algorithms}\label{sec:approx_algs}
Many fundamental optimization problems are NP-hard, that is, cannot be solved by current known tools in time which is polynomial in the input size. A common approach to tackle this limitation is to relax the requirements of finding an optimal solution and instead focus on finding a solution that is ``good enough'' in reasonable time.
This is the approach taken by \textit{approximation algorithms}~\cite{williamson11design}.
% which trade-off run-time for performance

The formal definition used to qualify approximation algorithms is similar to the one used in competitive analysis. We say that an algorithm \alg\ is a $c$-approximation algorithm for a maximization objective if (i) \alg\ runs in  polynomial-time, and (ii) for every finite input sequence $\sigma$, the performance of \alg\ for this sequence is at least $\frac{1}{c}$ times the performance
achievable by an optimal algorithm, \opt.

\redtext{In uncertain environments, providing such performance guarantees for \alg\ is especially hard, because not only that \alg\ doesn't know the inputs a-priory, but also \alg\ may not known the inputs {\em a-posteriori}. For instance, not only that \alg\ does not know which packets will arrive in the next cycle, but \alg\ may not know the characteristics of the packets which already arrived in the previous cycle.}

\subsection{Simulations}
Competitive analysis and approximation algorithms are very useful for showing {\em worst-case guarantees}. However, the performance of
algorithms on real traffic may be far superior to that of the worst case bounds. It is therefore instructive to
conduct a simulation study, which would compare the performance of different algorithms in more realistic scenarios. In
addition, a simulation study can provide further insights as to the algorithmic concepts and shed light on the effect
that various parameters have on the problem, well beyond the insights arising from a worst-case mathematical analysis.
Simulation traces may be generated either synthetically or using real-world data as follows. 

\paragraph*{Synthetic inputs.}
A common method to obtain input sequences for simulations is to generate an input trace \redtext{that} follows some known random distribution. 
In particular, arrival times are frequently modelled as a Poisson process. However, for better imitating an uncertain environment, and in particular bursty traffic, we sometimes use a {\em Modulated Markov Poisson Process} (MMPP)\cite{cite_MMPP}. Using such a stochastic process for traffic generation, the input occasionally (and randomly) switches between two states: HIGH -- which imitates bursty periods, with small expected inter-arrival time;  and LOW - which imitates relaxed periods, with large expected inter-arrival time. Thus, the LOW periods allow the system to drain the load generated during the bursty period. An example of an MMPP used in our work can be found in Section~\ref{sec:simulation_set}.
\paragraph*{Real-world data.}
Obviously, synthetic traces do not accurately reflect the characteristic of real-world input sequences. It is therefore
instructive to generate input sequences using real-world data. This data may include the arrival times of the inputs, the characteristics of incoming requests (e.g., required processing time), and the properties of the underlying system - e.g., the topology and capacities of nodes and links. 

Concrete examples of real-world traces used in this work are Wikibench~\cite{WikiBench}, a trace of real accesses to Wikipedia (see Section~\ref{sec:sim}); and datasets recorded in cloud services by Google and Amazon (see Section~\ref{sec:apsr:PlacementStudy}). 

\section{Sources of Heterogeneity and Uncertainty}\label{sec:srcs_of_htro_n_uncertain}
We now describe the sources of heterogeneity and uncertainty in each of the three networking environments addressed in this work.
Finally, we introduce our work and overview our contribution in three concrete networking environments: (i) packet buffers, (ii) network caching, and (iii) placement of virtual machines in the cloud.
\subsection{Scheduling and Management of Packet Buffers}
Some of the most basic tasks in computer networks involve scheduling and managing packet queues equipped with finite
buffers. Processing data packets in modern networking equipment spans multiple tasks, including various forms of DPI 
(Deep Packet Inspection), 
MPLS~\cite{MPLS}
%(Multi-Protocol Label Switching)
and VLAN~\cite{VLAN}
%(Virtual Local Area Network) 
tagging, encryption / decryption, compression / decompression, and more. This translates to increased heterogeneity in the required processing time and QoS (Quality of Service) of the arriving packets.

Traditionally, the research studying scheduling and management of packet queues assumed that the various properties of any packet –- e.g., its QoS characteristic, its required processing, its deadline -– are known upon its arrival~\cite{Goldwasser}. 
However, this assumption is in many cases unrealistic.
For instance, when a packet is recursively encapsulated a few times by MPLS, VLAN, or IPSec (secured Internet Protocol), it is hard to determine in advance the total number of processing cycles that such a packet would require~\cite{Folded, Kangaroo}.
Furthermore, the QoS features of a packet are commonly determined by its flow ID, which is in many cases known only after parsing~\cite{Kangaroo}.

In data center network architectures such as PortLand~\cite{PortLand}, ingress switches query a cache for an
application-to-location address resolution. A cache miss, which is unpredictable by nature, results in forwarding of
the packet to the switch software or to a central controller, which performs a few additional processing cycles
before the packet can be transmitted.
Similarly, in the realm of Software Defined Networks, ingress switches query a cache for obtaining rules for a
packet~\cite{DIFANE}, which may also depend on priorities~\cite{ETHANE}. In such a case, a cache miss results in
additional processing until the rules are retrieved and the profit from the packet is known.

In Chapter~\ref{sec:buf} we address the problem of maximizing throughput in packet buffers where the characteristics of some arriving traffic are unknown upon arrival. 

\subsection{Access Strategies in Network Caching}
Network caching is a prominent networking primitive, used for efficiently distributing data in numerous environments, such as Content Delivery Networks (CDNs)~\cite{CDN_OceanStore,CDN_AdaptSize}, 5G in-network caching~\cite{5GNetworks,5GNetworks2}, and wide area networks~\cite{summary_cache}. 

Typically, the provider distributes the data to multiple \emph{data stores}, thus bringing it closer to the end users, who can now obtain the data faster, and with lower bandwidth and energy consumption. The user can therefore look for a requested datum in multiple locations, where accessing each datastore incurs some cost, in terms of latency, bandwidth and energy. If the user fails to find the requested 
datum in any of the datastores he accesses, he pays a high \emph{miss penalty}, e.g. for retrieving the datum from a remote site.

For helping the user make a decision with regard to which datastores to access, each datastore commonly sends a periodic update, or \emph{indicator}, which is a compact summary of the list of items in the datastore. The most known indicator is the Bloom filter and its numerous variants~\cite{Bloom, Survey12, Survey18}, but there exist other indicators, such as the TinyTable~\cite{TinyTable}. 

Due to space and bandwidth constraints, the indicator is not an exact list of cached items, but merely a compact summary of it, and may therefore experience false replies. 
That is, an indicator may indicate that a given item is held by a certain datastore while it is actually not there, and vice versa. As a result, a user who is looking for datum $x$ for which there exist multiple positive indications 
experiences a hard problem of decision making under uncertainty: 
how to minimize the aggregate access cost to datastores, while keeping a low probability of a miss.

In Chapter~\ref{sec:BF} we address the problem of developing access strategies for multiple datastores in an heterogeneous and uncertain environment. 

\subsection{Placement of Virtual Machine in the Cloud}
The cloud computing paradigm provides computing and networking services by running virtual machines (VMs) in the cloud, without relying on local physical  machines~\cite{Middleboxes,EASE}.
One of the key tasks performed by a cloud service provider is the placement of new virtual machines in computing nodes (hosts). Although this is a fundamental task, performing it efficiently is challenging, as different VMs have heterogeneous requirements for computing resources (e.g. CPU, memory and storage). To place a new VM, one should find a host with enough available resources to meet its concrete resource requirements. However, obtaining a fresh status of the available resources in every host in such a heterogeneous and dynamic environment incurs a high communication overhead and complex synchronization mechanisms.

In Chapter~\ref{sec:APSR} we study the problem of placing virtual machines in  large clouds, where maintaining an always-fresh full system's state in impractical, resulting in a highly uncertain environment.

\section{Overview of the Thesis}
This thesis studies fundamental problems in several networking environments, which are characterized by high heterogeneity and uncertainty. 
In what follows we overview our contribution in each of these environments.

\paragraph*{Buffer management and Scheduling.} In Chapter~\ref{sec:buf} we introduce the problem of buffer management and scheduling in an heterogeneous and uncertain environment. \redtext{We propose a new model, in which the characteristics of some of the incoming packets are unknown upon arrival.}
We present lower bounds on the performance of any randomized algorithm for the problem. In our proofs we use a novel technique, in which we bound the expected number of packets in the buffer of an optimal offline algorithm by means of a Markov process;
we believe that this technique may be of independent interest. Next, we describe several algorithmic concepts tailored for the problem, and develop an algorithm that applies these algorithmic concepts and prove upper-bound on the competitive ratio of our algorithm. We further validate and evaluate the performance of our proposed algorithms via an extensive simulation study. Our results highlight the effect the various parameters have on the problem, well beyond
the insights arising from our rigorous mathematical analysis.
\redtext{Our work extends the rich literature about scheduling and buffer management by tackling the challenging
problem 
of  packets whose characteristics are unknown upon arrival, which is the common case in modern networking environments.}
These results also appear in~\cite{MIST, MIST_journal}.

\paragraph*{Access Strategies in Network Caching.}
In Chapter~\ref{sec:BF} we study the problem of developing access strategies for multiple datastores with uncertainty. We formally model this problem in very general and heterogeneous settings with varying access costs and per-datastore hit ratios. Next, we show that previously suggested strategies are too simplistic and implicitly rely on specific assumptions about the workload, or the underlying system.
Thus, in general, an access strategy that works well in one scenario may be inefficient for another.

Further on, we develop and analyze the performance of several practical approximation algorithms that work in polynomial time.
Through an extensive evaluation with varying system parameters, we show that our algorithms are more stable than existing approaches. That is, they outperform or achieve very similar access costs to the best competitor for any tested system configuration.
\redtext{Our results show that access strategies that take into account the level of uncertainty may improve upon the traditional approach for the datastore selection problem.}
Some of these results also appear in~\cite{Accs_Strategies_Infocom, Accs_Strategies_Journal, CAB, FN_aware}.

\paragraph*{Virtual Machine Placement in Virtual Network Functions.}
In Chapter~\ref{sec:APSR} we study the problem of placing new VMs for virtual network functions in an uncertain heterogeneous environment. We  explore the impact of parallelism on the ratio of failed placements (\emph{decline ratio}) in various popular placement algorithms. We show that employing multiple parallel deterministic schedulers may result in a very high communication overhead, and a high decline ratio. We therefore suggest using  independent \emph{random} schedulers, where each scheduler samples but a few hosts, as a solution which potentially allows for high throughput, while keeping low decline ratio, and low communication overhead.

Based on our suggested concept we introduce our proposed algorithm, \emph{Adaptive Partial State Random} (\algtop), that dynamically adjusts the number of parallel schedulers according to the system's utilization, and incorporates randomness into its decision making.
We formally analyze the performance of \algtop\ and provide guarantees as to its communication overhead and expected decline ratio. We also evaluate the performance of \algtop\ using three real-life datasets and show that it enables a high degree of parallelism (e.g., effectively running 20-100 schedulers) in a variety of realistic scenarios. We further show that \algtop\ reduces the communication overhead by over 85\% compared to state of the art algorithms. 
\redtext{These results highlight the promise of our approach, of using multiple independent lightweight random schedulers, as a scalable and highly reactive solution, which provides good performance and low communication overhead in various various configurations and scenarios.}
% \redtext{Our work positions using multiple independent lightweight random schedulers as an attractive alternative to the traditional approach of using complex central scheduling and synchronization mechanisms. Furthermore, our approach is highly scalable by adapting the number of schedulers and the sample size to the amount of available resources.}
Part of our results also appear in~\cite{APSR_infocom_poster, APSR_IFIP}.

\newpage
\chapter{Buffer Management and Scheduling}\label{sec:buf}

\section{Problem Overview}
\label{sec:buf:introduction}
Some of the most fundamental tasks in computer networks is the scheduling and management of packet buffers, where the primary goal in such settings is maximizing the throughput of the system. As explained in Section~\ref{sec:srcs_of_htro_n_uncertain}, in modern networks the various properties of a packet -- e.g., its QoS characteristic, its required processing, its deadline -- are typically unknown upon arrival. These characteristics usually become known once a packet undergoes some initial processing, or \emph{parsing}. However, for traffic corresponding to the same flow, it is common for characteristics to be unknown when the first few packets of the flow arrive at a network element, and once these properties are unraveled, they become known for all subsequent packets of this flow.
It therefore follows that only part of the arriving packets has unknown characteristics upon arrival, which become known after parsing.

In this chapter we address such scenarios where the characteristics of some arriving traffic are unknown upon arrival,
and are only revealed when a packet has undergone some initial processing
(parsing), ``causing the mist to clear''.
We model and analyze the performance of algorithms in such settings, and in particular we develop online scheduling and buffer management algorithms for the problem of maximizing the profit obtained from delivered packets, and
provide guarantees on their expected performance using competitive analysis.\footnote{Please refer to the definition of competitive analysis in Section~\ref{sec:compeititve_analysis}.}

We focus on the general case of heterogeneous processing requirements (work) and heterogeneous
profits~\cite{MultiV_MultiW}.
We assume priority queueing, where the exact priorities depend on the specifics of the model studied.
We present both algorithms and lower bounds for the problem of dealing with unknown characteristics in these models.
Furthermore, we highlight some design concepts for settings where algorithms have limited knowledge, which we
believe might apply to additional scenarios as well.

\begin{figure}[t]
\centering
\includegraphics[width=0.8\columnwidth]{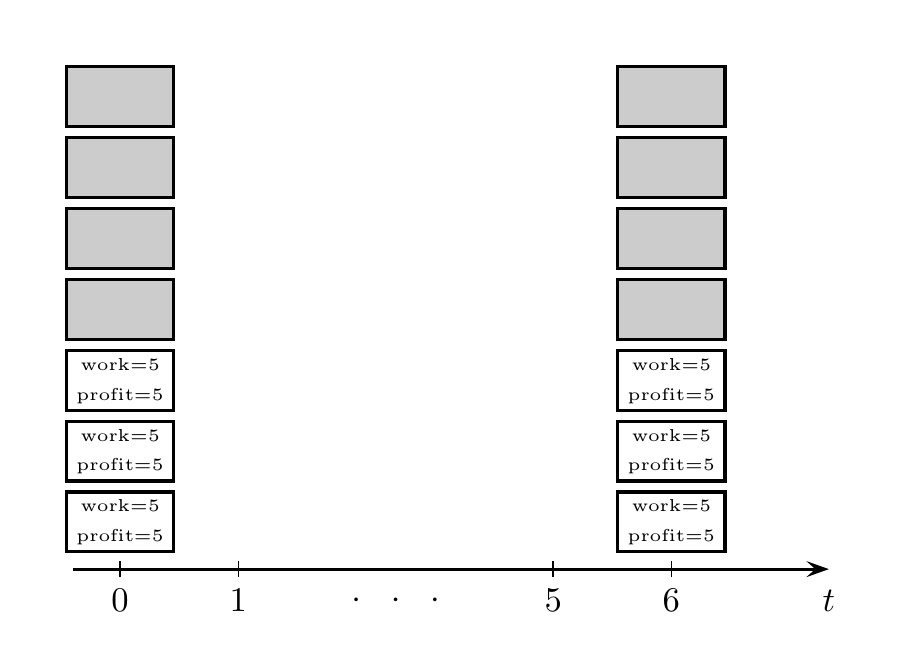}
\caption[An example of an arrival sequence with known and unknown packets]{An illustrative example of an arrival sequence with known and unknown packets}
\label{fig:toy_example}
\end{figure}

As an illustration of the problem, assume we have a 3-slots buffer, equipped with a single processor, and consider
the arrival sequence depicted in Fig.~\ref{fig:toy_example}.
In the first cycle, we have seven unit-size packets arriving, out of which three will provide a profit of 5 upon successful delivery, each requiring 5 processing cycles (work).
The characteristics of these three packets are known immediately upon arrival.
The characteristics of the remaining four packets (marked gray) are \emph{unknown} upon arrival. We therefore dub
such packets $U$-packets (i.e., unknown packets). Each of these four $U$-packets may turn out to be either a "best"
packet, requiring minimal work and having maximal profit; a "worst" packet, requiring maximal work and having
minimal profit; or anything in between. Thus, already at the very beginning of this simple scenario, any buffering
algorithm would encounter an \emph{admission control} dilemma: how many $U$-packets to accept, if any?
This dilemma can be addressed by various approaches including, e.g., allocating some buffer space for $U$-packets,
accepting $U$-packets only when current known packets in the buffer are of poor characteristics, in terms of profit,
or of profit to work ratio, etc. In case that the algorithm accepts $U$-packets, an additional question arises:
which of the $U$-packets to accept into the buffer? Obviously, for any online deterministic algorithm, there exists a
simple adversarial scenario, which would cause it to accept only the "worst" $U$-packets (namely, packets with
maximal work and minimum profit), while an optimal offline algorithm would accept the best packets. This motivates
our decision to focus our attention on randomized algorithms.

We now turn to consider another aspect of handling traffic with some unknown characteristics. Assume the scenario continues with 5 cycles without any arrival, and then a cycle with an identical arrival pattern - namely, three known packets with both work and profit of 5 per packet, and four $U$-packets. This
sheds light on a \emph{scheduling} dilemma: which of the accepted packets should better be processed first? every
scheduling policy impacts the buffer space available in the next burst. For instance, a run-to-completion attitude
would enable finishing the processing of one known packet by the next burst, thus allowing space for accepting a new
packet without preemption. However, one may consider an opposite attitude - namely, parsing as many $U$-packets as
possible, thus "causing the mist to clear",
allowing more educated decisions, once there are new arrivals. In terms
of priority queuing, this means over-prioritizing some
 $U$-packets, and allowing them to be parsed immediately upon arrival.
We further develop appropriate algorithmic concepts based on the insights from this illustrative example in Section~\ref{sec:algorithmic_conecpts}.

\subsection{System Model}\label{sec:buf:model}

We now describe our system model. For ease of reference, all the notations used in this chapter are summarized in Table~\ref{tbl:buf:notations}.
Our system model consists of four main modules, namely, \begin{inparaenum}[(a)]
\item an input queue equipped with a finite buffer,
\item a buffer management module which performs admission control, 
\item a scheduler module which decides which of the pending packets should be processed, and
\item a processing element (PE), which performs the processing of a packet.
\end{inparaenum}

We divide time into discrete cycles,
where each cycle represents a fixed time slot,
 and consists of three steps:
\begin{inparaenum}[(i)]
\item The {\em transmission} step, in which fully-processed packets leave the queue,
\item the {\em arrival} step, in which new packets may arrive, and the buffer management module decides which of them should be retained in the queue, and which of the currently buffered packets should be pushed-out and dropped, and finally
\item the {\em processing} step, in which the scheduler assigns a single packet for processing by the PE, which in turn processes the packet.
\end{inparaenum}

We consider a sequence of unit-size packets arriving at the queue. Upon its arrival, the characteristic of each packet may be \emph{known} - in which case we refer to the packet as a {\em $K$-packet} (i.e., known packet); or \emph{unknown}  - in which case we refer to the packet as a {\em $U$-packet} (i.e., unknown packets). We let $M$ denote the maximum number of $U$-packets that may arrive in any single cycle.
We focus our attention on the case where $M>0$, unless specifically stated otherwise.

Each arriving packet $p$ has some
\begin{inparaenum} [(1)]
\item intrinsic benefit (\emph {profit})
$v(p) \in \set{1,\dots,V}$, and
\item required number of processing cycles (\emph {work}),
$w(p) \in \set{\minWork, \minWork+1, \dots ,W}$. Unless explicitly stated otherwise, we consider the most general case, namely, $\minWork = 1$.
\end{inparaenum}

To simplify the expressions throughout the paper, we assume that both $V$ and $W$ are powers of 2.\footnote{Our results degrade by a mere constant factor otherwise.}
We use the notation {\em $(w,v)$-packet} to denote a packet with work $w$ and profit $v$.
We note that the {\em uniform} case where all packets require the same amount of work, and all packets have the same profit, is trivial, since the simple run-to-completion policy is optimal. We therefore focus our attention on non-uniform traffic.

In our model, similarly to~\cite{Shpiner}, upon processing a $U$-packet for the first time, its properties become known. We therefore refer to such a first processing cycle of a $U$-packet as a {\em parsing cycle}. Non-parsing cycles where the processor is not idle are referred to as {\em work cycles}.

The queue buffer can contain at most $B$ packets.
We assume $B \geq 2$, since the case where $B = 1$ is degenerate. The {\em head-of-line} (HoL) packet at time $t$
(for a given algorithm \alg) is the highest priority packet stored in the buffer just prior to the processing step
of cycle $t$, namely, the packet to be scheduled for processing in the processing step of $t$.
We say the buffer is {\em empty} at cycle $t$ if there are no packets in the buffer after the transmission step of cycle $t$.

We study {\em queue management} algorithms, which are responsible for both the buffer management
and the scheduling of packets for processing. In particular, we focus our attention on algorithms targeted at
maximizing the {\em throughput} of the queue, i.e. the overall profit from all packets successfully transmitted out
of the queue.
The throughput of algorithm \alg\ is denoted by $TP(\malg)$. We use the terms throughput and performance interchangeably.

An algorithm is said to be {\em greedy} if it accepts packets as long as there is available buffer space. We further
focus our attention on {\em work-conserving}
algorithms, i.e., algorithms which never leave the PE idle unnecessarily.
We evaluate the performance of online algorithms using competitive analysis, as detailed in Section~\ref{sec:compeititve_analysis}.

\LTcapwidth=0.95\textwidth
\begin{longtable}{| p{.15\textwidth} | p{.68\textwidth} | p{.06\textwidth} |}%{|l|l|l|}
\caption[List of symbols used in Chapter~\ref{sec:buf}]{List of symbols used in this chapter. The rightmost colon details the section where the symbol is used. A blank entry means that the symbol is used throughout the chapter.}
\label{tbl:buf:notations} \\

\hline \multicolumn{1}{| p{.15\textwidth} |}{\textbf{Symbol}} & \multicolumn{1}{| p{.68\textwidth} |}{\textbf{Meaning}} & \multicolumn{1}{| p{.06\textwidth} |}{\textbf{Section}} \\ \hline 
\endfirsthead

\multicolumn{3}{c}%
{{\bfseries \tablename\ \thetable{} -- continued from previous page}} \\
\hline \multicolumn{1}{|c|}{\textbf{Symbol}} & \multicolumn{1}{c|}{\textbf{Meaning}} & \multicolumn{1}{c|}{\textbf{Section}} \\ \hline 
\endhead

\hline \multicolumn{3}{|r|}{{Continued on next page}} \\ \hline
\endfoot

\hline \hline
\endlastfoot
        \hline
    	\hline
    	$K$-packet & Packet whose characteristics are known upon arrival&\tabularnewline
    	\hline
    	$U$-packet & Packet whose characteristics are unknown upon arrival&\tabularnewline
    	\hline
    	$M$ & maximum number of $U$-packets that may arrive in any single cycle &\tabularnewline
    	\hline
    	$w(p)$ & work of packet $p$  &\tabularnewline
    	\hline
    	$v(p)$ & profit of packet $p$  &\tabularnewline
    	\hline
    	$(w,v)$-packet & packet with work $w$ and profit $v$ &\tabularnewline
    	\hline
    	$W_0$ & minimal work of a packet&\tabularnewline
    	\hline
    	$W$ & maximal work of a packet&\tabularnewline
    	\hline
    	$V$ & maximal profit of a packet&\tabularnewline
    	\hline
    	$B$ & \redtext{Size of the buffer (maximal number of packets)}&\tabularnewline
    	\hline
    	$TP(\malg)$ & throughput (performance) of algorithm \alg\ &\tabularnewline
    	\hline
    	$N$ & number of cycles during the fill phase & \ref{sec:lower_bounds}\tabularnewline
    	\hline
    	$r$ & probability that an algorithm parses new packets during a given cycle &\tabularnewline
    	\hline
    	$q_j$ & the state where there are $j$ $\mbest$-packets in \subopt's buffer at the beginning of an iteration & \ref{sec:lower_bounds}\tabularnewline
    	\hline
    	$\malpha_k$ & the probability of having exactly $k$ $\mbest$-packets arriving during one iteration & \ref{sec:lower_bounds}\tabularnewline
    	\hline
    	\rule{0pt}{3ex}    
    	$C_i^{(W)}$ & work-class&\ref{sec:algorithms},~\ref{sec:improved_algorithms}\tabularnewline
    	\hline
    	\rule{0pt}{3ex}    
    	$C_j^{(P)}$ & profit-class&\ref{sec:algorithms},~\ref{sec:improved_algorithms}\tabularnewline
    	\hline
    	\rule{0pt}{4ex}    
    	$C_{(i,j)}$ & combined-class of a packet belonging to work-class $C_i^{(W)}$ and profit-class $C_j^{(P)}$&\ref{sec:algorithms},~\ref{sec:improved_algorithms}\tabularnewline
    	\hline
    	$\delta_W$ &  maximal ratio between the work values of two packets belonging to the same work-class&\ref{sec:algorithms},~\ref{sec:improved_algorithms}\tabularnewline
    	\hline
    	$\delta_V$ &  maximal ratio between the profit values of two packets belonging to the same profit-class&\ref{sec:algorithms},~\ref{sec:improved_algorithms}\tabularnewline
    	\hline
    	$G$ & the selected class&\ref{sec:algorithms},~\ref{sec:improved_algorithms}\tabularnewline
    	\hline
    	$G$-packets & packets which belong to the selected class, $G$&\ref{sec:algorithms},~\ref{sec:improved_algorithms}\tabularnewline
    	\hline
    	$\ell_W$ & number of work-classes&\ref{sec:algorithms},~\ref{sec:improved_algorithms}\tabularnewline
    	\hline
    	$\ell_V$ & number of profit-classes&\ref{sec:algorithms},~\ref{sec:improved_algorithms}\tabularnewline
    	\hline
    	\rule{0pt}{4ex}    
    	$\Au_{(w,v)}(t)$ & number of $U$-packets with work $w$ and profit $v$, which arrive in cycle $t$&\ref{sec:algorithms},~\ref{sec:improved_algorithms}\tabularnewline
    	\hline
    	\rule{0pt}{4ex}    
    	$P\mFILL$ & sets of cycles in which \SAM\ is in the fill phase&\ref{sec:algorithms},~\ref{sec:improved_algorithms}\tabularnewline
    	\hline
    	\rule{0pt}{4ex}    
    	$P\mFLUSH$ & sets of cycles in which \SAM\ is in the flush phase&\ref{sec:algorithms},~\ref{sec:improved_algorithms}\tabularnewline
    	\hline		$S_{\alpha}(t)$ & expected profit of \SAM\ from $\alpha$-packets which arrive during cycle $t$&\ref{sec:algorithms},~\ref{sec:improved_algorithms}\tabularnewline
    	\hline
    	$S_{\alpha}, O_{\alpha}$ & overall expected profit of \SAM, \opt\ from $\alpha$-packets&\ref{sec:algorithms},~\ref{sec:improved_algorithms}\tabularnewline
    	\hline
    	\rule{0pt}{4ex}    
    	$O\mFILL_{\CsU}$ & expected profit of \opt\ from $\CsU$-packets which arrive during a $P\mFILL$ period&\ref{sec:algorithms},~\ref{sec:improved_algorithms}\tabularnewline
    	\hline
    	\rule{0pt}{4ex}    
    	$O\mFLUSH_{\CsU}$ & expected profit of \opt\ from $\CsU$-packets which arrive during a $P\mFLUSH$ period&\ref{sec:algorithms},~\ref{sec:improved_algorithms}\tabularnewline
    	\hline		
    % 	\rule{0pt}{4ex}    
        $C^*_{(i,j)}$ & the $(i,j)$-closure class: $$C^*_{(i,j)}=\bigcup_{i'\leq i, j'\geq j} C_{(i',j')}$$
        &\ref{sec:improved_algorithms}\tabularnewline
    	\hline
\end{longtable}

\subsection{Related Work}
Competitive algorithms for scheduling and management of bounded buffers have been extensively studied for the past two decades. The problem was first introduced in the context of differentiated services, where packets have uniform size and processing requirements, but some of the packets have higher priorities,
represented by a higher profit associated with them~\cite{Rosen, OF, Smoothing}. The numerous variants of this problem include models where packets have deadlines or maximum lifetime in the switch~\cite{OF}, environments involving multi-queues~\cite{Albers2005, Azar&Richter2004, Buf_Xbar_Kogan, Buf_Xbar_Kanizo} and cases with packets dependencies~\cite{Dependencies, OF_sizes}, to name but a few.
An extensive survey of these models and their analysis can be found in~\cite{Goldwasser}.

While traditionally it was assumed that packets have heterogeneous profits but uniform work (processing
requirements), some recent work introduced the complementary problem, of uniform profits with heterogeneous
work~\cite{Multipass}. This work presented an optimal algorithm for the fundamental problem, as well as
online algorithms and bounds on the competitive ratio for numerous variants. Subsequent research investigated related
problems with heterogeneous work combined with heterogeneous packet sizes~\cite{Sizes}, or with heterogeneous
profits~\cite{MultiV_MultiW, Azar2015}.
In particular,~\cite{MultiV_MultiW} showed that the competitive ratio of some straight-forward deterministic algorithms for the problem of heterogeneous work combined with heterogeneous profits is linear in either the maximal work $W$, or in the maximal profit $V$, even when the characteristics of all packets are known upon arrival. These results motivate our focus on randomized algorithms.

While most of the literature above assumed that all the characteristics of packets are known upon arrival, this
assumption was put in question recently~\cite{Shpiner} by noting that it is often invalid.
However, the main problem addressed in~\cite{Shpiner} revolved around developing schemes for transmitting packets of the same flow in-order, while our work focuses on maximizing throughput with limited buffering resources, and designing both buffer management and scheduling policies targeted at this objective.

Maybe closest to our work are the recent studies considering serving in the dark~\cite{Dark,Dark2}, which investigate an extreme case where the online algorithm learns the profit from a packet only after transmitting it. These studies
consider highly oblivious algorithms, whereas our model and our proposed algorithms dwell in a middle-ground between
the well studied models with complete information, and these recent oblivious settings. Our work further considers
traffic with variable processing requirements, whereas~\cite{Dark,Dark2} focus on settings where all packets require
only a single processing cycle, and they differ only by their profit.

The problem of optimal buffering of packets with variable work is closely related to the problem of job scheduling
in a multi-threaded processor, which was extensively studied. A comprehensive survey of online algorithms for this
problem can be found in~\cite{Job}. This body of work, however, differs significantly from our currently studied
model.
The major differences are that packet buffering has to deal with limited buffering capabilities, and is targeted at
maximizing throughput. Processor job scheduling, however, usually has no strict buffering limitations, and is mostly
concerned with minimizing the response time.

\subsection{Our Contribution}
We introduce the problem of buffering and scheduling which aims to maximize throughput where the characteristics of
some of the packets are unknown upon arrival. We focus our attention on traffic where every packet has some required
processing cycles, and some profit associated with successfully transmitting it.

In Section~\ref{sec:lower_bounds} we present lower bounds on the performance of any randomized algorithm for the
problem. Specifically, we show that no algorithm can have a competitive ratio better than $\Omega(\min\set{WV,M})$, even against an adversary which can accommodate merely 2 packets in its buffer, where $W$ and $V$ denote the maximum work and profit of a packet, respectively, and $M$ represents the maximum number of unknown packets which may arrive in any single cycle.
We also prove stronger lower bounds for the general settings using a novel technique, in which we bound the expected number of packets in the buffer of an optimal offline algorithm by means of a Markov process.

In Section~\ref{sec:algorithmic_conecpts} we describe several algorithmic concepts tailored for dealing with unknown characteristics in such systems.
We follow by presenting an algorithm that applies our suggested algorithmic concepts in Section~\ref{sec:algorithms}.
For the most general case, we prove our algorithm has a competitive ratio of $O(M \log V \log W)$. We further show how to improve this bound in several important special cases.

In Sections~\ref{sec:improved_algorithms}-\ref{sec:Practical_Implementation}
we present some modifications and heuristics applicable to our algorithm that, while leaving the worst-case guarantees intact, are designed to improve performance compared to the baseline algorithmic design.
The modified algorithm can cope with cases where neither the maximal amount of work and profit, nor the maximum number of unknown packets per cycle, are known in advance.

We further validate and evaluate the performance of our proposed algorithms in Section~\ref{sec:simulations} via an
extensive simulation study. Our results highlight the effect the various parameters have on the problem, well beyond
the insights arising from our rigorous mathematical analysis.

We conclude in Section~\ref{sec:buf:conclusions} with a discussion of our results, and also highlight several
interesting open questions.

\section{Lower Bounds}
\label{sec:lower_bounds}
In this section, we present lower bounds on the competitive ratio of any randomized algorithm for our problem.

These lower bounds serve two main objectives:
\begin{inparaenum}[(i)]
\item They represent the best competitive ratio which one can hope to achieve; and
\item the hard scenarios used in the proofs of these lower bounds highlight the challenges which any competitive online algorithm would have to tackle.
\end{inparaenum}

\subsection{Highly-restricted Adversaries}\label{sec:lower_bounds_restricted}
In this section, we prove lower bounds on the competitive ratio of any online algorithm for our problem, compared to a highly-restricted adversary which uses a buffer which can only store a single packet. This restriction on the amount of buffer space available for the adversary enables us to better highlight the scaling laws of the problem, depending on the various parameters.

\begin{theorem} \label{Theorem_WP_rand}
If $V \geq 1$, $M \geq 1$ and the work of each packet is $w(p) \in \set{\minWork, \minWork+1, \dots, W}$ where $W \geq 2$,
then the competitive ratio of any randomized algorithm  for non-uniform traffic is at least
$$
\frac{V (W-1)}{2\minWork} \left[1 - \left(1- \frac{1} {V(W-1)+1-\minWork}\right)^{M\minWork} \right],
$$

even against an optimal offline algorithm which has a buffer which can only store a single packet.
\end{theorem}

\begin{proof}
\redtext
{First note that since traffic is non uniform, we are guaranteed to have either $V>0$ and / or $W>\minWork$, and therefore $V(W-1) + 1 - \minWork \neq 0$.}

\redtext{We prove the theorem using Yao's principal~\cite{yao77probabilistic}. In a nutshell, Yao's principal claims that the expected performance of any randomized algorithm is at most the expected performance of any {\em deterministic} algorithm for a worst-case probability distribution of the input. 
Hence, we define a carefully crafted distribution over arrival sequences, and show a lower bound on the ratio between the expected performance of an optimal clairvoyant algorithm for the problem, and the expected performance of any  deterministic algorithm for the problem.
}

We will show that the claim is true even if the optimal offline algorithm uses a buffer that can hold only a single packet.
We define the following collection of arrival sequences, where each arrival sequence has two phases: a {\em Fill phase}, and a {\em Flush phase}.
The Fill phase consists of \emph{iterations} as follows. Each iteration begins with $\minWork$ cycles without arrivals; and continues with $\minWork$ cycles with $M$ $U$-packets arriving per cycle, where each packet is a $\mbest$-packet with probability $p$, and a $\mworst$-packet with probability $(1-p)$, for some constant $p$ to be determined later. The total number of cycles during the fill phase is $N$, where $N$ is a large integer, so we have $\frac{N}{2 \minWork}$ iterations.
Once the fill phase ends, it is followed by the Flush phase, which consists of $BW$ cycles without arrivals.
We note that due to the random choices of packets being either $\mbest$-packets or $\mworst$-packets, the above structure induces a distribution over a collection of possible arrival sequences.

To simplify our analysis, we define the \subopt\ policy, which works as follows. Within the fill phase, during each iteration, \subopt\ accepts at most one $\mbest$-packet which has arrived during the iteration, if such a packet exists. This packet is the one considered {\em picked} by \subopt\ in that iteration.
Starting from the second iteration, during the first $\minWork$ cycles of each iteration, \subopt\ processes the packet it picked during the previous iteration (if such a packet exists), and transmits it.
During the flush phase, \subopt\ processes and finally transmits the packet it picked during the last iteration. It should be noted that \subopt\ is neither greedy nor work conserving.
Moreover, the expected throughput of \subopt\ serves as a lower bound on the expected optimal throughput
possible.

\begin{table}
        \centering
    \begin{tabular}{|l|l|l|l|p{2.8 cm}|p{0.9 cm}|}%
        \hline
        Iteration & Cycle &Arrivals	&Operation	& Buffer &Total Gain \tabularnewline
        \hline
        \hline
        1 &	1 & -	&-	&-	&0 \tabularnewline
        & 2 &	-	&  	& -	&0 \tabularnewline
        & \dots &	-	& - &$\mbest$&	0 \tabularnewline
        &	$W_0$ &	- &	-	&$\mbest$&	0 \tabularnewline 
        &	$W_0$ + 1&	$\mbest$\tikzmark{pick1_begin}&	-	&$\mbest$\tikzmark{pick1_end}&	0 \tabularnewline 
        &	\dots &	$(W,1)$&	-	&$\mbest$&	0 \tabularnewline 
        &	2$W_0$	& $(W,1)$&	-	&$\mbest$\tikzmark{end_buf1}&	0 \tabularnewline  \hline

        2  &	2$W_0$+1&	- &	Process &$\mbest$\tikzmark{proc1_begin}&	0 \tabularnewline 
        &	\dots &	- &	Process	&$\mbest$&	0 \tabularnewline 
        &	3$W_0$ &	- &	Process	&$\mbest$\tikzmark{proc1_end}&	0 \tabularnewline 
        &	3$W_0$+1 &	$(W,1)$&	Transmit\tikzmark{xmt1} &	- &	$V$ \tabularnewline 
        &	\dots &	$\mbest$ \tikzmark{pick2_begin} &		&$\mbest$\tikzmark{pick2_end}&	$V$ \tabularnewline 
        &	\dots &	$\mbest$ &	-	&$\mbest$&	$V$ \tabularnewline 
        &	4$W_0$ &	$(W,1)$ &	-	&$\mbest$\tikzmark{end_buf2}&	$V$ \tabularnewline  \hline

        3 &	$4W_0+1$ &	- &	Process& $\mbest$\tikzmark{proc2_begin}&	$V$ \tabularnewline 
        &	\dots&	- &	Process &	$\mbest$ 	& $V$ \tabularnewline 
        &	$5W_0$ &	- &	Process &	$\mbest$\tikzmark{proc2_end}& $V$ \tabularnewline 
        &	$5W_0+1$ &	$(W,1)$ &	Transmit\tikzmark{xmt2} &	-	& $2  V$ \tabularnewline 
        &	\dots &	$(W,1)$ &	-	& -	& $2  V$ \tabularnewline 
        &	6$W_0$ &	$(W,1)$ &	- &	-	& $2  V$ \tabularnewline  \hline
        
        4 &	6$W_0$+1 &	- &	-	& -	& $2  V$ \tabularnewline 
        &	\dots &	- &	-	&	- & $2  V$ \tabularnewline 
        &	$7W_0$ &	- &	- &	 - & $2V$ \tabularnewline 
        &	$7W_0+1$ &	$(W,1)$ &	- &	-	& $2V$ \tabularnewline 
        &	\dots &	$(W,1)$ &	-	& -	& $2  V$ \tabularnewline 
        &	8$W_0$ & $\mbest$\tikzmark{pick3_begin}&		 & $\mbest$\tikzmark{pick3_end}&	$2  V$ \tabularnewline  \hline

\end{tabular}
    
    \newcommand{\yshift}{0.12 cm}
    \newcommand{\xshiftleft}{0.03 cm}
    \newcommand{\xshiftright}{-1.5 cm}
    \newcommand{\arrowlinewidth}{1.5 pt}
    \newcommand{\textabovearrow}{-0.18 cm}
    \newcommand{\bracesxshit}{0 cm}
    \newcommand{\bracesyshithi}{0.3 cm}
    \newcommand{\bracesyshitlo}{-0.1 cm}
    \newcommand{\bracestotxtx}{1.92 cm}
    \newcommand{\amplitude}{7 pt}
    
    \begin{tikzpicture}[
        overlay, remember picture, shorten >=-3pt,
            blockblue/.style={
              rectangle,
              draw=blue,
              thick,
              align=center,
              text width=1.15 cm,
              minimum height = 0.5 cm,
            },
            blockgreen/.style={
              rectangle,
              draw=green,
              thick,
              align=center,
              text width=1.15 cm,
              minimum height = 0.5 cm,
            },
            blockycyan/.style={
              rectangle,
              draw=cyan,
              thick,
              align=center,
              text width=1.15 cm,
              minimum height = 0.5 cm,
            },
        ]
    
        % Arrows showing picking a pkt
        \path[->] ([shift={(\xshiftleft,  \yshift)}]pic cs:pick1_begin) edge [] node [midway,above=\textabovearrow] 
        {\footnotesize{pick a packet}} ([shift={(\xshiftright,  \yshift)}]pic cs:pick1_end);
        
        \path[->] ([shift={(\xshiftleft,  \yshift)}]pic cs:pick2_begin) edge [] node [midway,above=\textabovearrow] 
        {\footnotesize{pick a packet}} ([shift={(\xshiftright,  \yshift)}]pic cs:pick2_end);

        \path[->] ([shift={(\xshiftleft,  \yshift)}]pic cs:pick3_begin) edge [] node [midway,above=\textabovearrow] 
        {\footnotesize{pick a packet}} ([shift={(\xshiftright,  \yshift)}]pic cs:pick3_end);
        
        % Arrows from picked pkt to its transmission
        \path[->] ([shift={(\xshiftleft,  \yshift)}]pic cs:pick1_begin) edge [blue, style=thick] node [midway,above=\textabovearrow] 
        {} ([shift={(-1.7 cm,  0.17 cm)}]pic cs:xmt1);

        \path[->] ([shift={(\xshiftleft-0.05 cm,  \yshift)}]pic cs:pick2_begin) edge [green, style=thick] node [midway,above=\textabovearrow] 
        {} ([shift={(-1.7 cm,  0.17 cm)}]pic cs:xmt2);

        % Curly braces 
        \draw  [decorate,decoration={brace, amplitude=\amplitude}]
        ([shift={(\bracesxshit,  \bracesyshithi)}]pic cs:pick1_end) -- ([shift={(0.1 cm,\bracesyshitlo)}]pic cs:end_buf1) node 
        [black, right, midway, anchor=east, xshift=\bracestotxtx, text width=1.5cm] {\footnotesize buffer picked packet};

        \draw  [decorate,decoration={brace, amplitude=\amplitude}]
        ([shift={(\bracesxshit,  \bracesyshithi)}]pic cs:proc1_begin) -- ([shift={(\bracesxshit,\bracesyshitlo)}]pic cs:proc1_end) node 
        [black, right, midway, anchor=east, xshift=\bracestotxtx, text width=1.5cm] {\footnotesize process picked packet};

        \draw  [decorate,decoration={brace, amplitude=\amplitude}]
        ([shift={(\bracesxshit,  \bracesyshithi)}]pic cs:pick2_end) -- ([shift={(-0 cm,\bracesyshitlo)}]pic cs:end_buf2) node 
        [black, right, midway, anchor=east, xshift=\bracestotxtx, text width=1.5cm] {\footnotesize buffer picked packet};

        \draw  [decorate,decoration={brace, amplitude=\amplitude}]
        ([shift={(\bracesxshit,  \bracesyshithi)}]pic cs:proc2_begin) -- ([shift={(0 cm,\bracesyshitlo)}]pic cs:proc2_end) node 
        [black, right, midway, anchor=east, xshift=\bracestotxtx, text width=1.5cm] {\footnotesize process picked packet};
        
        % Boxes around the picked packets
        \draw ([shift={(-0.7 cm,  0.1)}]pic cs:pick1_begin) node[blockblue] () {};
        \draw ([shift={(-0.8 cm,  0.1)}]pic cs:pick2_begin) node[blockgreen] () {};
        \draw ([shift={(-0.67 cm,  0.14)}]pic cs:pick3_begin) node[blockycyan] () {};

    \end{tikzpicture}
    
    \caption[Running example of \subopt.]{Running example of \subopt\ during the Fill phase.}
    \label{tbl:subopt}
\end{table}

\redtext{Table~\ref{tbl:subopt} illustrates a running example of \subopt. For each cycle, the table details the arriving packets, the action taken by \subopt, its  buffer's content, and its gain so far. For simplicity, we set the maximal number of $U$-packets arriving per cycle to $M=1$. Recall that each arriving packet is either a $(W,1)$-packet or a $\mbest$-packet, and that \subopt\ has a single-slot buffer. We now turn to explain the scenario depicted in Table~\ref{tbl:subopt} iteration by iteration.} 

\redtext{\textbf{iteration 1.} \subopt\ picks a single $\mbest$-packet that arrives in the cycle $\minWork+1$. 
In this iteration, \subopt\ does not process or transmit any packet. Hence, \subopt\ is not work-conserving.} 

\redtext{\textbf{iteration 2.} In the first half of the iteration (cycles $2 \minWork$ through $3\minWork$), \subopt\ processes the packet it picked in iteration 1 and finally transmits it in cycle $3\minWork+1$, thus gaining a profit of $V$. By doing so, \subopt\ empties its single-slot buffer, which is now available for storing a new packet. In one of the cycles in the second half of iteration 2 \subopt\ picks again a $\mbest$-packet. However, \subopt\ does not start processing this packet yet, as \subopt\ never processes packets during the second half of each iteration.} 

\redtext{\textbf{iteration 3.} In the first half of the iteration (cycles $2 \minWork$ through $3\minWork$), \subopt\ processes the packet it picked in iteration 2 and finally transmits it in cycle $5\minWork+1$, thus thus increasing its total gain by $V$. In the second half of the iteration (cycles $5 \minWork + 1$ through $6\minWork$), no $\mbest$-packets arrive. Hence, \subopt\ does not pick any packet. Note that by doing so, \subopt\ is not greedy.}

\redtext{\textbf{iteration 4.} \subopt\ begins the iteration with an empty buffer. However, until cycle $8\minWork-1$ none of the arriving packets is a $\mbest$-packet, and therefore \subopt\ does not pick any of the arriving packets. Finally, in cycle $8\minWork$, a $\mbest$-packet arrives, and \subopt\ picks it into its buffer.}

We have $\frac{N}{2\minWork}$ iterations, and the probability that \subopt\ successfully picks a $\mbest$-packet during an iteration is exactly the probability of there being a $\mbest$-packet arriving during that iteration, which is $1 - (1-p)^{M\minWork}$.
The throughput of \subopt, which we recall is denoted by $TP(\msubopt)$, therefore satisfies
\begin{equation} \label{TP_Opt_rand_WP}
TP(\msubopt) \geq \frac{NV}{2\minWork} \left[1 - (1-p)^{M\minWork}\right]
\end{equation}

We now turn to consider the expected performance of any deterministic algorithm \alg\ for the problem.
We first assume that \alg\ begins the flush phase with a buffer
full of $\mbest$-packets,
all of them unparsed. This provides \alg\ with a profit of $BV$ during the flush phase, while
still having $N$ processing cycles during the fill phase for processing additional packets. This
profit is clearly an upper bound on the maximum possible throughput attainable by \alg\ from packets transmitted during the flush phase, regardless of when they were processed.
For evaluating the gain of \alg\ during the fill phase, it therefore suffices to consider only packets which \alg\ fully processes during this phase.

Consider now the profit of \alg\ from packets transmitted during the fill phase.
Recall that we assume that \alg\ is work-conserving. We assume that \alg\ is also greedy, that is, \alg\ never discards a packet when its buffer is not full; being greedy cannot decrease \alg's performance. \alg\ has packets to process during the entire fill phase, except for the first $\minWork$ cycles (where there are no arrivals yet), namely, for $N' = N - \minWork$ cycles.
Furthermore, since \alg\ is assumed to always accept packets when the buffer is not full, and is work conserving, there exists some $0 < r \leq 1$
such that the number of parsing, and work, cycles performed by \alg\ are $N'r$, and $N'(1-r)$, respectively.

Consider a case where \alg\ reveals a $\mbest$-packet $q$. Then, processing $q$ and finally transmitting it would surely
not decrease the throughput of \alg\ when contrasted with the alternative of dropping $q$.
Thus, the best deterministic algorithm \alg\ would work at least $\minWork-1$ work cycles per each parsing cycle, in which
a $\mbest$-packet is parsed
(recall that we are merely interested in packets, which \alg\ fully processes and transmits during the fill phase).
Therefore, the total number of work cycles contributing to the transmission of such packets is at least $\minWork-1$ times larger than the expected number of parsing
cycles, in which a $\mbest$-packet is revealed: $N'(1-r) \geq N'rp(\minWork-1)$.

If the total number of work cycles during the fill phase exceeds the number of cycles that are necessary for
transmitting all the parsed $\mbest$-packets, \alg\ may work also on $\mworst$-packets. Namely, if $N'(1-r) > N'rp(\minWork-1)$,
then \alg\ may work on $\mworst$-packets for
$N'(1-r) - N'rp(\minWork-1)$ cycles, transmitting at most one $\mworst$-packet once in $W-1$ such cycles.

Combining the above reasoning we conclude that the overall throughput of \alg\ satisfies
\begin{equation}\label{TP_alg_lower_b_WP}
\begin{split}
TP(\malg) & \leq N'rpV + \frac {N'(1-r) - N'rp(\minWork-1)} {W-1} + BV \\
& = (N - \minWork) \left[Vrp + \frac {(1-r) - rp(\minWork-1)} {W-1}\right] + BV
\end{split}
\end{equation}

Considering the ratio between the lower bound on the expected performance of \subopt\ (as captured by Eq.~\ref{TP_Opt_rand_WP}) and the upper bound on the expected performance of \alg\ (as captured by Eq.~\ref{TP_alg_lower_b_WP}) and letting $N \to \infty$, we conclude that no algorithm can have a competitive ratio better than
$$
\frac{V(W-1)}{2\minWork} \cdot \frac{1 - (1-p)^{M\minWork}}{Vrp(W-1)+ 1-r -rp(\minWork-1)}
$$
By choosing $p^* = \left[ V(W-1)+1-\minWork \right]^{-1}$, the result follows.
\end{proof}

We now aim to relate the lower bound established in Theorem~\ref{Theorem_WP_rand} to a simpler and more intuitive function of $M, V$ and $W$. We do so by means of two propositions, which relate the bound to either $\Omega(M)$ or $\Omega(VW)$ for different ranges of $M$. In the propositions we use our notation $p^* = \left[ V(W-1)+1-\minWork \right]^{-1}$ from the proof of Theorem~\ref{Theorem_WP_rand}.

Using this notation, note that Theorem~\ref{Theorem_WP_rand} shows that the competitive ratio is at least
$$
\frac{V(W-1)}{2\minWork} \left[1 - \left(1- p^* \right)^{M\minWork} \right].
$$

In the proofs of both propositions we will repeatedly use the following simple inequality, which holds for any $\minWork \geq 1$:\emph{•}
\begin{equation} \label{simple_inequality_p_star}
\frac{1}{V(W-1)} = \frac{1}{\frac{1}{p^*} + \minWork - 1} \leq p^*.
\end{equation}

The following proposition shows that if $M$ is relatively small, then the lower bound established in Theorem~\ref{Theorem_WP_rand} is $\Omega (M)$.

\begin{proposition} \label{prop:M_over_2_VW}
If $V \geq 1, \minWork \geq 1, W \geq 2$ and $1 \leq M \leq \frac{V(W-1)}{\minWork}$, then
\begin{equation}
\frac{V(W-1)}{2\minWork} \left[1 - \left(1- p^* \right)^{M\minWork} \right] \geq \frac{M}{4} \notag
\end{equation}
\end{proposition}

\begin{proof}
We show by induction on $n$ that for any $1 \leq n \leq V(W-1)$
\begin{equation}
\label{eq:induction_claim}
(1-p^*)^n \leq 1 - \frac{n}{2V(W-1)}.
\end{equation}
By setting $n=M \cdot \minWork$, which is at most $V(W-1)$ by our assumption on $M$, and applying some algebraic manipulation, the result follows.

For $n=1$, Eq.~\ref{eq:induction_claim} reduces to requiring that $\frac{1}{2V(W-1)} \leq p^*$, which holds true due to Eq.~\ref{simple_inequality_p_star}.
For the induction step, by the
induction hypothesis on $n$ we have
$$
\left(1-p^*\right)^{n+1}
\leq
\left(1 - p^*\right) \left[  1 - \frac{n}{2 V(W-1)} \right].
$$
It therefore suffices to prove that
$$
\left(1 - p^*\right) \left[  1 - \frac{n}{2 V(W-1)} \right]
\leq
1 - \frac{n+1}{2 V(W-1)},
$$
which is equivalent to requiring that
$$
\frac {1} {2 V(W-1)}
\leq
p^* \left[ 1 - \frac{n}{2 V(W-1)} \right].
$$
By Eq.~\ref{simple_inequality_p_star} we have
$\frac{1}{2 V(W-1)} \leq \frac{p^*}{2}$, which implies that it suffices to show that
$$
\frac{p^*}{2}
\leq
p^* \left[ 1 - \frac{n}{2 V(W-1)} \right]
$$
which is satisfied for every $n \leq V(W-1)$.
\end{proof}

The following proposition shows that if $M$ is relatively large, then the lower bound established in Theorem~\ref{Theorem_WP_rand} is $\Omega \left(\frac{VW}{\minWork}\right)$.

\begin{proposition}\label{prop:omegaVW}
If $V \geq 1, \minWork \geq 1, W \geq 2$ and $M > \frac{V(W-1)}{\minWork}$, then
\begin{equation}
\frac{V(W-1)}{2\minWork} \left[1 - \left(1- p^* \right)^{M\minWork} \right]
>
\frac{e-1}{4e} \cdot \frac{VW}{\minWork} \notag
\end{equation}
\end{proposition}

\begin{proof}
By our assumption on $M$, and using Eq.~\ref{simple_inequality_p_star}, we have $M \cdot \minWork > V(W-1) \geq \frac{1}{p^*}$.
It follows that
$M \cdot \minWork = a \frac{1}{p^*}$ for some $a >1$, which in turn implies that
$$
\left(1 - p^* \right)^{M \cdot \minWork} = \left[ \left(1 - p^*\right)^{\frac{1}{p^*}}\right]^a \leq e^{-a} <  e^{-1}.
$$
It follows that
\begin{align*}
\frac{V(W-1)}{2\minWork} \left[1 - \left(1- p^* \right)^{M \cdot \minWork} \right]
&\geq \frac{V(W-1)}{2\minWork}  \left(1 - \frac{1}{e}\right) \\ \notag
&= \frac{VW}{2\minWork} \cdot \frac{W-1}{W} \left(1 - \frac{1}{e}\right) \\ \notag
&\geq \frac{e-1}{4e} \frac{VW}{\minWork} \notag
\end{align*}
\end{proof}

Assigning $\minWork=1$ in Theorem \ref{Theorem_WP_rand} and Propositions \ref{prop:M_over_2_VW} and \ref{prop:omegaVW} implies the following corollary:
\begin{corollary}\label{cor:WP_ran}
The competitive ratio of any randomized algorithm is $\Omega(\min\set{VW,M})$.
\end{corollary}

In the special case of uniform-profits, we are essentially interested in maximizing the overall {\em number} of packets successfully transmitted. Therefore we may assign $V=1$ in Corollary \ref{cor:WP_ran}, implying the following corollary:

\begin{corollary}\label{cor:W_rand}
In the case of uniform-profits, the competitive ratio of any randomized algorithm is $\Omega(\min\set{W,M})$.%\leavevmode \\
\end{corollary}

In the special case of uniform-work, we can assign $\minWork = W$ in Propositions \ref{prop:M_over_2_VW} and \ref{prop:omegaVW}, implying the following corollary:
\begin{corollary}\label{cor:P_rand}
In the case of uniform-work, the competitive ratio of any randomized algorithm is $\Omega(\min\set{V,M})$.
\end{corollary}

\subsection{Non-restricted Adversaries}\label{sec:Markov}
In Section~\ref{sec:lower_bounds_restricted} we assumed that the optimal algorithm has a buffer capacity of storing only one packet. This assumption significantly simplified the proofs there.
In this section we relax this assumption, and show a stronger bound for the general, and more natural case, where the size of the buffer available to the optimal algorithm is identical to the size that available to the online algorithm. We use again Yao's method~\cite{yao77probabilistic}, which we used in the proof of Theorem~\ref{Theorem_WP_rand}. Furthermore, we use the same scenario and algorithm \subopt, defined in Section~\ref{sec:lower_bounds_restricted}.
However, as we now allow \subopt\ to store multiple packets in its buffer, \subopt\
can increase its expected throughput by buffering $\mbest$-packets whenever the number of arriving $\mbest$-packets in a single iteration is larger than one, and processing them in iterations where no $\mbest$-packets arrive.
We now evaluate the performance in such settings.

Denote by $q_j$ the state where there are $j$ $\mbest$-packets in the buffer of \subopt\ at the beginning of an iteration. Note, that when $j > 0$, the count represented by $q_j$ also includes the packet, which is to be transmitted during the iteration. Namely, \subopt\ successfully transmits a packet in every iteration, unless its buffer's state is $q_0$. For ease of reference, we provide a summary of this additional notation in the middle section of Table~\ref{tbl:buf:notations}.

We now turn to describe the transition matrix. Denote the probability of having exactly $k$ $\mbest$-packets arriving during one iteration
 by $\malpha_k$. In each iteration we have $M \cdot \minWork$ arriving packets ($M$ packets per cycle, times $\minWork$ cycles per iteration) which are i.i.d. where each packet is a $\mbest$-packet with probability $p$.
Therefore
$\malpha_k = \binom{M \cdot \minWork}{k}p^k (1-p)^{M \cdot \minWork - k}$ when $0 \leq k \leq M \minWork$ and $\malpha_k = 0$ otherwise.

Then, the transition matrix is

\[
\Pi =
\left \{
  \begin{tabular}{ccccccc}
  $ \malpha_0$           & $\malpha_1$ & $\malpha_2$ & \dots       & $\malpha_{B-1}$  & $1 - \sum_{j=0}^{B-1} \malpha_j$ \\
  $ \malpha_0$           & $\malpha_1$ & $\malpha_2$ & \dots       & $\malpha_{B-1}$  & $1 - \sum_{j=0}^{B-1} \malpha_j$ \\
  0                     & $\malpha_0$ & $\malpha_1$ & \dots       & $\malpha_{B-2}$  & $1 - \sum_{j=0}^{B-2} \malpha_j$ \\
  \dots                 & \dots      & \dots      & \dots       & \dots           & \dots \\
  0                     & \dots      & 0          &$\malpha_0$   & $\malpha_1$      & $1 - \sum_{j=0}^{1} \malpha_j$\\
  0                     & 0          & \dots      & 0           & $\malpha_0$      & $1 - \malpha_0$\\
  \end{tabular}
\right \}
\]
where $\Pi_{ij}$ is the probability of transition from state $i$ to state $j$ for each $0 \leq i, j \leq B$.
$\Pi$ is irreducible, because it is possible to get from any buffer state to any other buffer state by some arrival sequence. $\Pi$ is also aperiodic, because its diagonal is non-zero, which represents the fact that if the buffer contains $i$ packets at the beginning of a certain iteration, there exists a positive probability that it would contain $i$ packets also at the beginning of the next iteration.
Furthermore, as $\Pi$ is finite, irreducible and aperiodic, it is also
ergodic, namely, there exists a steady state.
For a long enough input sequence, we can neglect the transient "warm-up" period, and assume that the expected number of iterations where \subopt\ gains nothing during phase 1 is $\frac{N}{2 \minWork} \cdot p_0$, where $p_0$ is the probability that \subopt\ is in state $q_0$. In the rest of the iterations in phase 1 \subopt\ gains $V$ per iteration.
Therefore, the expected throughput of \subopt\ satisfies
\begin{equation} \label{W_rand_SO_B_gt_1}
TP(\msubopt) \geq \frac{N}{2 \minWork} \cdot V(1 - p_0)
\end{equation}

The expected throughput of \alg\  remains the same as in Eq.~\ref{TP_alg_lower_b_WP}.
In order to obtain the competitive ratio for the fully heterogenous case, we divide Eq.~\ref{W_rand_SO_B_gt_1} by Eq.~ \ref{TP_alg_lower_b_WP} and assign again $\minWork=1$ and $p^* = \frac{1}{V(W-1)}$. Then, when $N \to \infty$ the competitive ratio is $c \geq \frac{V}{2}(W-1)(1 - p_0)$.

We find $p_0$ by solving the balance equations defining the steady state of the system, i.e., finding the eigenvector of the transition matrix $\Pi$.
Fig.~\ref{fig:tight_lower_bnds} depicts the lower bounds as a function of $M$ when $V = W = 10$ for various buffer sizes.
Recall that the probability of a certain packet to be a $\mbest$-packet is $p_* = \frac{1}{V(W-1)} = \frac{1}{90}$. Therefore
only when $M$ is large enough, the expected number of $\mbest$-packets per iteration is sufficient for allowing \subopt\ to really take advantage of its buffer for increasing its performance, resulting in a stronger lower bound on the competitive ratio.

\begin{figure}
\centering
\begin{tikzpicture}
\begin{axis}[
    xlabel = {M},
    ylabel = {Competitive ratio lower bound},
    xmin=0, xmax=97,
    ymin=0, ymax=46,
    xtick={0,20,40,60,80,100},
    ytick={0,5, 10, 15, 20,25, 30, 35, 40, 45},
    legend pos=north west,
    ymajorgrids=true,
    grid style=dashed,
    legend style = {font=\figfontsize},
    label style={font=\figfontsize},
    every axis plot/.append style={very thick},
    tick label style={font=\figfontsize},
]

\addplot[color=blue,mark=square,] coordinates {
(1, 0.50000)(6, 2.91789)(11, 5.20441)(16, 7.36669)(21, 9.41148)(26, 11.34517)(31, 13.17379)(36, 14.90306)(41, 16.53836)(46, 18.08481)(51, 19.54724)(56, 20.93021)(61, 22.23803)(66, 23.47479)(71, 24.64436)(76, 25.75037)(81, 26.79630)(86, 27.78539)(91, 28.72074)(96, 29.60526)};
  \addlegendentry{B = 1}

\addplot[color=yellow,mark=o] coordinates {
(1, 0.50000)(6, 2.99863)(11, 5.48822)(16, 7.95838)(21, 10.39756)(26, 12.79346)(31, 15.13346)(36, 17.40519)(41, 19.59693)(46, 21.69811)(51, 23.69968)(56, 25.59435)(61, 27.37676)(66, 29.04357)(71, 30.59338)(76, 32.02659)(81, 33.34524)(86, 34.55276)(91, 35.65367)(96, 36.65337)};
  \addlegendentry{B = 2}

\addplot[color=red,mark=triangle*] coordinates {
    (1, 0.50000)(6, 3.00000)(11, 5.49999)(16, 7.99988)(21, 10.49932)(26, 12.99732)(31, 15.49172)(36, 17.97839)(41, 20.45024)(46, 22.89618)(51, 25.30028)(56, 27.64139)(61, 29.89366)(66, 32.02817)(71, 34.01569)(76, 35.83005)(81, 37.45160)(86, 38.86960)(91, 40.08321)(96, 41.10089)};
  \addlegendentry{B = 4}

\addplot[color=cyan,mark=x,]    coordinates {
    (1, 0.50000)(6, 3.00000)(11, 5.50000)(16, 8.00000)(21, 10.50000)(26, 13.00000)(31, 15.50000)(36, 17.99997)(41, 20.49984)(46, 22.99928)(51, 25.49719)(56, 27.99033)(61, 30.47028)(66, 32.91774)(71, 35.29416)(76, 37.53424)(81, 39.54817)(86, 41.24314)(91, 42.55981)(96, 43.49918)};
  \addlegendentry{B = 8}

\addplot[color=black,mark=triangle,]coordinates {
    (1, 0.50000)(6, 3.00000)(11, 5.50000)(16, 8.00000)(21, 10.50000)(26, 13.00000)(31, 15.50000)(36, 18.00000)(41, 20.50000)(46, 23.00000)(51, 25.50000)(56, 27.99999)(61, 30.49992)(66, 32.99933)(71, 35.49503)(76, 37.96951)(81, 40.34821)(86, 42.41841)(91, 43.86681)(96, 44.60602)};
  \addlegendentry{B = 16}
\end{axis}
\end{tikzpicture}
\caption[Lower bounds on the competitive ratio of every randomized algorithm]{Lower bounds on the competitive ratio of every randomized algorithm when $W = V = 10$ where the number of unknown packets arriving in a time slot varies, for different values of $B$. }
\label{fig:tight_lower_bnds}
\end{figure}
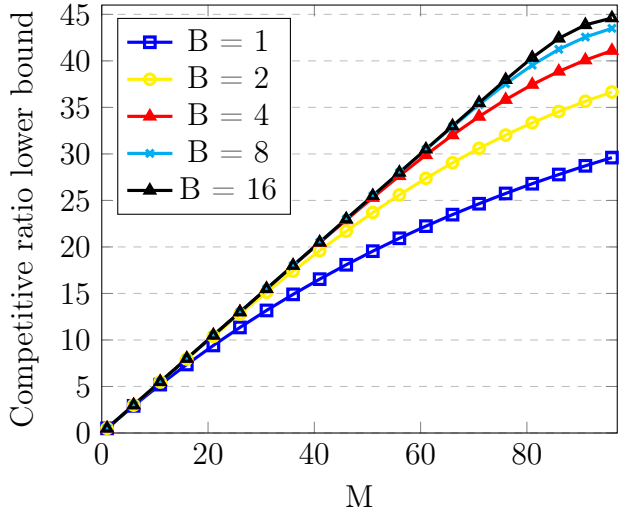

In the next section we use the insight obtained from the analysis in the current section to identify several algorithmic concepts appropriate for the problem of buffering with limited knowledge.

\section{Algorithmic Concepts}
\label{sec:algorithmic_conecpts}
In this section we describe the algorithmic concepts underlying our proposed algorithms for
dealing with scenarios of limited knowledge.

\paragraph*{Random selection}
For obtaining a good competitive ratio we would like to avoid a scenario where \opt\ successfully transmits a bulk of ``good'' packets, which are originally unknown, while having the online algorithm discard all these packets. This translates to assuring each arriving $U$-packet has some minimal probability of being accepted and parsed.

\paragraph*{Speculatively Admit}
Competitive algorithms must ensure they retain throughput from both $K$-packets and $U$-packets.
Furthermore, once a $U$-packet is accepted, there is a high motivation to reveal its
characteristics as soon as possible, thus making educated decisions in the next cycles.

We therefore propose to {\em speculatively} over-prioritize unknown packets over known packets in certain cycles. We refer to
the act of over-prioritizing an unknown packet $p$ in some cycle $t$ as {\em admitting} $p$. Respectively, we refer to such a cycle $t$ as an {\em admittance cycle}, and to such a packet $p$ as an {\em
admitted packet}.

\paragraph*{Classify and randomly select}
Intuitively, as unknown packet characteristics are drawn from a wider range of values, the task of maximizing throughput becomes harder, especially when compared to the optimal throughput possible. To deal with this diversity, we apply a Classify and Randomly Select scheme~\cite{RCnS}.

This approach is based on the following notion: Assume we have an algorithm $\algc$ which is guaranteed to be $c$-competitive if traffic is sufficiently uniform, i.e., for cases where traffic characteristics are within some well-defined range of values. Given some arbitrary input sequence, which might be highly heterogeneous, we virtually partition the sequence of arriving packets into $N>1$ disjoint sub-sequences, which we refer to as {\em classes}, such that each class is sufficiently uniform, i.e., for any specific class $1 \leq i \leq N$ the characteristics of packets corresponding to class $i$ are within some well-defined range of values (as prescribed by $\algc$). The scheme then dictates selecting one of the classes {\em uniformly at random}, and applying $\algc$ to this class, while ignoring all packets corresponding to other classes. One then shows that the overall competitive ratio of this randomized approach is $O(N\cdot c)$-competitive for the overall input sequence.

\paragraph*{Alternate between fill $\&$ flush}
This paradigm is especially crucial in cases of limited information. The main motivation for this approach is that whenever a ``good'' buffer state is identified, the algorithm should focus all its efforts on monetizing the current state, maybe even at the cost of dropping packets indistinctly.
In terms of buffer management and scheduling, this translates to defining some periods, in which the algorithm processes and transmits all the packets in its buffer, even at the cost of discarding all the arrivals. If these flush periods are short enough, the algorithm gains the high throughput from flushing its buffer, yet without compromising too much throughput due to having packets discarded during the flush.

\section{Competitive Algorithms}
\label{sec:algorithms}
In this section we present a basic competitive online algorithm for the problem of buffering and scheduling with limited knowledge. We first provide a high-level description of our algorithm % (Section~\ref{sec:algorithms_high_level_description}), 
and then turn to specify its details 
% (Sections~\ref{sec:RC&S}-\ref{sec:RC&S_concrete}) 
and analyze its performance. To ease understanding, we also provide a running example of our algorithm (Section~\ref{sec:buf:run_example}).

For simplicity of analysis and algorithm presentation, we assume that the set of possible values of $W$ and $V$ -- the work and profit per packet, respectively -- are known to the algorithm in advance.
In Sections~\ref{sec:improved_algorithms} and~\ref{sec:Practical_Implementation} we show how to remove this assumption without harming the performance of our algorithm, and present several improved variants of this algorithm.
We further note that neither of our proposed solutions \redtext{requires} knowing the value of $M$ - the maximum number of unknown packets arriving in a single cycle - in advance.

\subsection{High-level Description of Proposed Algorithm}
\label{sec:algorithms_high_level_description}

Our algorithm is designed according to the algorithmic concepts presented in Section~\ref{sec:algorithmic_conecpts} as follows.

\paragraph*{Randomly select and speculatively admit}
In every cycle $t$ during which a $U$-packet arrives, the algorithm picks
$t$ as an admittance cycle with some probability $r$
(to be determined in the sequel). In every cycle chosen
as an
admittance cycle, the algorithm picks exactly one of the $U$-packets arriving at $t$ to serve as the {\em admitted} packet.  This $U$-packet is chosen uniformly at random out of all $U$-packets arriving at $t$.
At the end of the arrival step, the algorithm schedules the admitted $U$-packet (if one exists) for processing,
hence {\em parsing} the packet. We note that if no such $U$-packet exists, or if $t$ is not an admittance cycle, then the algorithm may only accept known arriving packets, and would eventually schedule the top-priority packet residing in the Head-of-Line (HOL) for processing. The exact notion of priority will be detailed later.

\paragraph*{Classify and randomly select}
We implicitly partition the possible types of arriving packets into classes $C_1, C_2, \ldots C_m$; the criteria for partitioning and the exact value of $m$ will be specified later. Our algorithm picks a single {\em selected} class, uniformly at random from the $m$ classes.
Our goal is to provide {\em guarantees} on the performance of our proposed algorithm for packets belonging to the selected class, which is henceforth denoted $\Cs$.
Packets which belong to the selected class are referred to as {\em $\Cs$-packets}.
Following our previously introduced notation, known (unknown) packets that belong to the selected class, i.e.,
$\Cs$-packets for which their attributes are known (unknown), are denoted as {\em $\CsK$-packets} ({\em $\CsU$-packets}).

Focusing solely on packets belonging to $\Cs$ may seem like a questionable choice, especially if there are few packets arriving which belong to this class, or if the characteristics of packets belonging to this class are poor (e.g., they have low profit and require much work). However, this naive description is meant only to simplify the analysis.
In Section~\ref{sec:improved_algorithms} we show how to remedy
this naive approach in order to deal with these apparent shortcomings, while keeping the analytic guarantees intact.

\paragraph*{Alternate between fill $\&$ flush}
Our algorithm will be alternating between two states: the {\em fill} state, and the {\em flush} state.
We define an algorithm to be {\em \hfull} if its buffer is filled with known $\Cs$-packets.
Once becoming \hfull, our algorithm switches to the flush state, during which it discards all arriving packets and
continuously processes queued packets. Once the buffer empties, the algorithm returns to the fill phase.
Again, in Section~\ref{sec:improved_algorithms} we show how to improve upon this naive simplified approach.

\subsection{A General Classify and Randomly Select Mechanism}\label{sec:RC&S}
We now turn to explain the fundamentals of the classifying mechanism of our algorithm.

For each packet $p$ we assign a \emph {work-class} $C\W_i$, and denote the set of potential characteristic values within class $C\W_i$ by $X\W_i$.
 Let $\deltaW$ denote the maximal ratio between the work values of two packets, which belong to the same work-class. Similarly, for each packet $p$ we assign a \emph {profit-class} $C\P_i$, and denote the set of potential characteristic values within class $C\P_i$ by $X\P_i$. Let $\deltaP$ denote the maximal ratio between the profits of two packets, which belong to the same profit-class. Throughout our analysis, we will use $\deltaP$ and $\deltaW$ which are both constants.

Denote by $\mW$ and $\mP$ the number of work-classes and profit-classes, respectively.
We say a packet $p$ is of {\em combined-class} $C_{(i,j)}$  if it is of work-class $C\W_i$ and of profit-class $C\P_j$.
Note that in terms of work, the class to which a packet $p$ belongs is defined statically by the total work of
$p$, and does not depend upon its remaining processing cycles, which may change over time.

Upon initialization, the algorithm selects a class by picking $i^* \in \set{1,\ldots,\mW}$ and $j^* \in
\set{1,\ldots,\mP}$, each chosen uniformly at random. Then, the selected combined-class is $\Cs =
C_{(i^*,j^*)}$.

We will later define several ways to partition the packets into classes, each tailored and optimized for some specific scenarios of possible work and profit values.
For ease of reference, we provide a summary of this additional notation in the bottom section of Table~\ref{tbl:buf:notations}.

\subsection{The \SAM\ Algorithm}
\label{sec:SAM}
We now describe the details of our algorithm, Speculatively Admit (\SAM), depicted in Algorithm~\ref{alg:SAM}.
The pseudo-code in Algorithm~\ref{alg:SAM} uses the procedures \emph{UpdatePhase()}, \emph{SortBuf()},
and \emph{\makeroom}, whose pseudo-code appears in Algorithms~\ref{alg:UpdatePhase},~\ref{alg:SortBuf} and~\ref{alg:MakeRoom}, respectively. The procedure \makeroom\ is destined to assure a free space for a high-priority arriving packet, even at the cost of pushing-out and dropping a lower-priority packet from the tail of the buffer, if the buffer is full.

\newcommand\AlgPhase[1]{%
\hspace*{\dimexpr-\algorithmicindent-2pt\relax}%
\vspace*{-.5\baselineskip}\Statex\hspace*{-\algorithmicindent}{\em #1}%
\vspace*{-.9\baselineskip}\Statex\hspace*{\dimexpr-\algorithmicindent-2pt\relax}%
}

\begin{algorithm}[t!]
\caption {UpdatePhase()}\label{alg:UpdatePhase}
\begin{algorithmic}[1]
\If {buffer is empty}
    \State {\em phase} = fill
\ElsIf {buffer is \hfull}
    \State {\em phase} = flush
\EndIf
\Comment {if buffer is neither empty  nor \hfull, {\em phase} is unchanged.}
\end{algorithmic}
\end{algorithm}

\begin{algorithm}[t!]
\caption {SortBuf()}\label{alg:SortBuf}
\begin{algorithmic}[1]
\State sort queued packets as follows: admitted packet first; $\CsK$-packets next; rest of the packets last; break ties by FIFO
\end{algorithmic}
\end{algorithm}

\begin{algorithm}[t!]
\caption {MakeRoom()}\label{alg:MakeRoom}
\begin{algorithmic}[1]
\If {the buffer is full}
	\State SortBuf()
	\State drop a packet from the tail
\EndIf
\end{algorithmic}
\end{algorithm}

\begin{algorithm}[t!]
\caption {\SAM: at every time slot $t$ after transmission} \label{alg:SAM}
\begin{algorithmic}[1]
\Statex
\AlgPhase{Arrival Step:}
\State {\em phase} = UpdatePhase()
\label{alg:SAM:line:UpdatePhaseAtArrivalBegin}
\State {\em admittance} = true w.p. $r$
\label{alg:SAM:line:DecideAdmittance}
\While {{\em phase} $==$ fill \AND exists arriving packet $p$} \label{alg:SAM:line:while_begin}
    \If {$p$ is a $\CsK$-packet}\label{alg:SAM:if_is_Gk}
	   \If {there are $B-1$ $\CsK$-packets in the buffer}\label{alg:SAM:if_B_minus_1}
	       \State drop admitted packet if exists        \label{alg:SAM:drop_ap}
      \EndIf    \label{alg:SAM:if_B_minus_1_endif}
	   \State \makeroom\label{alg:SAM:MakeRoom_by_Gk}
      \State accept $p$ \label{alg:SAM:accept_Gk}
    \ElsIf {$p$ is unknown AND \emph{admittance}}\label{alg:SAM:if_p_is_U}
        \If {$\Autp=1$}\label{alg:SAM:if_reservoir}
    	   \State \makeroom \label{alg:SAM:MakeRoom_by_ap}
            \State mark $p$ as admitted
            \State accept $p$ \label{alg:SAM:accept_ap}
        \Else
    	   \State w.p. $1/\Autp$, swap the admitted packet with $p$.
            \label{alg:SAM:line:reservoir}
        \EndIf\label{alg:SAM:if_reservoir_endif}
    \EndIf
    \If {buffer is not full} \label{alg:SAM:line:if_full}
        \State accept $p$ \label{alg:SAM:line:greedily_accept}
    \EndIf
	\State {\em phase} = UpdatePhase()
	\label{alg:SAM:line:UpdatePhase_in_arrival}
	\State SortBuf()
    \label{alg:SAM:line:sort_queue_in_arrival}
\EndWhile \label{alg:SAM:line:while_end}
\Statex
\AlgPhase{Processing Step:}
\State process HoL-packet
\label{alg:SAM:line:processHoL}
\State {\em phase} = UpdatePhase()
\label{alg:SAM:line:UpdatePhase_after_processing}
\State SortBuf()
\label{alg:SAM:line:sort_queue_in_processing}
\end{algorithmic}
\end{algorithm}

Once in the arrival step, algorithm \SAM\ updates its phase (line \ref{alg:SAM:line:UpdatePhaseAtArrivalBegin}).
In each cycle, the algorithm tosses a coin with some probability $r$, to be determined later, to decide whether this is an \emph{admittance cycle}, namely, a cycle in which the algorithm may admit an unknown packet (line~\ref{alg:SAM:line:DecideAdmittance}).
If the phase is flush, the algorithm skips the while loop (lines
\ref{alg:SAM:line:while_begin}-\ref{alg:SAM:line:while_end}), thus discarding all arriving packets.

If the phase is fill, which in particular implies that the buffer is not \hfull, the algorithm accepts every arriving $\CsK$-packet (lines~\ref{alg:SAM:if_is_Gk}-\ref{alg:SAM:accept_Gk}).
For assuring a free slot for the arriving $\CsK$-packet, the algorithm calls \makeroom\ (line~\ref{alg:SAM:MakeRoom_by_Gk}) before accepting the packet (line~\ref{alg:SAM:accept_Gk}).
The if-clause in lines~\ref{alg:SAM:if_B_minus_1}-\ref{alg:SAM:if_B_minus_1_endif} handles the special case where there are already $B-1$ $\CsK$-packets in the buffer; in this special case, after accepting the arriving $\CsK$-packet, the buffer will become \hfull, and therefore it should stop admitting packets.

If the phase is fill and this is an \emph{admittance} cycle (line~\ref{alg:SAM:if_p_is_U}), the algorithm admits a single $U$-packet arriving in this cycle, if such a packet exists.
In lines~\ref{alg:SAM:if_reservoir},\ref{alg:SAM:line:reservoir}, $\Autp$ denotes the number of $U$-packets which arrive in cycle $t$ by the arrival of packet $p$, including $p$ itself. Lines~\ref{alg:SAM:if_reservoir}-\ref{alg:SAM:if_reservoir_endif}
essentially perform a reservoir sampling~\cite{Reservoir}, which imply that the admitted $U$-packet is chosen uniformly at random out of all $U$-packets arriving in this cycle.

Finally, if the buffer is not full, the algorithm greedily accepts every arriving packet (lines \ref{alg:SAM:line:if_full}-\ref{alg:SAM:line:greedily_accept}).

While in the processing step, the algorithm simply processes the top-priority packet in the buffer (line~\ref{alg:SAM:line:processHoL}). Finally, the algorithm updates its phase and sorts the queued packets each time it either accepts or processes a packet (lines
\ref{alg:SAM:line:UpdatePhase_in_arrival}-\ref{alg:SAM:line:sort_queue_in_arrival} and
\ref{alg:SAM:line:UpdatePhase_after_processing}-\ref{alg:SAM:line:sort_queue_in_processing}). Note that the marking of a packet as an ``admitted packet'' is \emph{cycle-based}, namely, once an admitted packet is processed, it is not considered ``admitted'' anymore.
To better understand \SAM, please refer to Appendix~\ref{sec:buf:run_example}, showing a running example of the algorithm.

\subsection{Performance Analysis}
\label{sec:SAMWP}

We now turn to show an upper bound on the performance of our algorithm (for $W, V > 1$), captured by the following theorem. 
\begin{theorem}\label{thm:SAMWP}
\SAMWP\ is $O 
\left(\left[
\frac{M}{r} + \deltaW \cdot \deltaP
\right] \cdot \mW \cdot \mP \right)$
-competitive.
\end{theorem}

We now define additional notation which we will use for proving Theorem~\ref{thm:SAMWP}.
For ease of reference, this additional notation appears in the bottom section of Table~\ref{tbl:buf:notations}.

For every cycle $t$ and packet type $\alpha$, we denote by $A^{\alpha}(t)$ the number of $\alpha$-packets that arrive in cycle $t$. For instance, $\Ak(t)$ ($\Au(t)$) denotes the number of $K$-packets ($U$-packets) which arrive in cycle $t$. This notation can be combined with the work and profit values of packets. For instance, $\Au_{(w,v)}(t)$ denotes the number of $U$-packets with work $w$ and profit $v$, which arrive in cycle $t$.

Our proofs involve a careful analysis of the expected profit of our algorithms from packets which arrive when it is either in the fill or the flush phase. Therefore, we now turn to define the exact notion of cycles belonging to either phase.
We say that an algorithm is in the flush phase in a specific cycle $t$ if it is in the flush state at the end of the arrival step of cycle $t$. If it's not in the flush phase in cycle $t$, then we say it is in the fill phase in cycle $t$. Denote by $P\mFILL$ and $P\mFLUSH$ the sets of cycles in which our algorithm is in the fill and flush phases, respectively.

For every packet type $\alpha$, we denote by $S_{\alpha}(t)$ the expected profit of the algorithm from
$\alpha$-packets which {\em arrive} in cycle $t$, and by
$S_{\alpha}=\sum_{t}S_{\alpha}(t)$ the overall
expected profit of \alg\ from $\alpha$-packets. We denote by $O_{\alpha}$ the expected profit of some optimal solution, \opt, from $\alpha$-packets. Again, this notation can be combined with previous notations. For instance, $O_{\CsU}$ denotes the overall expected profit of \opt\ from $\CsU$-packets. Furthermore, $O\mFILL_{\CsU}$ denotes the expected profit of \opt\ from $\CsU$-packets which arrive during $P\mFILL$.

The proof Theorem~\ref{thm:SAMWP} follows from a series of propositions.
Initially, we aim to prove that \SAM\ successfully transmits every $\CsK$-packet which arrives during the fill phase, by showing that it never drops such a packet once it is accepted to the buffer.

\begin{proposition}
\label{prop:transmits_all_G_K}
\SAM\ successfully transmits every $\CsK$-packet which arrives during the fill phase.
\end{proposition}

\begin{proof}

We first note, that any $\CsK$-packet arriving during the fill phase (depicted by the while loop in lines~\ref{alg:SAM:line:while_begin}-\ref{alg:SAM:line:while_end}) is accepted (line~\ref{alg:SAM:accept_Gk}).

Next, we show that \SAM\ never drops a $\CsK$-packet which resides in its buffer.
We consider all cases where \SAM\ drops a packet from its buffer, and prove that it cannot be a $\CsK$-packet.

In line~\ref{alg:SAM:drop_ap}, \SAM\ drops an admitted packet, namely, a picked $U$-packet, and not a $\CsK$-packet.

In line~\ref{alg:SAM:MakeRoom_by_Gk}, \SAM\ performs the \makeroom\ procedure, which may result in dropping the last packet in the buffer. However, as this line dwells within the while loop of lines~\ref{alg:SAM:line:while_begin}-\ref{alg:SAM:line:while_end}, we know that the phase is fill, and therefore there are at most $B-1$ $\CsK$-packets in the buffer. Furthermore, if there are exactly $B-1$ $\CsK$-packets in the buffer, the if-clause in lines~\ref{alg:SAM:if_B_minus_1}-\ref{alg:SAM:if_B_minus_1_endif} assures that there is no admitted packet in the buffer. Hence, if the buffer is full, it contains at least one low-priority packet -- namely, a packet which is not admitted and not a $\CsK$-packet. After sorting the buffer, this low-priority, non-$\CsK$ packet, will be located in the tail of the queue and dropped.

\SAM\ may perform the \makeroom\ procedure also in line~\ref{alg:SAM:MakeRoom_by_ap}, if $\Autp = 1$. In this case, the arriving packet $p$ is the first $U$-packet arriving in this cycle -- and it is not admitted yet. As a result, there is no admitted packet in the buffer. Furthermore, as this line is executed during the fill phase (the while loop of lines~\ref{alg:SAM:line:while_begin}-\ref{alg:SAM:line:while_end}), there are at most $B-1$ $\CsK$-packets in the buffer. Hence, if the buffer is full, it contains at least one low-priority, non-$\CsK$-packet, which is the packet dropped.
\end{proof}

The following lemma shows that the overall number of $\Cs$-packets transmitted by \SAMWP\ is at least a significant fraction of the number of $\Cs$-packets accepted by an optimal policy during a fill phase.

\begin{lemma}\label{Cs_K_WP}
$S_{\Cs} \geq \frac{r}{M} O\mFILL_{\Cs}$.
\end{lemma}

\begin{proof}
Let $t$ denote a cycle in the fill phase, in which $U$-packets arrive. Then, with probability $r$ \SAMWP\ admits one $U$-packet, denoted $p$. As the algorithm implements reservoir sampling~\cite{Reservoir}, $p$ is picked uniformly at random out of at most $M$ unknown arrivals, and therefore the probability that $p \in \CsU$ is at least $A_{\CsU}(t)/M$.
As $p$ is parsed in the cycle of arrival, in the subsequent cycle it is known.
By Proposition~\ref{prop:transmits_all_G_K}, if $p$ is a $\CsK$-packet, then \SAMWP\ will eventually transmit $p$.
Recalling that $X\W_{i^*}$ and $X\P_{j^*}$ denote the ranges of the work and profit values within the selected work and profit class $C_{(i^*,j^*)}$ (see Section~\ref{sec:RC&S_concrete}), we conclude that
\begin{equation} \label{Eq:S_G_U_t}
S_{\CsU} (t) \geq \frac{r}{M} \sum_{ w \in X\W_{i^*}, v \in X\P_{j^*} }
[v \cdot A_{(w,v)}\U(t)].
\end{equation}

Summing Eq.~\ref{Eq:S_G_U_t} over all the cycles within the fill phase,

\begin{equation}\label{Eq:S_G_U}
S_{\CsU} \geq \frac{r}{M} \sum_{t \in P\mFILL} \sum_{ w \in X\W_{i^*}, v \in X\P_{j^*} }
[v \cdot A_{(w,v)}\U(t)]
\geq \frac{r}{M} O\mFILL_{\CsU}.
\end{equation}

In addition, by Proposition~\ref{prop:transmits_all_G_K}, $S_{\CsK} \geq O\mFILL_{\CsK}$. Therefore
\begin{equation}\label{Eq:S_G}
S_{\Cs} = S_{\CsK} + S_{\CsU}
\geq \frac{r}{M} (O\mFILL_{\CsK} + O\mFILL_{\CsU}) = \frac{r}{M} O\mFILL_{\Cs}.
\end{equation}
\end{proof}

We are now in a position to prove Theorem \ref{thm:SAMWP}.
\begin{proof}[Proof of Theorem \ref{thm:SAMWP}]

Every class $C_{(i,j)}$ is the selected class with probability $\frac{1}{\mW \cdot \mP}$. Using Lemma \ref{Cs_K_WP}
we therefore have for all
$i \in \set {1,2, \dots, \mW}$ and $j \in \set {1,2, \dots, \mP}$,
$S_{(i,j)} \geq \frac{r}{M \cdot \mW \cdot \mP} O\mFILL_{(i,j)}$.

Summing over all the classes, we obtain that the expected performance of our algorithm satisfies
\begin{equation} \label{eq:case0WV}
\sum_{i=1}^{\mW}
\sum_{j=1}^{\mP} S_{(i,j)} \geq
\frac{r}{M \cdot \mW \cdot \mP} \sum_{i=1}^{\mW}
\sum_{j=1}^{\mP}
O\mFILL_{(i,j)}.
\end{equation}

If \SAMWP\ is never \hfull\ during an arrival sequence, then $O_{(i,j)} = O\mFILL_{(i,j)}$ and therefore, by Eq.~\ref{eq:case0WV} the ratio between the performance of \opt\ and the expected throughput of \SAMWP\ is at most $\frac{M}{r} \cdot \mW \cdot \mP$, as required.

Assume next that \SAMWP\ becomes \hfull\ during an input sequence.
In such a case we compare the overall throughput due to packets {\em transmitted} by \SAMWP\ until the first cycle in which its buffer is empty again, and the profit obtained by \opt\ due to packets {\em accepted} by \opt\ during the same interval. We note that our analysis would also apply to subsequent such intervals, namely, until the subsequent cycle in which \SAMWP\ is empty again.

We note that in case \SAMWP\ becomes \hfull, \SAMWP\ holds in its buffer exactly $B$ $\Cs$-packets, and all these packets are transmitted by the time \SAMWP\ is empty again.
By the definition of $\deltaW$ in Section~\ref{sec:RC&S}, the maximal work which \SAMWP\ dedicates to any of these packets is at most $\deltaW$ times higher than the \emph {minimal} work which \opt\ dedicates to any $\Cs$-packet.
As a result, during the flush phase, in which \SAM\ handles $B$ $\Cs$-packets, \opt\ can handle at most
$\deltaW B + B$ $\Cs$-packets. Furthermore, by the definition of $\deltaP$ in Section~\ref{sec:RC&S}, the maximal profit of \opt\ from any $\Cs$-packet is at most $\deltaP$ higher than the \emph{minimal} profit of \SAMWP\ from any $\Cs$-packet. Combining the above reasoning implies that

\begin{equation}\label{Eq:delta_pi}
\frac{O\mFLUSH_{\Cs} }{S_{\Cs}} \leq \frac{\deltaW B + B}{B} \cdot \deltaP = (\deltaW + 1)\deltaP.
\end{equation}

As every class $C_{(i,j)}$ is the selected class w.p. $\frac{1}{\mW \cdot \mP}$, we have
$$\forall
i \in \set {1 \dots \mW},
j \in \set {1 \dots \mP},
S_{(i,j)} \geq
\frac{1}{(\deltaW + 1)\deltaP \cdot \mW \cdot \mP} O\mFLUSH_{(i,j)}.$$

Summing over all the classes we obtain
\begin{equation}\label{eq:case1WV}
\sum_{i=1}^{\mW}
\sum_{j=1}^{\mP}
S_{(i,j)} \geq
\frac{1}{(\deltaW + 1)\deltaP \cdot \mW \cdot \mP}
\sum_{i=1}^{\mW}
\sum_{j=1}^{\mP}
O\mFLUSH_{(i,j)}.
\end{equation}

Combining Equations~\ref{eq:case0WV} and~\ref{eq:case1WV} implies that the competitive ratio of \SAMWP\ is at most
\begin{equation}\label{Eq:case01WV}
\frac {
\sum_{i=1}^{\mW}
\sum_{j=1}^{\mP}
\left[ O\mFILL_{(i,j)} + O\mFLUSH_{(i,j)} \right]}
{
\sum_{i=1}^{\mW}
\sum_{j=1}^{\mP}
S_{(i,j)}}
\leq \left[ \frac{M}{r} + (\deltaW + 1)\deltaP \right] \cdot \mW \cdot \mP,
\end{equation}
which completes the proof.
\end{proof}

Theorem~\ref{thm:SAMWP} shows an inverse linear dependency of the competitive ratio on the
probability of choosing a cycle as an admittance cycle $r$. Thus, the best competitive ratio is attained for $r=1$, i.e., every cycle where $U$-packets arrive should be an admittance cycle.
In practical scenarios, however, one might want to be more conservative in choosing admittance cycles. E.g., one might choose $r<1$ so as to allow non-parsing cycles even when $U$-packets arrive, thus speeding up the processing of $\CsK$-packets. If one indeed chooses $r=1$, randomization should be maintained only for choosing the specific $U$-packet to be admitted, and the choice of the selected class.
We further explore the effect of the choice of parameter $r$ in Section~\ref{sec:simulations}.

In the special cases of homogeneous work values (homogeneous profit values), we assign $\deltaW = \mW = 1$ ($\deltaP = \mP = 1$, resp.) in the upper bound implied by Theorem \ref{thm:SAMWP}, and obtain the following corollary:
\begin{corollary}\begin{inparaenum}[(a)]
\hfill \break
\item In the special case of homogeneous work values, \SAMWP\ is $O\left(\left(\frac{M}{r} + \deltaP \right) \cdot \mP\right)$-competitive.
\\\item In the special case of homogeneous profit values, \SAMWP\ is $O\left( \left(\frac{M}{r} + \deltaW \right) \cdot \mW\right)$-competitive.
\end{inparaenum}
\end{corollary}

In the special case where all packets are known upon arrival, we obtain the following upper bound on the competitive ratio of \SAMWP:
\begin{corollary}
When $M=0$, \SAMWP\ is $O(\deltaW \cdot \deltaP \cdot \mW \cdot \mP)$-competitive.
\end{corollary}

\begin{proof}
We follow the proof of Theorem~\ref{thm:SAMWP}, and carefully check the required changes.

When all packets are known, Proposition~\ref{prop:transmits_all_G_K} remains essentially intact.
Furthermore, we have $S_{\Cs} = S_{\CsK} \geq O\mFILL_{\CsK} = O\mFILL_{\Cs}$, which replaces Lemma~\ref{Cs_K_WP}. Accordingly, Eq.~\ref{eq:case0WV} is modified to
\begin{equation}\label{eq:case0WV_M0}
\sum_{i=1}^{\mW}
\sum_{j=1}^{\mP} S_{(i,j)} \geq
\frac{1}{\mW \cdot \mP} \sum_{i=1}^{\mW}
\sum_{j=1}^{\mP}
O\mFILL_{(i,j)}.
\end{equation}

Eq.~\ref{eq:case1WV} remains intact, as in deriving it we use the classify and randomly select scheme, independently of $M$. Combining Equations~\ref{eq:case0WV_M0} and~\ref{eq:case1WV} implies that when all packets are known, the competitive ratio of \SAMWP\ is at most
\begin{equation}\label{Eq:case01WV_M0}
\frac {
\sum_{i=1}^{\mW}
\sum_{j=1}^{\mP}
\left[O\mFILL_{(i,j)} + O\mFLUSH_{(i,j)}\right]}
{
\sum_{i=1}^{\mW}
\sum_{j=1}^{\mP}
S_{(i,j)}}
\leq \left[1 + (\deltaW + 1)\deltaP \right] \cdot \mW \cdot \mP,
\end{equation}
which completes the proof.
\end{proof}

\subsection{Concrete Classification Mechanisms}\label{sec:RC&S_concrete}
We now show various classify and randomly select mechanisms, which are tailored and optimized for different scenarios, depending on the profit and work values.

\paragraph{A linear classification} When a characteristic consists of a small set of potential values, we let each class include a single value of this characteristic. As a result, the competitive ratio of the algorithm is linearly depended upon the number of distinct potential values of the respective characteristic. For instance, when the set of potential work values
is small, we let each potential work value define a class. As a result, the competitive ratio of \SAMWP, implied by Theorem~\ref{thm:SAMWP}, is linearly depended upon the number of distinct work values, captured by the parameter $\mW$. Note that in this case we have $X\W_i = \set{w_i}$, implying that $\deltaW$, the max-to-min ratio of values within $X\W_i$, is 1.

\paragraph{A logarithmic classification} When the set of potential values of a characteristic is large, letting each value define a unique class results in a poor competitive ratio. Therefore, in such cases we use a logarithmic-scaled class partitioning as follows. We say that a packet $p$ is of a certain class (either work- or profit-) $i$ if its corresponding value is in the
interval
\begin{equation} \label{Eq:log_classes}
  X_i =\begin{cases}
              [1,2] & i=1 \\
              [2^{i-1}+1, 2^i] & i > 1.
            \end{cases}
\end{equation}

In particular, using the above partition packets into classes, we obtain that $\deltaP = \deltaW = 2$, $\mP = \log_2V$ and $\mW = \log_2W$. Using Theorem~\ref{thm:SAMWP}, we obtain the following corollary:

\begin{corollary}
\SAMWP\ is $O\left(\frac{M}{r} \log_2 W \log_2 V\right)$-competitive.
\end{corollary}

We note that if we know the number of distinct values for each characteristic and the values of $W$ and $V$, we can choose the appropriate classification scheme and have $\mW$ to be the minimum between $\log_2 W$, and the number of distinct work values; and have $\mP$ to be the minimum between $\log_2 V$, and the number of distinct profit values. Moreover, in any of our classification schemes, $\deltaW, \deltaP \leq 2$.

\section{Improved Algorithms}
\label{sec:improved_algorithms}
Algorithm \SAM\ selects a single class uniformly at random so that the characteristics of
packets on which it focuses, namely, $\Cs$-packets, differ by at most a constant factor. This gives the sense of ``uniformity'' of traffic within the class being targeted, which in turn reduces the variability of characteristics of packets on which the algorithm focuses.
However, in practice there are various cases where the strict decisions made by \SAM\ can be relaxed without harming its competitive performance guarantees. In practice, such relaxations actually allow obtaining a throughput far superior to
that of \SAM.
In what follows we describe such modifications, which we incorporate into our improved algorithm, \SAO, and prove that all our performance guarantees for \SAM\ still hold for \SAO.

\paragraph*{Class closure}
Recall the partitioning of packets into classes, described in Section~\ref{sec:RC&S}, namely,
 $\set{C_{(i,j)} | i=1,\ldots,\mW, j=1,\ldots,\mP}$. We let the {\em $(i,j)$-closure class} be defined as
 $C^*_{(i,j)}=\bigcup_{i'\leq i, j'\geq j} C_{(i',j')}$.

This definition means that the work of any packet in $C^*_{(i,j)}$
is within a ratio of at most $\deltaW$ of the work of any packet in $C_{(i,j)}$, and similarly for the profit of any packet in $C^*_{(i,j)}$.
Formally, for any packets $p \in C_{(i,j)}$ and  $p^* \in C^*_{(i,j)}$, $w(p^*) \leq \deltaW \cdot w(p)$ and $ v(p^*) \geq \frac{v(p)}{\deltaP}$.

We let \SAOWP\ denote the algorithm where the selected class $\Cs$ is chosen to be $C^*_{(i,j)}$, for some values of
$i,j$ chosen uniformly at random from the appropriate sets.
A simple substitution argument shows that thus picking $C^*_{(i,j)}$ by \SAOWP, instead of selecting $C_{(i,j)}$ as done
in \SAMWP, leaves the analysis detailed in Section~\ref{sec:SAMWP} intact.

\paragraph*{Fill during flush (pipelining)}
Algorithm \SAM\ was defined such that no arriving packets are ever accepted during the flush phase. This enables the partitioning of time into disjoint intervals (determined by \SAM's buffer being empty 
\textcolor{red}{at} the end of such an interval), and applying the comparison of the performance of \opt, on the one hand, and \SAM, on the other hand, independently for each interval. In practice, however, allowing the acceptance of packets during a flush phase
cannot harm the analysis, nor the actual performance, if this is done prudently:
packets which arrive during the flush phase are accepted according to the same priority suggested by the algorithm's behavior in the fill phase.
Furthermore, the algorithm stores in the buffer packets which arrive during the flush phase, but never schedules them for processing before it successfully transmits all $B$ packets that were stored in the buffer when it turned \hfull.

\paragraph*{Improved scheduling}
\SAM\ sorts the queued packets in $\CsK$-first order.
For simplicity of presentation, we assumed in Section~\ref{sec:algorithms} that within the set of $\CsK$-packets, as well as within the set of non-$\CsK$-packets, packets are  internally ordered by FIFO.
However, one may consider other approaches as well to performing such scheduling for each of these sets (while maintaining $\CsK$-first order between the sets). We consider specifically the following methods:
\begin{inparaenum}[(i)]
\item FIFO,
\item $W$-then-$V$, which orders packets by a non-decreasing order of remaining work, and breaks ties by non-increasing order of profit, and
\item non-increasing order of packet {\em effectiveness}, where the effectiveness of a packet is defined as its
    profit-to-work ratio.
\end{inparaenum}

We emphasize that the packet scheduled for processing during an admittance cycle remains a $U$-packet, which is selected uniformly at random from the arriving $U$-packets at this cycle.
All the \emph{non}-admitted $U$-packets, however, are located at the tail of the queue, thus representing the fact that their priority is lower than that of every known packet.
By applying different scheduling regimes, we obtain different flavors of \SAO.

The following Theorem shows that the performance of all flavors of \SAO\ is at least as good as the performance of \SAM.
\begin{theorem}\label{thm:SAOWP}
\SAOWP\ is
$
O
\left(
\left[
\frac{M}{r} + \deltaW \cdot \deltaP
\right]
\cdot
\mW \cdot \mP
\right)$
-competitive.
\end{theorem}

\begin{proof}

We first consider the effect of uniformly at random selecting a class closure, instead of selecting a specific class.
First, note that the proof of Lemma \ref{Cs_K_WP} also directly applies to \SAOWP, implying that $S^*_{\Cs^*} \geq \frac{r}{M} O\mFILL_{\Cs}$.
Furthermore, the arguments used in the proof of Theorem~\ref{thm:SAMWP} also apply to \SAOWP, and in particular \SAOWP\ satisfies Eq.~\ref{Eq:case01WV}, where we substitute in the denominator $S_{(i,j)}$ by $S^*_{(i,j)}$.

Consider next the effect of performing fill during flush.
In \SAOWP\ we accept packets also during the flush phase, but we never process any of these packets before all packets contributing to the algorithm being \hfull\ are transmitted, i.e., they are never processed before the flush phase is complete.
We enumerate the fill phases and the subsequent flush phases as follows: $P_{\FILL_1}, P_{\FLUSH_1}, P_{\FILL_2}, P_{\FLUSH_2}, \dots, P_{\FILL_n}, P_{\FLUSH_n}$, where $n \geq 1$. It should be noted that each such phase corresponds to a series of disjoint time intervals defined by the first cycle of the sequence of phases.
We further denote the $P_{\FLUSH_0}$ phase as an empty set of cycles, and in  case that the sequence ends by a fill phase, we also let $P_{\FLUSH_n}$ denote an empty set of cycles.
Similarly, we further define $P^*_{\FILL_i}, P^*_{\FLUSH_i}$, for the appropriate values of $i$, to denote the fill and flush phases corresponding to \SAOWP.

Denote the profit accrued by \SAMWP\ and \opt\ from packets which arrive during the $i^{th}$ fill phase by $S^{(P_{\FILL_i})}$ and $O^{(P_{\FILL_i})}$ respectively. Similarly, denote the profit of \SAMWP\ and \opt\ obtained from packets which arrive during the $i^{th}$ flush phase by $S^{(P_{\FLUSH_i})}$ and $O^{(P_{\FLUSH_i})}$, respectively.
Similarly, we let $S^{* (P^*_{\FILL_i})}$ and $S^{* (P^*_{\FLUSH_i})}$ indicate the profit of \SAOWP\ obtained from packets which arrive during its $i^{th}$ fill and flush phase, respectively.

Using this notation, we recall that, by the analysis of \SAMWP\ presented in Theorem~\ref{thm:SAMWP}
\begin{equation} \label{eq:P_old}
O^{(P_{\FILL_i})} + O^{(P_{\FLUSH_i})} \leq \left[ \frac{M}{r} + (\deltaW + 1) \deltaP \right]
\mW  \cdot \mP \cdot S^{(P_{\FILL_i})}
\end{equation}
for every $i=1,\ldots,n$.

This induces an implicit mapping $\phi$ of
the units of profit obtained from
$\Cs$-packets accepted by \opt\ during $P_{\FILL_i} \cup P_{\FLUSH_i}$ to the units of profit obtained from $\Cs$-packets accepted by \SAMWP\ during $P_{\FILL_i}$ (either known, or unknown that were parsed), such that every unit of profit obtained by \SAMWP\ has at most $\left[\frac{M}{r} + (\deltaW + 1) \deltaP \right] \mW  \cdot \mP$ units of profit mapped to it.

A key observation is noting that the image of mapping $\phi$ is essentially the profit attained from the set of $\Cs$-packets contributing to the algorithm being \hfull\ at the end of the corresponding fill phase.

As \SAOWP\ may accept packets during flush, in the beginning of the subsequent fill phase the buffer of \SAOWP\ may not be empty. In particular, there could be $\Cs$-packets accepted during the recent flush phase that are stored in the buffer.
However, none of these packets have any \opt\ packets mapped to them.
It follows that these packets can contribute to \SAOWP\ becoming \hfull\ in the new fill phase, and any profit implicitly mapped to the profit of these packets by $\phi$ would correspond to packets arriving during the new fill phase, or its subsequent
flush phase.
The implicit mapping is depicted in Fig.~\ref{fig:fill_during_flush}, along with the difference between the mapping arising from the behavior of \SAMWP\ (visualized above the time axis), and the mapping arising from the behavior of \SAOWP (visualized below the time axis). Note that the fill and flush phases of both algorithms need not be synchronized, since \SAOWP\ can potentially become \hfull\ ``faster'' than \SAMWP.

\begin{figure*}[t]
\centering
\includegraphics[width=1.0\textwidth]{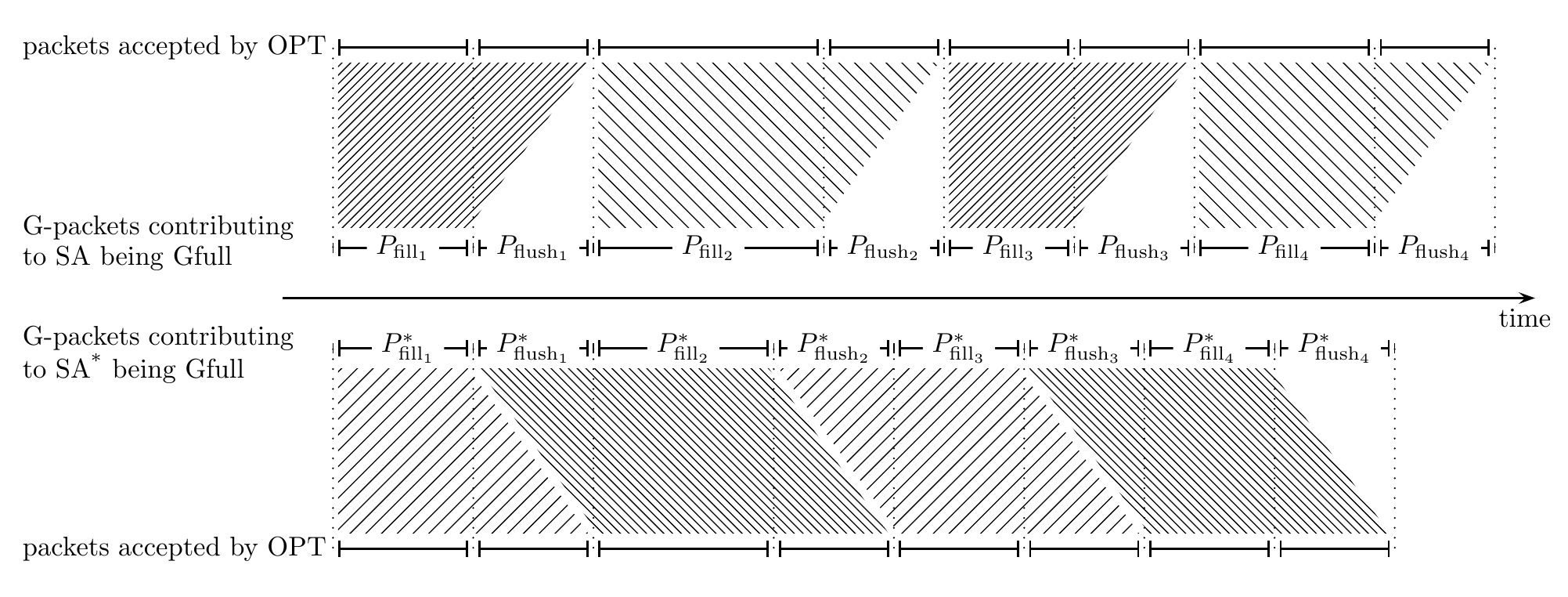}
\caption[Visualization of the mappings induced by the analysis of \SAMWP\ and \SAOWP]{Visualization of the mappings induced by the analysis of \SAMWP\ and \SAOWP, for the first 4 fill and flush phases. The fill and flush phases of \SAMWP\ are denoted $P_{\FILL_i}$ and $P_{\FLUSH_i}$, respectively, whereas the fill and flush phases of \SAOWP\ are denoted $P^*_{\FILL_i}$ and $P^*_{\FLUSH_i}$, respectively.
The top part shows the mapping of the profit corresponding to packets accepted by \opt\ along time, to the profit corresponding to $\Cs$-packets accepted by \SAMWP\ during the fill phase (since \SAMWP\ does not accept any packets during the flush phase).
The bottom part shows the induced mapping of the profit obtained by packets accepted by \opt\ along time to the profit of $\Cs$-packets accepted by \SAOWP\ during both the preceding flush phase, and the current fill phase.}
\label{fig:fill_during_flush}
\end{figure*}

It follows that Eq.~\ref{eq:P_old} now translates to
\begin{equation} \label{eq:P_new}
\begin{array}{l}
O^{(P^*_{\FILL_i})} + O^{(P^*_{\FLUSH_i})} \leq \\
\quad \quad \quad \left[ \frac{M}{r} + (\deltaW + 1) \deltaP \right] \mW  \cdot \mP \cdot \left[S^{* (P^* _{\FLUSH_{i-1}})} + S^{* (P^*_{\FILL_i})}\right]
\end{array}
\end{equation}
for every $i=1,\ldots,n$.
Summing over all $i=1,\ldots,n$, we obtain that the competitive ratio guarantee for \SAOWP\ is the same as that for \SAMWP.

Lastly, the analysis of \SAMWP\ does not assume any specific scheduling rule to be applied, as long as the $\Cs\K$-first order rule is maintained.
Thus, our competitive ratio guarantee is independent of the specific ordering within the set of $\CsK$-packets, as well as within the set of non-$\CsK$-packets.
\end{proof}

We study the performance of the various flavors of \SAO\ in Section~\ref{sec:simulations}.

\section{Practical Implementation}\label{sec:Practical_Implementation}
While presenting our basic algorithm in Section~\ref{sec:algorithms}, we assumed for simplicity that the values of $W$
and $V$ -- the maximal work and profit per packet, respectively -- are known to the algorithm in advance.
We now show how to relax these assumptions without harming the performance of our algorithms.

We refer to an algorithm implementation that does
not know these values in advance as a {\em values-oblivious} algorithm, and to an algorithm implementation that knows the values of $W$ and $V$ in advance as a {\em values-aware} algorithm. We will show that a values-oblivious algorithm can obtain a performance which is no worse than that of a values-aware algorithm, even if the values-aware algorithm knows not only $W$ and $V$, but also the concrete classes in which packets will arrive.

Our implementation of a values-oblivious algorithm is based on an application of reservoir sampling \cite{Reservoir} on
classes revealed during packet arrivals, as we will detail shortly.
A new class is revealed either due to the arrival of a $K$-packet $p$, or due to a $U$-packet $q$ being parsed,
corresponding to a class previously unknown to the algorithm. We call such an event an {\em uncovering of a new
class}.

The values-oblivious algorithm implementation performs the following alongside all decisions made by the values-aware algorithm:
Before the arrival sequence begins we initiate a counter $N$ of known classes to be $N=0$.
Upon the uncovering of a new class at $t$ the algorithm increments $N$ by one (to reflect the updated number of
known classes), and replaces the previously selected class with the new class with probability $1/N$.

As the above procedure essentially performs a reservoir sampling on the collection of classes known to the algorithm,
it essentially implements the selection of a class uniformly at random among all {\em a posteriori} known
classes~\cite{Reservoir}.

It therefore follows that the distribution of the packets corresponding to the \redtext{eventually} selected class (after the sequence ends) handled by the values-oblivious algorithm is identical to the distribution of the packets handled by the values-aware algorithm. Therefore the expected performance of the values-oblivious algorithm is lower bounded by the expected performance of the values-aware algorithm.
We note that the implementation of the values-oblivious algorithm can be applied to any of the variants described in our previous sections.

\section{Simulation Study}
\label{sec:simulations}
In this section we present the results of our simulation study intended to validate our theoretical results, and
provide further insight into our algorithmic design.
Our choice of distributions for the parameters of the traffic characteristic enables us to evaluate our algorithms' performance in a wide range of settings. These choices, as we show in the sequel, are also motivated by the properties of real-world traffic.

\subsection{Simulation Settings}\label{sec:simulation_set}
We simulate a single queue in a gateway router which handles a bursty arrival sequence of packets with high work
requirements (corresponding, e.g., to IPSec packets, requiring AES encryption/decryption) as well as packets with
low work requirements (such as simple IP packets requiring merely IPv4-trie processing).
Arriving packets also have arbitrary profits, modeling various QoS levels.

Our traffic is generated by a Markov modulated Poisson process (MMPP)
with two states, LOW and HIGH, such that the burst during the HIGH state
generates an average of 10 packets per cycle, while the LOW state generates an average of only $0.5$ packet per
cycle. The average duration of LOW-state periods is a factor $W$ longer than the average duration of HIGH-state
periods. This is targeted at allowing some traffic arriving during the HIGH-state to be drained during the LOW-state.

In our simulations, we do not deterministically bound the maximum number, $M$, of $U$-packets arriving in a cycle, but rather control the expected intensity of $U$-packets by letting each arriving packet be a $U$-packet with some probability $\alpha \in [0,1]$. We thus obtain that the expected number of $U$-packets per cycle during the HIGH state is $10 \alpha$.

In real-life scenarios, the maximum work, $W$, required by a packet, is highly implementation-depended. It depends on the specific hardware, processing elements, and software modules.
However, several works that investigated the required work on typical tasks~\cite{ramaswamy09, salehi09, salehi12} indicate that $W$ is two orders of magnitude larger than the work required for doing an IPv4-trie search or classification of a packet. We refer to IPv4-trie search or classification of a packet as the baseline unit of work, captured by our notion of ``parsing''. We therefore set the maximum work required by a packet to $W=256$ throughout this section.
As the potential set of characteristics is large, we use a logarithmic classification scheme (recall Section~\ref{sec:RC&S_concrete}).

Determining the maximum profit, $V$, associated with a packet, is a challenging task. This value depends both on implementation details, as well as on proprietary commercial and business considerations. In order to have a diverse set of values, which model distinct QoS requirements, we set the maximum profit associated with a packet to $V=16$ throughout this section.

The values $W=256$ and $V = 16$ imply a total of $8 \cdot 4 = 32$ potential classes for the algorithm to select from, respectively.
The value of each characteristic for each packet is drawn from an approximation of a Pareto-distribution as follows. First, we randomly generate numbers, following a Pareto-distribution. Next, numbers are rounded, to get integer values. Finally, for disallowing values above the maximum (256 for work values and 16 for profit values), all the cases where the randomly generated values were above the maximum were truncated, namely, treated as if the generated value was exactly the maximal value.
The averages and standard deviations of the values obtained after this generation process are 17.97 and 22.22 for packet work, and 3.66 and 3.20 for packet profit.
The schematic probability distribution function of the characteristics values is depicted in Fig.~\ref{fig:PDF_of_values}.
Note the spike at its maximum, due to the truncation described above.
Unless stated otherwise, we assume that $B=10$, $r = 1$ and each arriving packet is a $U$-packet with probability $\alpha = 0.3$. We thus obtain that the expected number of $U$-packets arriving during the HIGH state is $0.3 \cdot 10 = 3$ per cycle.

\begin{figure}[t]
\centering
\begin{tikzpicture}
	\begin{axis}[
		ybar,
		bar width=2ex,		
		ticks=none,
		xlabel=Value,
        legend style = {font=\figfontsize},
        label style={font=\figfontsize},
        tick label style={font=\figfontsize},
		xticklabels=empty,
		ylabel=PDF,
		yticklabels=empty]
	\addplot[fill=gray] plot coordinates {
		(1,0.2136)
		(2,0.2779)
		(3,0.1659)
		(4,0.1053)
		(5,0.0635)
		(6,0.0440)
		(7,0.0303)
		(8,0.0219)
		(9,0.0173)
		(10,0.0113)
		(11,0.0080)
		(12,0.0063)
		(13,0.0054)
		(14,0.0047)
		(15,0.0030)
		(16,0.0214)
	};
	\end{axis}
\end{tikzpicture}	
\caption{Probability distribution function of the characteristics values}
\label{fig:PDF_of_values}
\end{figure}
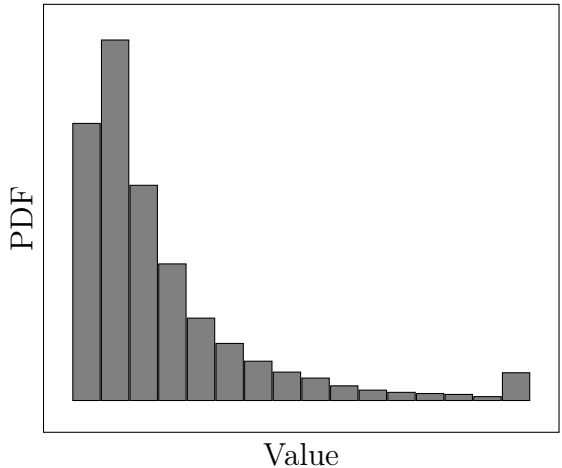

As a benchmark which serves as an upper bound on the optimal performance possible, we consider a relaxation of the offline problem as a knapsack problem. Arriving packets are viewed as items, each with its size (corresponding to the packet's work) and value (corresponding to the packet's profit). The allocated knapsack size equals the number of time slots during which packets arrive. The goal is to choose a highest-value subset of items that fits within the given knapsack size.
This is indeed a relaxation of the problem of maximizing throughput during the arrival sequence in the offline setting, since the knapsack problem is not restricted by any finite buffer size during the arrival sequence, nor by the arrival time of packets (e.g., it may ``pack'' packets even before they arrive).

We employ the classic 2-approximation greedy algorithm for solving the knapsack problem~\cite{williamson11design}, and
use its performance as an approximate upper bound on the performance of \opt.
For considering the additional profit which \opt\ may gain from packets which reside in its buffer at the end of the arrival sequence, we simply allow the offline approximation an additional throughput of $BV$ for free, which is an upper bound on the benefit it may achieve after the arrival sequence ends.

We compare the performance of studied algorithms by evaluating their {\em performance ratio}, which is the ratio between the algorithm's performance and that of our approximate upper bound on the performance of \opt.

We compare the performance of the following algorithms:
\begin{enumerate}
\item {\em FIFO}: A simple greedy non-preemptive FIFO discipline that simply accepts packets and processes each
    packet until completion, regardless of its required work or value.
\item {\em \SAM}: Algorithm \SAM, described in Section \ref{sec:algorithms}.
\item {\em \SAO\ FIFO}: Algorithm \SAO\ where priority ties are broken by FIFO order.
\item {\em \SAO\ $W$-Then-$V$}: Algorithm \SAO\ where priority ties are broken in non-decreasing order of remaining
    work, and further ties are broken in non-increasing order of profit. This variant is denoted by \SAO $W-V$ in Figures~\ref{fig:g_W}-\ref{fig:r}.
\item {\em \SAO\ EFFECT}: Algorithm \SAO\ where priority ties are broken in non-increasing order of their
    profit-to-work ratio.
\end{enumerate}

We recall that all the flavors of \SAO\ listed above maintain a $\CsK$-first order, and differ only in the internal
ordering \emph{within} each set (namely, within the set of $\CsK$-packets, as well as within the set of {\em
non}-$\CsK$-packets).
%Furthermore, all flavors locate all the non-admitted unknown packet in the bottom of the queue.

All flavors of \SAO\ described above employ the class-closure and the fill-during-flush modifications defined in
Section~\ref{sec:improved_algorithms}.
For each choice of parameters, we show the average of running 100 independently-generated traces of 10K packets each. In all our simulations the standard deviation was below 0.035.

\subsection{Simulation Results}
Figures~\ref{fig:g_W}-\ref{fig:r} show the results of our simulation study.

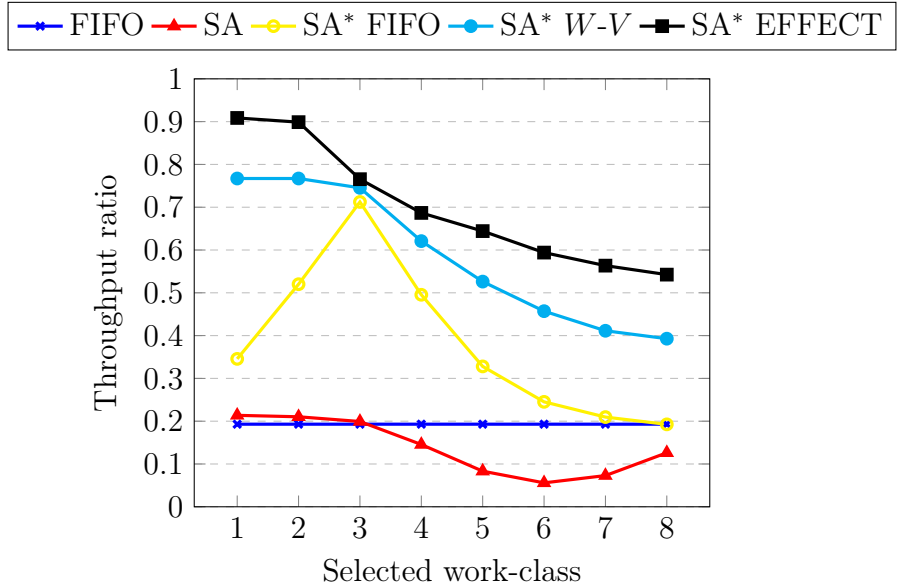
\begin{figure}[t]
\centering
\begin{tikzpicture}
	\begin{axis}[
		legend style={at={(0.5,1.18)},anchor=north,legend columns=-1,font=\figfontsize},
	    label style={font=\figfontsize},
	    tick label style={font=\figfontsize},
		ylabel=Throughput ratio,
		xlabel=Selected work-class,
		ytick={0,0.1,0.2,0.3,0.4,0.5,0.6,0.7,0.8,0.9,1},
		ymin=0,
		ymax=1,
	    ymajorgrids=true,
	    grid style=dashed,
        every axis plot/.append style={very thick},
		xtick=data]
	\addplot[color=blue,mark=x] coordinates {
		(1,0.1930)(2,0.1930)(3,0.1930)(4,0.1930)(5,0.1930)(6,0.1930)(7,0.1930)(8,0.1930)
	};
	\addlegendentry{FIFO}
	\addplot[color=red,mark=triangle*] coordinates {
		(1,0.2137)(2,0.2103)(3,0.1993)(4,0.1455)(5,0.0832)(6,0.0559)(7,0.0729)(8,0.1262)
	};
	\addlegendentry{SA}
	\addplot[color=yellow,mark=o] coordinates {
		(1,0.3456)(2,0.5201)(3,0.7123)(4,0.4954)(5,0.3280)(6,0.2452)(7,0.2096)(8,0.1926)
	};
	\addlegendentry{SA$^*$ FIFO}
	\addplot[color=cyan,mark=*] coordinates {
		(1,0.7671)(2,0.7671)(3,0.7453)(4,0.6208)(5,0.5262)(6,0.4571)(7,0.4113)(8,0.3928)
	};
	\addlegendentry{SA$^*$ $W$-$V$}
	\addplot[color=black,mark=square*] coordinates {
		(1,0.9086)(2,0.8988)(3,0.7655)(4,0.6867)(5,0.6442)(6,0.5940)(7,0.5635)(8,0.5425)
	};
	\addlegendentry{SA$^*$ EFFECT}
	\end{axis}
\end{tikzpicture}	
\caption{Effect of chosen work-class $i^*$}% \in \set{1,\ldots,8}$}
\label{fig:g_W}
\end{figure}

First we note that \SAM\ exhibits a very low performance ratio, similar to that of a simple FIFO (which disregards packets parameters altogether). This is due to the fact that \SAM\ focuses only on a specific class, which consists of a relatively small part of the input, and it thus
spends processing cycles on packets that would not be eventually transmitted.

For the variants of \SAO\ we consider, in all simulations the best scheduling policy is by non-increasing
effectiveness, followed by employing the $W$-then-$V$ approach. FIFO scheduling, in spite of it being simple and
attractive, comes in last in all scenarios.
This behavior is explained by the fact that both former scheduling policies in \SAO\ clear the buffer more
effectively once it is \hfull. The latter FIFO scheduling approach clears the buffer in an oblivious manner, and
therefore doesn't free up space for new arrivals fast enough.
We now turn to discuss each of the scenarios considered in our study.

\subsection{The Effect of Selected Class}
Our first set of results sheds light on the effect of the class selected by an algorithm on its performance.
Fig.~\ref{fig:g_W} shows the results where the selected profit-class $j^*$ is 1, which makes \SAO\ allow all profits, and the choice of work-class $i^*$ varies.
The most interesting phenomena are exhibited by \SAO\ FIFO. Its performance is very poor if the work-class may
contain packets requiring very little work. This is due to the fact that only a small fraction of the traffic
requires this little work, and the algorithm scarcely arrives at being \hfull. As a consequence, the algorithm
handles many low-priority packets, which are handled in FIFO order, giving rise to far-from-optimal decisions. The
algorithm steadily improves up to some point, and then its performance deteriorates fast as it assigns high-priority
to packets with increasingly higher processing requirements. In this case, the algorithm becomes \hfull\ too
frequently, and allocates many processing cycles to low-effectiveness packets.
The maximum performance is achieved for $i^*=3$,
which implies that the algorithm flushes whenever its buffer is filled up with packets whose work is at most
$2^{i^*}=8$. This value suffices to allow the algorithm to prioritize a rather large portion of the arrivals
(recalling the Pareto distribution governing packet work-values), while ensuring the processing toll of
high-priority packet is not too large.
This strikes a (somewhat static) balance between the amount of work required by a packet, and its expected potential
profit.

The other variants of \SAO\ exhibit a gradually decreasing performance, due to
their higher readiness to compromise over the required work of packets they deem as high-priority traffic.
\SAM\ shows a similar performance deterioration, for a similar reason, when the selected work-class $i^*$ is increased from 1 up to 6. However, when increasing $i^*$ above 6, \SAM's performance increases again. This improvement is explained by the fact that, due to the Pareto-distribution of the work values, the number of packets that belong to each work-class rapidly diminishes when switching to work-class indices closest to the maximum of 8;
recall that \SAM\ over-prioritizes only packets which belong to a single randomly selected class, i.e., \SAM\ does not employ the class closure optimization (described in Section~\ref{sec:improved_algorithms}).
In such a case, \SAM\ is coerced into processing also packets which do not belong to the selected class -- namely, packets with {\em lower} work -- which somewhat compensates for the poor choice of the work-class.
We verified this explanation by additional simulations (not shown here), in which the work-class of packets was chosen from the uniform distribution. In such a case, where there is an abundance of packets from every possible work-class, the performance of \SAM\ consistently degrades with the increase of $i^*$, which implies a poorer choice of work-class.

\begin{figure}[t]
\centering
\begin{tikzpicture}
	\begin{axis}[
		legend style={at={(0.5,1.18)},anchor=north,legend columns=-1,font=\figfontsize},
	    label style={font=\figfontsize},
	    tick label style={font=\figfontsize},
		ylabel=Throughput ratio,
		xlabel=Selected profit-class,
		ytick={0,0.1,0.2,0.3,0.4,0.5,0.6,0.7,0.8,0.9,1},
		ymin=0,
		ymax=1,
	    ymajorgrids=true,
	    grid style=dashed,
	    every axis plot/.append style={very thick},
		xtick=data]
	\addplot[color=blue,mark=x] coordinates {
	    (1,0.1919)(2,0.1919)(3,0.1919)(4,0.1919)
	};
	\addlegendentry{FIFO}
	\addplot[color=red,mark=triangle*] coordinates {
		(1,0.1255)(2,0.1554)(3,0.1686)(4,0.1803)
	};
	\addlegendentry{SA}
	\addplot[color=yellow,mark=o] coordinates {
		(1,0.1922)(2,0.2960)(3,0.4248)(4,0.4997)
	};
	\addlegendentry{SA$^*$ FIFO}
	\addplot[color=cyan,mark=*] coordinates {
		(1,0.3910)(2,0.4486)(3,0.5201)(4,0.6358)
	};
	\addlegendentry{SA$^*$ $W$-$V$}
	\addplot[color=black,mark=square*] coordinates {
		(1,0.5427)(2,0.5137)(3,0.5372)(4,0.6557)
	};
	\addlegendentry{SA$^*$ EFFECT}
	\end{axis}
\end{tikzpicture}
\caption{Effect of chosen profit-class $j^*$}% \in \set{1,\ldots,4}$}
\label{fig:g_V}
\vspace{\vspacebelowcaption}
\end{figure}
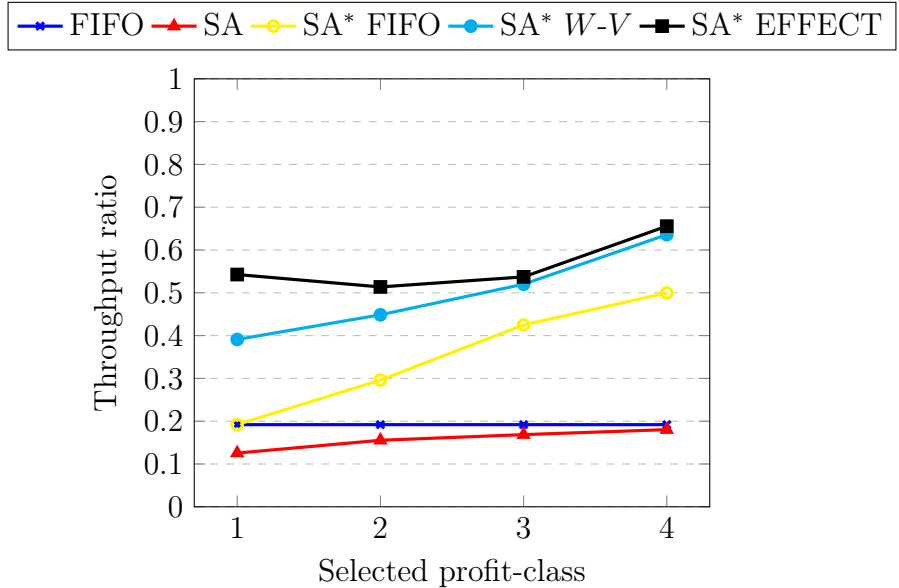

Similar phenomena are exhibited in Fig.~\ref{fig:g_V}, where we consider the effect of the profit-class $j^*$
selected by an algorithm on its performance. In this set of simulations all work-values were allowed (i.e., the
selected work-class is 8).
In this scenario the performance
of all algorithms improves as the selected profit-class index increases, and the
algorithms are able to better restrict their focus on high profit packets as the packets receiving high-priority.
We note the fact that \SAO\ FIFO and regular FIFO have a matching performance in the case the selected profit-class
is 1, since in this case \SAO\ FIFO is identical to plain FIFO (since it simply indiscriminately accepts and processes all incoming packets in FIFO order).

In subsequent results described hereafter, we fix both the work-class and the profit-class to be 3, which represents
a mid-range class for both the profit and the work.

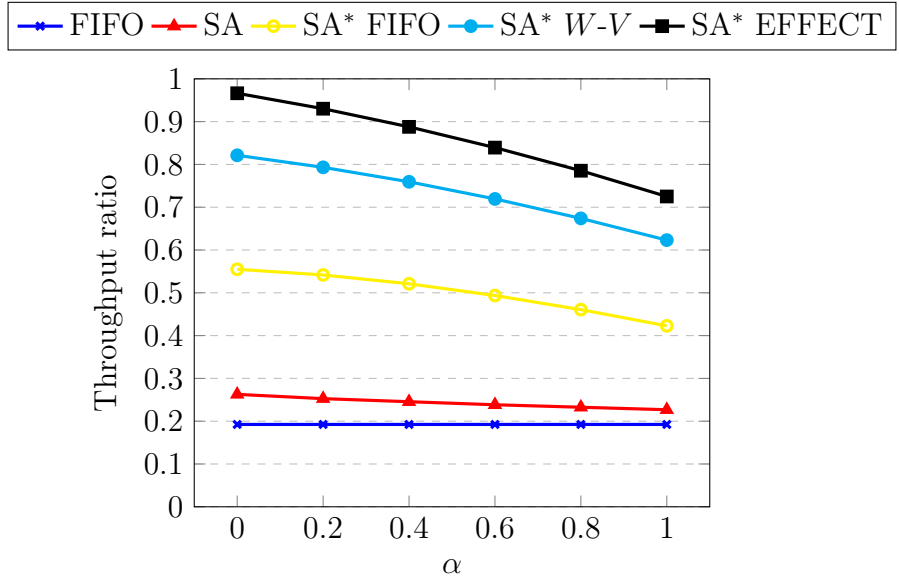
\begin{figure}[h!]
\centering
\begin{tikzpicture}
	\begin{axis}[
		legend style={at={(0.5,1.18)},anchor=north,legend columns=-1,font=\figfontsize},
	    label style={font=\figfontsize},
	    tick label style={font=\figfontsize},
		ylabel=Throughput ratio,
		xlabel=$\alpha$,
		ytick={0,0.1,0.2,0.3,0.4,0.5,0.6,0.7,0.8,0.9,1},
		ymin=0,
		ymax=1,
	    ymajorgrids=true,
	    grid style=dashed,
	    every axis plot/.append style={very thick},
		xtick=data]
	\addplot[color=blue,mark=x] coordinates {
	    (0,0.1924)(0.2,0.1924)(0.4,0.1924)(0.6,0.1924)(0.8,0.1924)(1,0.1924)
	};
	\addlegendentry{FIFO}
	\addplot[color=red,mark=triangle*] coordinates {
		(0,0.2627)(0.2,0.2527)(0.4,0.2455)(0.6,0.2384)(0.8,0.2326)(1,0.2268)
	};
	\addlegendentry{SA}
	\addplot[color=yellow,mark=o] coordinates {
		(0,0.5550)(0.2,0.5417)(0.4,0.5211)(0.6,0.4937)(0.8,0.4607)(1,0.4228)
	};
	\addlegendentry{SA$^*$ FIFO}
	\addplot[color=cyan,mark=*] coordinates {
		(0,0.8213)(0.2,0.7932)(0.4,0.7594)(0.6,0.7194)(0.8,0.6739)(1,0.6231)
	};
	\addlegendentry{SA$^*$ $W$-$V$}
	\addplot[color=black,mark=square*] coordinates {
		(0,0.9663)(0.2,0.9302)(0.4,0.8878)(0.6,0.8393)(0.8,0.7853)(1,0.7251)
	};
	\addlegendentry{SA$^*$ EFFECT}
	\end{axis}
\end{tikzpicture}	
\caption{Effect of expected number of $U$-packets during the HIGH state}
\label{fig:M}
\end{figure}

\begin{figure}[h!]
\centering
\begin{tikzpicture}
	\begin{axis}[
		legend style={at={(0.5,1.18)},anchor=north,legend columns=-1,font=\figfontsize},
	    label style={font=\figfontsize},
	    tick label style={font=\figfontsize},
		ylabel=Throughput ratio, 
		xlabel=$r$,
		ytick={0,0.1,0.2,0.3,0.4,0.5,0.6,0.7,0.8,0.9,1},
		ymin=0,
		ymax=1,
	    ymajorgrids=true,
	    grid style=dashed,
	    every axis plot/.append style={very thick},
		xtick=data]
	\addplot[color=blue,mark=x] coordinates {
	    (0,0.1933)(0.2,0.1935)(0.4,0.1944)(0.6,0.1920)(0.8,0.1930)(1,0.1931)
	};
	\addlegendentry{FIFO}
	\addplot[color=red,mark=triangle*] coordinates {
		(0,0.1933)(0.2,0.1974)(0.4,0.2078)(0.6,0.2144)(0.8,0.2202)(1,0.2289)
	};
	\addlegendentry{SA}
	\addplot[color=yellow,mark=o] coordinates {
		(0,0.1933)(0.2,0.2251)(0.4,0.2738)(0.6,0.3242)(0.8,0.3746)(1,0.4260)
	};
	\addlegendentry{SA$^*$ FIFO}
	\addplot[color=cyan,mark=*] coordinates {
    	(0,0.1933)(0.2,0.2582)(0.4,0.4106)(0.6,0.5089)(0.8,0.5781)(1,0.6268)
	};
	\addlegendentry{SA$^*$ $W$-$V$}
	\addplot[color=black,mark=square*] coordinates {
		(0,0.1933)(0.2,0.2642)(0.4,0.4467)(0.6,0.5727)(0.8,0.6653)(1,0.7286)
	};
	\addlegendentry{SA$^*$ EFFECT}
	\end{axis}
\end{tikzpicture}	
\caption{Effect of admittance probability of $U$-packets $r$}
\label{fig:r}
\end{figure}
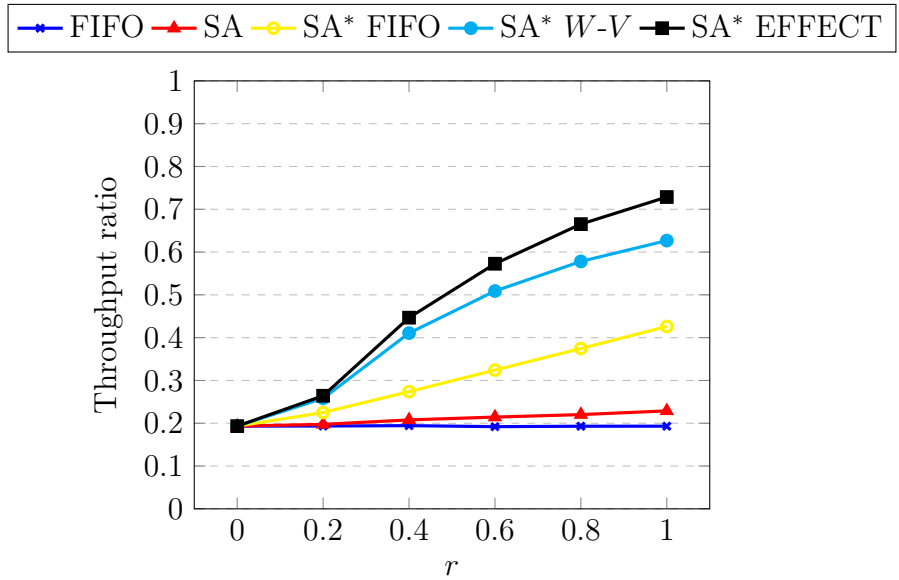

\subsection{The Effect of Missing Information}
Fig.~\ref{fig:M} illustrates the performance ratio of our algorithms as a function of the expected number of $U$-packets arriving during the HIGH state, where we vary the value of $\alpha$ from 0 to 1. This provides further insight as to the performance of each algorithm as a function of the intensity of unknown packets. We recall that for our choice of parameters, the values of $\alpha$ translate to having the expected number of unknown packets per cycle during the HIGH state vary from 0 to 10.
As one could expect, the performance ratio of \SAM\ and of all versions of \SAO\ degrades as the amount of uncertainty increases.

Finally, we study the intensity of exploring unknown packets, as depicted by the choice of parameter $r$ which
determines whether a cycle is an admittance cycle or not.
The results depicted in Fig.~\ref{fig:r} consider the case of high uncertainty, where $\alpha=1$, that is, all arriving packets are unknown.

Observe first the special case where $r=0$, which represents an extreme case, in which, although all arriving packets are unknown, our algorithms do not explore any new packets, and actually degenerate to a simple FIFO, and therefore exhibit identical performance.
Increasing the admittance probability $r$, however, yields a steady increase in performance, albeit with diminishing returns.
Similar results were obtained also when some of the packets are known, but with smaller marginal benefits. These results coincide with our analytic results, which further validate our algorithmic approach.

\section{Discussion}\label{sec:buf:conclusions}
This chapter introduces the problem of managing buffers where traffic has unknown characteristics, namely required processing
and profits.
We show lower bounds on the competitive ratio of any online algorithm for the problem.
We define several algorithmic concepts targeted at such settings, and develop several algorithms that
follow our suggested prescription.
Our theoretical analysis shows that the competitive ratio of our algorithms is not far from the best competitive ratio any online algorithm can achieve.
We validate the performance of our algorithms via simulation which further serves to elucidate our design criteria.
Our work can be viewed as a first step in developing fine-grained algorithms handling scenarios of limited knowledge
in networking environments for highly heterogeneous traffic.

Our work gives rise to a multitude of open questions, including:
\begin{inparaenum}[(i)]
\item closing the gap between our lower and upper bound for the problem,
\item applying our proposed approaches to other limited knowledge networking environments, and
\item devising additional algorithmic paradigms for handling limited knowledge in heterogeneous settings.
\end{inparaenum}

\newpage
\chapter{Access Strategies in Network Caching}\label{sec:BF}

\section{Problem Overview}\label{sec:BF:intro}
Having access to multiple network connected \emph{data stores} is common in modern network settings such as 5G in-network caching~\cite{5GNetworks,5GNetworks2}, content delivery networks (CDN)~\cite{CDN_OceanStore,CDN_AdaptSize}, information centric networking~\cite{ICN,ICN2}, wide-area networks~\cite{SummaryCache}, as well as in any multi data center Internet company.
Data stores can be cache enabled network devices, memory layers within a server, virtual machines, physical hosts, remote data centers or any combination of the above examples.
In such settings, each data store acts as a network cache by holding a potentially overlapping fraction of the entire data that may be accessed by applications and services hosted in the network.

\begin{figure}[h!]
	\centering
	\subfloat[\label{fig:bf_toy_example}
	The client is looking for item $x$ and needs to select which data stores to access. Data stores provide an indication ($I(x)$) if they store $x$. The grayed content ($C(x)$) indicates if they actually store $x$. Notice the false positive in Data store 1. The client pays the cost for the selected data stores, and in case of a failure to find $x$ in any of the accessed data stores, it incurs a miss penalty.
	]{
		\includegraphics[width=\columnwidth]{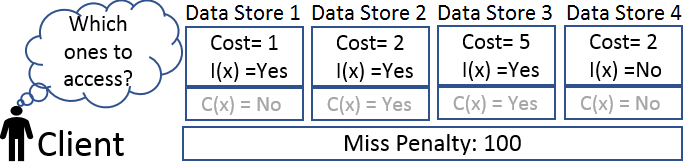}
	}
	\\
	\subfloat[\label{fig:example}An example of the average access cost of different strategies (lower is better), when varying the cache hit ratio. The number of data stores here is $20$, the access cost to each of them is $1$ while the miss penalty is $100$ and the false positive ratio is 0.02.]{
		\resizebox{\columnwidth}{0.45\columnwidth}{
\begin{tikzpicture}
    [spy using outlines={rectangle, connect spies}]
	\begin{axis}[
		width=\columnwidth,
		height=5cm,
		legend style={at={(0.5,1.23)},anchor=north,legend columns=-1,font=\figfontsize},
	    label style={font=\figfontsize},
	    tick label style={font=\figfontsize},
		ylabel=Expected Access Cost,
		ylabel near ticks,
		xlabel= Per Data Store Hit Ratio,
		xlabel near ticks,
		xtick={0,0.2,0.4,0.6,0.8,1},
		ymin=0,
		ymax=110,
        xmin = 0,
        xmax =1,
	    ymajorgrids=true,
	    grid style=dashed,
		]
	\addplot[color=red,mark=triangle*] coordinates {
		(0.00000, 100.00000)(0.05000, 36.49011)(0.10000, 13.03609)(0.15000, 4.83719)(0.20000, 2.14139)(0.25000, 1.31395)(0.30000, 1.07899)(0.35000, 1.01794)(0.40000, 1.00362)(0.45000, 1.00064)(0.50000, 1.00009)(0.55000, 1.00001)(0.60000, 1.00000)(0.65000, 1.00000)(0.70000, 1.00000)(0.75000, 1.00000)(0.80000, 1.00000)(0.85000, 1.00000)(0.90000, 1.00000)(0.95000, 1.00000)(1.00000, 1.00000)
	};
	\addlegendentry{Perfect Indicators}	
\addplot[color=purple,mark=o] coordinates {
			(0.00000, 100.00000)(0.05000, 37.22083)(0.10000, 14.26893)(0.15000, 6.39020)(0.20000, 3.54726)(0.25000, 2.59345)(0.30000, 2.26679)(0.35000, 2.14348)(0.40000, 2.08785)(0.45000, 2.05745)(0.50000, 2.03852)(0.55000, 2.02593)(0.60000, 2.01731)(0.65000, 2.01135)(0.70000, 1.84986)(0.75000, 1.66225)(0.80000, 1.49751)(0.85000, 1.35170)(0.90000, 1.22173)(0.95000, 1.10515)(1.00000, 1.00000)
		};
\addlegendentry{FPO}

\addplot[color=cyan,mark=x,]    coordinates {
	(0.00000, 100.40000)(0.05000, 37.22859)(0.10000, 14.51767)(0.15000, 7.21595)(0.20000, 5.47292)(0.25000, 5.61712)(0.30000, 6.35979)(0.35000, 7.27812)(0.40000, 8.24366)(0.45000, 9.22064)(0.50000, 10.20010)(0.55000, 11.18001)(0.60000, 12.16000)(0.65000, 13.14000)(0.70000, 14.12000)(0.75000, 15.10000)(0.80000, 16.08000)(0.85000, 17.06000)(0.90000, 18.04000)(0.95000, 19.02000)(1.00000, 20.00000)
};
\addlegendentry{EPI}

\addplot[color=black,mark=triangle,]coordinates {
	(0.00000, 100.33239)(0.05000, 45.63952)(0.10000, 23.05151)(0.15000, 13.47797)(0.20000, 9.11240)(0.25000, 6.85799)(0.30000, 5.50896)(0.35000, 4.59281)(0.40000, 3.91497)(0.45000, 3.38653)(0.50000, 2.96085)(0.55000, 2.61003)(0.60000, 2.31579)(0.65000, 2.06545)(0.70000, 1.84986)(0.75000, 1.66225)(0.80000, 1.49751)(0.85000, 1.35170)(0.90000, 1.22173)(0.95000, 1.10515)(1.00000, 1.00000)
};
\addlegendentry{CPI}

\addplot[color=blue,mark=square,] coordinates {
(0.00000, 100.00000)(0.05000, 55.84859)(0.10000, 32.15767)(0.15000, 23.31134)(0.20000, 18.39805)(0.25000, 15.16764)(0.30000, 12.82475)(0.35000, 11.07119)(0.40000, 9.67962)(0.45000, 8.52244)(0.50000, 7.56250)(0.55000, 6.83038)(0.60000, 6.02400)(0.65000, 5.50063)(0.70000, 4.81000)(0.75000, 4.39063)(0.80000, 3.80000)(0.85000, 3.33750)(0.90000, 3.00000)(0.95000, 2.25000)(1.00000, 1.00000)
};
\addlegendentry{No Indicators}

\coordinate (spypoint) at (axis cs:0.4,9.5);
\coordinate (magnifyglass) at (axis cs:0.65,77);
	\end{axis}

%using the 'spy' to magnify a piece of the picture
  \spy [green, width=6.2cm, height=2.2cm, magnification=1.7] on (spypoint)
             in node [fill=white] at (magnifyglass);

\end{tikzpicture}
}
	}
	\caption{Motivation for the access strategy problem.}
\end{figure}
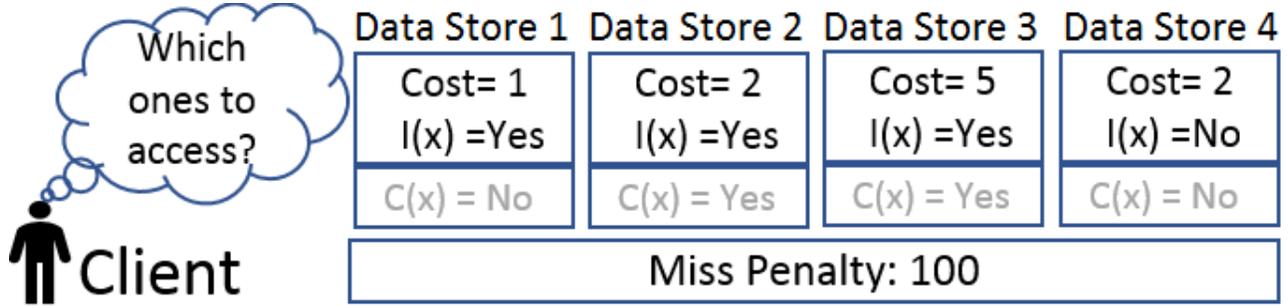
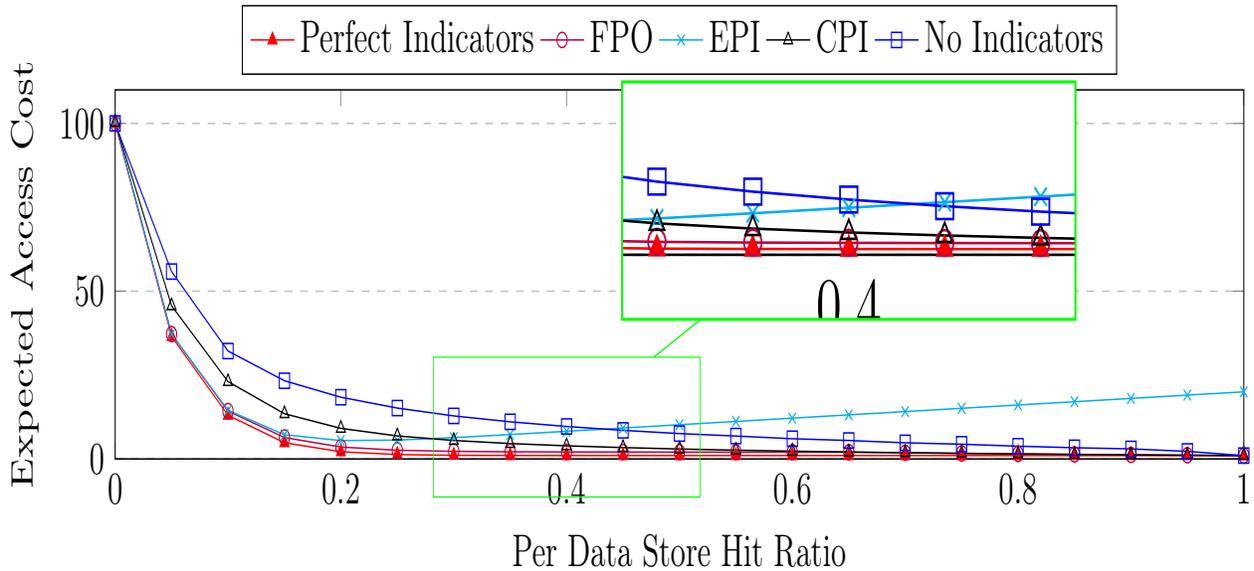

Accessing a data store incurs a certain cost in terms of latency, bandwidth, and energy~\cite{Joint_opt}. Hence, smart
utilization of data stores may reduce the operational costs of such systems and improve their users'
experience. Naturally, knowing which item is stored in each data store at any given moment is a key enabler
for efficient utilization, but maintaining such knowledge may not be feasible. Instead, it is more practical
to occasionally exchange space efficient \emph{indicators} for the content of the data
stores~\cite{SummaryCache}. Bloom filters~\cite{Bloom} are a common implementation for such indicators, but
many other space-efficient approximate membership representations can also be
used~\cite{Survey04,Survey12,TinySet,TinyTable,BloomParadox,CDN_OceanStore,AccessEfficientBF}.

The shortcoming of relying on such indicators is that they may exhibit {\em false positives}, meaning that they may indicate that a given item is held by a certain data store while it is actually not there.
Indeed, the work of~\cite{BloomParadox} formally showed that naively relying on indicators for accessing even a single data store may do more harm than good.
In this work, we are interested in the general case of accessing multiple data stores. The difference is that we require an access strategy that selects a~{\em subset} of the data stores to access per request.
Existing strategies for this problem include:
\begin{inparaenum}[(i)]
	\item the {\em Cheapest Positive Indication ($\cpi$)}~\cite{Digest,Survey12}
	strategy that accesses the cheapest data store

	with a positive indication for the requested item, and
	\item the {\em Every Positive Indication ($\epi$)}~\cite{SummaryCache} strategy that accesses every data store with a positive indication.
\end{inparaenum}
The access is considered successful if the item is stored in one of the accessed data stores, and incurs no further cost. Otherwise, we pay a {\em miss penalty} for retrieving the requested item, e.g., due to the need to fetch it from an external remote site.

In the example of Figure~\ref{fig:bf_toy_example}, \cpi{} accesses only data store 1, which is the cheapest with a positive indication (captured by $I(x) = \mathrm{Yes}$), and incurs a cost of 1 for this.
However, since $x$ is not in data store 1 (captured by $C(x)=\mathrm{No}$), this indication is a false positive, and an additional miss penalty of 100 is incurred for the request, for a total cost of 101 imposed on CPI.
Alternatively, the \epi\ policy accesses every data store with a positive indication (data stores 1, 2, and~3).
This implies an access cost of $1+2+5=8$.
In this case, no additional miss penalty is incurred, since item $x$ is indeed available in one of the accessed data stores, e.g., in data store 2.
One can also consider an ideal strategy equipped with a {\em perfect indicator} with no false positives. Such an ideal strategy would require a cost of merely 2 incurred for accessing data store 2 alone.

Figure~\ref{fig:example} provides a numerical example motivating this work (see~Section~\ref{sec:homo} for the exact settings).
The figure illustrates the expected access cost for varying strategies with a false positive ratio of $\fpr=0.02$. The strategies are compared to the performance of two baseline scenarios.
The No Indicators (blue) line illustrates the best that can be obtained without indicators (which can be viewed as using indicators with $\fpr =1$, or equivalently, using indicators that always return 'Yes').
In contrast, the Perfect Indicators (red) line corresponds to having no false positives $(\fpr =0)$ in any of the indicators.

The area between the plots describing the performance of the two baseline scenarios (blue and red) exhibits the potential gains of employing indicator based access policies.
Specifically, we observe that \epi\ is near optimal when the per data store hit ratio is low but becomes highly inefficient when it is high.
In fact, even the No Indicators approach outperforms \epi\ once the hit ratio is above a certain threshold (in our plot, this occurs at a hit-ratio of around 0.45).
In contrast, \cpi\ is near optimal when the hit ratio is very high but performs poorly when it is low.
Between these two extremes, there is a gap where both strategies are inefficient, as highlighted in the magnified area of Figure~\ref{fig:example}.
Our proposed strategies, described in Sections~\ref{sec:homo}-\ref{sec:pgm}, aim at providing near-optimal performance, independent of the actual hit ratio.
In particular, the performance of our false-positive-aware optimal policy, \fpo, depicted by the pink line, comes extremely close to the Perfect Indicators (red) line despite relying on indicators whose $\fpr=0.02$.

\subsection{Related work}\label{sec:related_work}
\paragraph*{Approximate Set Membership}
Approximate set membership is about encoding a set of items, such as the content of a data store, in a space efficient manner.
Intuitively, an accurate representation requires storing all identifiers which may be prohibitively expensive.
Alternatively, space can be conserved by allowing a small number of false positives.
Bloom filters~\cite{Bloom} offer space-efficient encoding but do not support the removal of items. Other works~\cite{TinySet,TinyTable,SummaryCache,AccessEfficientBF,Survey18} improve on them in various aspects, such as support for removals~\cite{TinyTable,dleftCBF,RankedIndexBF}, a more efficient access pattern~\cite{TinySet,AccessEfficientBF}, and lower transmission overheads~\cite{CompressedBF}.

\paragraph*{Applicability Examples}
Bloom filter variants are extensively used in multiple domains~\cite{Survey04, Survey12}. Most notable is their use in front of a cache or a slow memory hierarchy. Such usage leverages that Bloom filters do not exhibit false negatives. Thus, there is no need to access the data store on a negative indication.

The work of~\cite{SummaryCache} suggests an architecture for distributed caching on wide area networks. In this solution, caches share an approximation of their content.
Clients use this information to only contact the caches with positive indications (\epi).
A similar architecture is also considered in~\cite{Survey12, Digest}.
There, clients access the cheapest cache with a positive indication (\cpi).
However, the impact of the access strategy and its optimization in the face of false positive replies is overlooked in previous works.

\paragraph*{Access Strategies and Replica Selection}
The work of~\cite{Joint_opt} studied access strategies to datastores in a commercial content delivery network. 
Access strategies to datastores have been extensively studied also in the context of data grid systems. 
In such systems, the problem of selecting which datastore to access is commonly referred to as the \emph{replica selection} problem. 
A comprehensive survey of replica selection algorithms can be found in~\cite{Replica_survey_14}. 
However, all these works do not use indicators, but instead assume the existence of an exact and always-fresh list of locations of every stored datum. Maintaining such a repository incurs high overhead in terms of bandwidth consumption and synchronization mechanisms.

The work of~\cite{BloomParadox} considers the special case of a single data store, equipped with a Standard
Bloom Filter~\cite{Bloom} or a Counting Bloom Filter~\cite{CBF}. They identify cases where following a
positive indication may increase the overall cost. Thus, they suggest that in those cases the data store
should be ignored, regardless of its indicator value. We, on the other hand, address the more general
problem, which involves any number of data stores, equipped with any kind of indicators.

\subsection{System Model and Preliminaries}\label{sec:model}

\LTcapwidth=0.95\textwidth
\begin{longtable}{| p{.15\textwidth} | p{.68\textwidth} | p{.06\textwidth} |}%{|l|l|l|}
\caption[List of symbols used in Chapter~\ref{sec:BF}]{List of Symbols. The top part corresponds to our system model (Section~\ref{sec:model}, the middle part corresponds to the \pot\ and \umb\ algorithms (Sections~\ref{sec:pot}-\ref{sec:knapsack_algorithm}), and the bottom part corresponds to the \pgmalg\ algorithm (Section~\ref{sec:pgm}).}
\label{tbl:bf:notations} \\

\hline \multicolumn{1}{| p{.15\textwidth} |}{\textbf{Symbol}} & \multicolumn{1}{| p{.6\textwidth} |}{\textbf{Meaning}} & \multicolumn{1}{| p{.1\textwidth} |}{\textbf{Section}} \\ \hline 
\endfirsthead

\multicolumn{3}{c}%
{{\bfseries \tablename\ \thetable{} -- continued from previous page}} \\
\hline \multicolumn{1}{|c|}{\textbf{Symbol}} & \multicolumn{1}{c|}{\textbf{Meaning}} & \multicolumn{1}{c|}{\textbf{Section}} \\ \hline 
\endhead

\hline \multicolumn{3}{|r|}{{Continued on next page}} \\ \hline
\endfoot

\hline \hline
\endlastfoot
        \hline
    	\hline
		\ \ $\SetSs$ & \textcolor{red}{Set of all} data stores & \tabularnewline
		\hline
		\ \ $\NumSs$ &Number of data stores, $\NumSs = |\SetSs|$& \tabularnewline
		\hline
		\ \ $\SetSp$ &Data stores with positive indications for requested datum $x$& \tabularnewline
		\hline
		\ \ $\NumSp$ & Number of positive indications for requested datum $x$ ($|\SetSp|$)\normalsize
		& \tabularnewline
		\hline
		\ \ $\mS_j$ & The set of data items in data store $j$\normalsize
		& \tabularnewline
		\hline
		\rule{0pt}{3ex}    
		\ $\Phitj$ & Hit ratio of data store $j$& \tabularnewline
		\hline
		\ $\ind_j(x)$ & Indication of data store $j$ for datum $x$ & \tabularnewline
		\hline
		\ $q_j$ & Probability of positive indication by $\ind_j$: $\Pr (\ind_j(x) = 1)$ & \tabularnewline
		\hline
		\rule{0pt}{3ex}    
		$\fprj$ & False positive ratio for $\ind_j$: $\fprj = \Pr(\ind_j (x) = 1 | x \notin \mS_j)$ & \tabularnewline
		\hline
		\ \ $\mr_j$ & Misindication ratio for a data store $j$ & \tabularnewline
		\hline
		\ \ $\mr_{\mL}$ & Misindication ratio for a set of data stores $\mL$ & \tabularnewline
		\hline
		\ \ $\mc_i$ & Access cost for data store $i$ & \tabularnewline
		\hline
		\ \ $\mc_{\mL}$ & Total access cost (sum of costs of all data stores in set $\mL$) & \tabularnewline
		\hline
		\ \ $\phi$  & Cost function: $\phi(D) = \sum_{i \in D} \mc_i + \mbeta \prod_{i \in D} \mr_i$ & \tabularnewline
		\hline
		\ \ $\mbeta$  & Miss penalty & \tabularnewline
		\hline
		\hline
        \rule{0pt}{3ex}
		$H_k$  & Access cost for the $k$ highest data stores in $N_x$  &\ref{sec:pot},\ref{sec:knapsack_algorithm} \tabularnewline
		\hline
        \rule{0pt}{3ex}
		$L_k$  & Access cost for the $k$ lowest data stores in $N_x$  &\ref{sec:pot},\ref{sec:knapsack_algorithm}\tabularnewline
		\hline
		\rule{0pt}{3ex}
		$P(k^*)$ & Potential function: $P(k^*) = L_{k^*} + \mbeta \prod\nolimits_{j = 1}^{k^*} \mr_{\ell_j}$&\ref{sec:pot},\ref{sec:knapsack_algorithm}\tabularnewline
		\hline
		\rule{0pt}{3ex}
		$M$  & $M=\min\set{\sum_{j \in \SetSp} c_j, \beta}$   &\ref{sec:pot},\ref{sec:knapsack_algorithm}\tabularnewline
		\hline
    	\rule{0pt}{3ex}
    	$\Oset$ & Set of datastores used by \opt &\ref{sec:pot},\ref{sec:knapsack_algorithm},\ref{sec:pgm}\tabularnewline
		\hline
        \hline
        \ \ $\logbeta$ & $\logbeta = \log \mbeta$ &\ref{sec:pgm}\tabularnewline
        \hline
		\rule{0pt}{3ex}    
        $\piset{\mrgline}{j}$ & Partition of $\SetSp$ in level $\mrgline$ &\ref{sec:pgm}\tabularnewline
		\hline
		\rule{0pt}{3ex}    
		$\OOset{\mrgline}{j}$ & Subset of datastores which \opt\ selects out of $\piset{\mrgline}{j}$
		&\ref{sec:pgm}\tabularnewline
		\hline
		\rule{0pt}{3ex}    
        $\mvec{\mrgline}{j}$ & Candidate sub-solutions which \pgmalg\ considers out of $\piset{\mrgline}{\ell}$&\ref{sec:pgm}\tabularnewline
		\hline
\end{longtable}

This section formally defines our system model and notations.
For ease of reference, our notation is summarized in Table~\ref{tbl:bf:notations}.
We consider a set $\SetSs$ of $\NumSs$ {\em data stores}, containing possibly overlapping subsets of
items. We denote by $\mS_j$ the set of items stored at data store $j$.
Given a sequence of requests for items
$\sigma$ (with possible repetitions), the {\em hit ratio} of a data store $j$ is the fraction of requests in $\sigma$ that were available in data store $j$ (when requested).
Our work assumes that past hit ratio is a good indication for the near
future~\cite{EinzigerEFM18,EinzigerFM17}. We denote by $\Phitj$ the hit ratio of data store $j$, i.e., the
probability that the next accessed item $x$ is stored in $\mS_j$.

Each data store $j$ maintains an \emph{indicator} $\ind_j$, which approximates $\mS_j$; given an item $x$, $\ind_j(x) = 1$ indicates that $x$ is likely to be in $\mS_j$  while $\ind_j(x) = 0$ indicates that it is surely not in $\mS_j$.
These are referred to as a {\em positive indication} and a {\em negative indication}, respectively.
Our model assumes indicators that may exhibit only one-sided errors, i.e., they never err when providing a negative indication\footnote{This means having no false negatives, i.e., $\Pr(\ind_j (x) = 0 | x \in \mS_j) = 0$.}.
In practice, most implementations satisfy this assumption~\cite{TinySet,TinyTable,SummaryCache,AccessEfficientBF}.
The {\em false positive ratio} of $\ind_j$ is defined by $\fprj = \Pr(\ind_j (x) = 1 | x \notin \mS_j)$.

Given an item $x$ within sequence $\sigma$, a query for $x$ triggers a {\em data access} which consists of
selecting a subset of the data stores $\mL$ and accessing this subset in parallel. The data access is considered
successful, or a {\em hit}, if the item $x$ is found in at least one of the data stores being accessed and is
considered unsuccessful, or a {\em miss}, otherwise. Since by our assumption all indicators might have
a one-sided error, we focus our attention only on subsets of data stores which all provide a positive
indication. Given such a subset of the data stores $D$ all providing a positive indication, we denote by
$\mr_{\mL}$ the misindication ratio of $\mL$, i.e., the probability that an item is not available in any of
the data stores in $\mL$, in spite of their positive indications. Note, that if $\mL = \emptyset$, then
$\mr_{\mL} = 1$. We make no assumptions on the sharing policy among the data stores. Yet, in the
analysis sections we assume that the misindication ratios are mutually independent, that is, $\mr_{\mL} =
\prod_{j \in \mL} \mr_j$. Under this assumption our analysis provides a baseline for understanding the
performance of such systems.

Each data store has some predefined {\em access cost}, $\mc_j$, which is incurred whenever data store $j$ is
being accessed. These access costs induce the overall cost for accessing a set $\mL$ of data stores, defined
by $\mc_{\mL} = \sum_{j \in \mL} \mc_j$. We assume without loss of generality that $\min_j c_j = 1$. In case
the data access results in a miss, it incurs a {\em miss penalty} of $\mbeta$, for some $\mbeta \geq 1$. For
a subset of data stores $\mL$, which all provide a positive indication, we define its (expected) {\em miss
cost} by $\beta \cdot \mr_{\mL}$.

For any query item $x$, let $\SetSp \subseteq \SetSs$ denote the subset of data stores with a positive indication, i.e., $\SetSp = \set{j \in \SetSs | I_j(x)=1}$, and denote the size of this set by $\NumSp = |\SetSp|$.
The expected {\em cost} of accessing any $\mL \subseteq \SetSp$ is defined to be the sum of its access cost and its expected miss cost, i.e.,
\begin{equation}
\phi(D)=\mc_{\mL} + \beta \cdot \mr_{\mL}.
\end{equation}
%In case the
When misindication ratios are mutually independent we have
\begin{equation}
\label{eq:min-independent}
\phi(D)=\mc_{\mL} + \beta \cdot \mr_{\mL} = \sum\nolimits_{j \in \mL} \mc_j + \beta \prod\nolimits_{j \in \mL} \mr_j.
\end{equation}

The {\em Data Store Selection (\cacheprob) problem} is to find a subset of data stores $\mL \subseteq \SetSp$ that minimizes the expected cost $\phi(\mL)$. 

We denote by $\Pone_j$ the probability that indicator $j$ positively replies to a query for an item $x$.
This happens when either $x \in S_j$; or $x \notin S_j$, and a false positive occurs. Therefore,
\begin{equation}\label{Eq:\ind_eq_1}
\Pone_j = \Pr(\ind_j(x) = 1) = \Phitj + (1 - \Phitj) \fprj .
\end{equation}
Using Bayes' theorem and Eq.~\ref{Eq:\ind_eq_1}, the misindication ratio $\rho_j$ is%therefore
\begin{align}\label{Eq:mrj1}
\mr_j
&\equiv \Pr (x \notin S_j | \ind_j (x)= 1) \notag \\
&= \fprj (1 - \Phitj) / [\Phitj + (1 - \Phitj) \fprj].
\end{align}

To simplify expressions throughout this chapter, we omit the base of the logarithms; we always use logarithms of base 2.

\subsection{Our Contribution}
As mentioned, despite the popularity of indicators, the problem of efficiently working with indicators and of forming a successful access strategy has remained unexplored.
In Section ~\ref{sec:homo} we analyse the case of a fully homogeneous settings. Our analysis shows that even in this highly-simplified settings previously suggested strategies are too simplistic and implicitly rely on specific assumptions about the workload, or the underlying system. Thus, in general, an access strategy that works well in one scenario may be inefficient for another.

In Sections~\ref{sec:pot}-\ref{sec:pgm} we propose and \textcolor{red}{analyze} several polynomial-time approximation algorithms for the fully heterogeneous case.\footnote{Please recall that we detail about approximation algorithms in Section~\ref{sec:approx_algs}.}
We further validate and evaluate our proposed algorithms via an extensive evaluation in Section~\ref{sec:sim}. Our evaluation is based on real data 
with varying system parameters. Our results show that our algorithms are more stable than existing approaches.
That is, they outperform or achieve very similar access costs to the best competitor for any tested system configuration.
We conclude in Section~\ref{sec:discussion} with a discussion of our results.

\section{The Fully Homogeneous Case}\label{sec:homo}

To gain some insight about the challenges in developing an access strategy, we start with a simplified
fully-homogeneous case. In this setting, the cost of accessing each data store is the same ($c=1$). The per
data store hit ratios and false positive ratios are uniform, i.e., for each $j$, $\Phitj = \Phit$ and $\fprj
= \fpr$, for some constants $\Phit, \fpr \in [0,1]$. Consequently, the per data store misindication ratios,
captured by Eq.~\ref{Eq:mrj1}, are also uniform, i.e., for each $j$, $\mr_j = \mr$ for some constant $\mr \in
[0,1]$. Recall that our objective is to pick a subset of data stores with positive indications, $D \subseteq
\SetSp$, so as to minimize the overall expected cost of a query, $\phi(D)=\sum_{j \in \mL} c_j + \beta
\prod_{j \in \mL} \mr_j$. In the fully-homogeneous case considered here, the expected cost reduces to
$\phi(D)=|D| + \beta \mr^{|D|}$, which merely depends on the {\em size} of the chosen set $D$ of data stores
to be accessed. The task of choosing which subset of data stores to access is reduced to deciding on the
number $0 \leq k \leq \NumSp$ of data stores one should access. For any such potential number $k$, we denote
the expected cost of accessing $k$ data stores by
\begin{equation}\label{eq:total_cost_homo}
\costhomo (k) = k + \beta \mr^k,
\end{equation}
and focus our attention on studying the cost $\costhomo(\cdot)$ incurred by different data store selection schemes.

The size of the selected subset is clearly upper-bounded by the number of positive indications, $\NumSp$. So
we start by calculating the distribution of $\NumSp$. Ideally, one can interpret each positive indication as
a result of an independent Bernoulli trial with success probability $\Pone$. By Eq.~\ref{Eq:\ind_eq_1},
$\Pone=\Phit + (1-\Phit)\fpr$. Hence, $\NumSp$ is binomially distributed such that
\begin{equation}\label{eq:dist_of_n_x}
\Pr (\NumSp =k) = \binom{\NumSs}{k} \Pone^k (1- \Pone)^{\NumSs - k}.
\end{equation}

Using equations~\ref{eq:total_cost_homo} and~\ref{eq:dist_of_n_x} we now derive the expected costs of several selection schemes, where we let $D_X$ denote the set of data stores selected by selection scheme $X$.  

The \epi\ policy accesses all the data stores with positive indications, and therefore its expected overall cost is
\begin{align}\label{eq:homo_cost_of_epi}
\phi(D_{\epi})
& = \sum_{k = 0}^{\NumSs} \left[ \Pr (\NumSp =k) \cdot \costhomo(k) \right] \notag \\
%& = \sum_{k = 0}^{\NumSs} \left[ \Pr (\NumSp =k) \cdot \left( k + \beta \mr^k \right) \right] \\
& = \sum_{k = 0}^{\NumSs} \left[ \Pr (\NumSp =k) \cdot k \right] + \beta \cdot \sum_{k = 0}^{\NumSs} \left[ \Pr (\NumSp =k) \cdot \mr^k \right] \notag \\
& = E[\NumSp] + \mbeta \cdot \PGF_{\NumSp}(\mr)  \\
& = n \cdot \Pone + \mbeta \left( 1 - \Pone + \Pone \cdot \mr \right)^{\NumSs}, \notag
\end{align}
where $\PGF_X(t)$ denotes the probability generating function for random variable $X$ at point $t$.

\cpi\ accesses either a single data store with a positive indication, if one exists, or no data store if there are no positive indications.
The expected overall cost of \cpi\ is therefore
\begin{align}\label{eq:homo_cost_of_cpi}
\begin{split}
\phi(D_{\cpi}) & =  \Pr (\NumSp = 0) \cdot \costhomo (0) +
\Pr (\NumSp > 0) \cdot \costhomo (1)  \\
& = (1 - \Pone)^{\NumSs} \mbeta +
\left[1 - (1 - \Pone)^{\NumSs}\right] (1 + \mbeta \mr).
\end{split}
\end{align}

We now turn to analyze the false-positive-aware optimal policy, \fpo, which minimizes the expected overall
cost, given the false positive ratio, $\fpr$. In the fully homogeneous case, this translates to finding
$\argmin_k \costhomo (k)$. Consider $\costhomo (y)$ defined in Eq.~\ref{eq:total_cost_homo} as a function
defined over the reals. This function is convex since its second derivatives is non-negative, and it obtains
its minimum at $y^* = -\ln (-\mbeta \ln (\mr)) / \ln (\mr)$ for $0 < \mr < 1$. In practice, the number of
data stores accessed must be an integer between $0$ and $\NumSp$. The optimal number $m^*(k)$ of data stores
to access given that there are $k$ positive indications satisfies $m^*(k) \in \set{0, k, \lfloor y^* \rfloor,
\lceil y^* \rceil}$, where $\lfloor y^* \rfloor$ and $\lceil y^* \rceil$ should be considered only if $y^*\in
[0, k]$. Hence, The expected overall cost of \fpo\ is
\begin{equation}\label{eq:homo_cost_of_opt} \phi(D_{\fpo}) = \sum_{k = 0}^{\NumSs} \left[ \binom{\NumSs}{k}
\Pone^k (1- \Pone)^{\NumSs - k} m^*(k) \right].
\end{equation}

Having studied the overall cost of the above policies, we may revisit Figure~\ref{fig:example}. The expected
costs of each of the policies are presented as a function of $\Phit$, using
Equations~\ref{eq:homo_cost_of_epi}-\ref{eq:homo_cost_of_opt}. In particular, in the special case where $\fpr
= 0$, the expected overall costs of \cpi, \fpo\ and the perfect indicators benchmark are identical. This fits
our intuition that when there are no false indications, the optimal policy is to access a single data store
among those with positive indications if such a data store exists. At the other extreme, we have the case
where $\fpr = 1$, in which we always have $\NumSp = \NumSs$, i.e., all the indicators are positive. This
extreme case renders the indicators useless and is thus equivalent to not having indicators at all. In
particular, note that depending on the values of $\NumSs$ and $\mbeta$, \epi\ might end up being worse than
not having any indicators at all.

In this section we addressed the fully homogeneous case, in which minimizing our objective function
$\phi(D)$ was made tractable due to the uniformity of the settings. However, many systems are
heterogeneous, making the minimization of $\phi(D)$ a much more challenging task.
In the following sections we describe several approximation algorithm for solving the \cacheprob\ problem in fully heterogeneous
settings and provide a rigorous analysis of their performance. In particular, we also study trade-offs
between the time complexity and the performance guarantees of our proposed solutions.

\section{A Potential-based Algorithm}\label{sec:pot}
In this section we describe our first approximation algorithm for solving the \cacheprob\ problem in fully heterogeneous
settings and provide a rigorous analysis of their performance. 

Recall that our goal is to select a subset $D \subseteq N_x$ of data stores with positive indications
minimizing the expected cost
$$\phi(D) = c_D + \beta \mr_D = \sum\nolimits_{i \in \mL} \mc_j + \beta \prod\nolimits_{j \in D} \mr_j,$$
as defined in Eq.~\ref{eq:min-independent}. This can be viewed as a combined bi-criteria optimization
problem, of minimizing two objectives simultaneously:
\begin{inparaenum}[(i)]
	\item $c_D$, which is monotone {\em non-decreasing} as we pick more data stores to include in $D$, and
	\item $\mr_D$, which is monotone {\em non-increasing} as we pick more data stores to include in $D$,
\end{inparaenum}
where the latter objective is ``regularized'' by $\beta$.

In the special case where the non-decreasing orderings of data stores by access costs and by misindication ratios are the same, a simple substitution argument shows that a greedy approach will yield an optimal solution $D$ which consists of a prefix of this ordering.

In what follows we generalize the above observation and suggest an algorithm for the general case based on
the special case described above. We denote by $L_k$ and $H_k$ the sum of the $k$ smallest access costs of
data stores in $\SetSp$ and the $k$ largest access costs of data stores in $\SetSp$, respectively. Our
algorithm, \pot, described in Algorithm~\ref{alg:potential_algorithm}, considers the data stores ordered in
non-decreasing order of miss-ratio, $\ell_1, \ldots, \ell_{\NumSp}$, such that $\mr_{\ell_j} \leq
\mr_{\ell_{j+1}}$ for all $j=1,\ldots,\NumSp-1$. The algorithm iterates over all prefixes of indices in this
order, and picks a subset of data stores corresponding to a prefix which minimizes the {\em potential
function} $P(k) = L_k + \beta \prod\nolimits_{j=1}^k \mr_{\ell_j}$.

\begin{algorithm}[t!]
	\caption{\pot($\SetSp$,$c$,$\mr$,$\beta$)}
	\label{alg:potential_algorithm}
	\begin{algorithmic}[1]
		\State $\ell_1, \ldots, \ell_{\NumSp} \gets \SetSp$ in non-decreasing order of $\mr_j$
		\For{$k=1,\ldots,\NumSp$} \label{alg:pot:for_start}
		\State $D_k \gets \set{\ell_1,\ldots,\ell_k}$
		\EndFor \label{alg:potential_alg:for_end}
		\State \Return $D = \arg\min_k \set{ P(k) = L_k + \beta \prod\nolimits_{j=1}^k \mr_{\ell_j}}$
	\end{algorithmic}
\end{algorithm}

We now turn to analyze the performance of our proposed algorithm \pot. In particular, we show the following
theorem:
\begin{theorem}
	\label{thm:papp}
Let $\Dset$ be an optimal set of data stores for the \cacheprob\ problem, and let $D$ be
the solution found by \pot. Then $\phi(D) \leq \frac{H_{\abs{D}}}{L_{\abs{D}}} \phi(\Dset)$.
\end{theorem}
\begin{proof}
	Let $k=\abs{D}$. We therefore have
	\begin{align}
	\label{eq:potential1}
	\phi(D)
	&= \sum\nolimits_{j = 1}^k c_{\ell_j} + \mbeta \prod\nolimits_{j = 1}^k \mr_{\ell_j}
	\leq H_k + \mbeta \prod\nolimits_{j = 1}^k \mr_{\ell_j} \notag \\
	%    &= \frac{H_k}{L_k} L_k + \mbeta \prod\nolimits_{j = 1}^k \mr_{\ell_j}
	&\leq \frac{H_k}{L_k} \left( L_k + \mbeta \prod\nolimits_{j = 1}^k \mr_{\ell_j} \right)
	= \frac{H_k}{L_k} P(k),
	\end{align}
	where the penultimate inequality follows from the definitions of $L_k$ and $H_k$, and the last equality follows from the definition of the potential function $P(k)$.
	Let $k^*=\abs{\Dset}$. Since data stores are ordered in non-decreasing order of misindication ratio, it follows that $\prod\nolimits_{j=1}^{k^*} \mr_{\ell_j} \leq \prod\nolimits_{j \in \Dset} \mr_j$, and by the definition of $L_{k^*}$ as the sum of the $k^*$ smallest access costs of data stores in $\SetSp$, it follows that
	\begin{align}
	\label{eq:potential2}
	P(k^*)
	&= L_{k^*} + \mbeta \prod\nolimits_{j = 1}^{k^*} \mr_{\ell_j} \\
	&\leq \sum\nolimits_{j \in \Dset} c_j + \mbeta \prod\nolimits_{j \in \Dset} \mr_j
	= \phi(\Dset). \notag
	\end{align}
	Since $D$ is chosen to be the set of data stores that minimizes $P(k)$, where $k$ is the length of the prefix $\SetSp$ considered in non-decreasing order of miss-ratio, we have $P(k) \leq P(k^*)$. Combining this with Eqs.~\ref{eq:potential1} and~\ref{eq:potential2}, the result follows.
\end{proof}

Since for every $k$ we have $\frac{H_k}{L_k} \leq \max_j\set{c_j}$ and the running time of \pot\ is
dominated by the time required to sort the data stores, we obtain the following corollary:
\begin{corollary}
	\label{cor:potential_c_max}
	\pot\ is a $(\max_j\set{c_j})$-approximation algorithm, running in time $O(\NumSp \log \NumSp)$.
\end{corollary}

In particular, Corollary~\ref{cor:potential_c_max} implies that for the case where all accesses costs are
equal, \pot\ yields an optimal solution to the \cacheprob\ problem.

\section{A Knapsack-based Algorithmic Framework}
\label{sec:knapsack_algorithm}

In this section we develop an alternative algorithm for the \cacheprob\ problem and provide guarantees on
its performance. We begin by recalling that the main difficulty in solving the \cacheprob\ problem stems from
the fact that our objective function is composed of a {\em linear} component (the access cost) and a {\em
multiplicative} component (the miss cost). The algorithmic framework we propose in the sequel is based on
carefully linearizing the multiplicative component, and defining a collection of {\em knapsack problems} for
which their solution space contains a good approximate solution to the \cacheprob\ problem.

We associate each data store $j$ with its {\em log-hit weight}, defined by $w_j=-\log(\mr_j)$. We therefore
have for every subset of data stores $D \subseteq N$, $ - \log (\mr_D) = \sum\nolimits_{j \in D} w_j$.
Therefore, any set of data stores has a minimal miss cost if and only if it has a maximal log-hit weight. In
what follows we define a collection of Knapsack problems, where the Knapsack problem is defined as follows:
Given a budget $B$, and collection of items $U$, such that each item $j \in U$ has some profit $\pi_j$ and
cost $\gamma_j$, the goal is to find a subset of items $S \subseteq U$ such that $\sum\nolimits_{j \in S}
\gamma_j \leq B$ and $\sum\nolimits_{j \in S} \pi_j$ is maximized. We refer to such an instance as the
$(B,U,\pi,\gamma)$-Knapsack problem, and denote by $A_{\knap}(B,U,\pi,\gamma)$ the set of items produced as
output by an algorithm $A_{\knap}$ for the Knapsack problem. The Knapsack problem is known to be NP-hard, but
it can be solved exactly by dynamic programming in pseudo-polynomial time, and can be approximated to within
a $(1+\epsilon)$ factor in polynomial time by an FPTAS~\cite{williamson11design}.

We now turn to define our collection of knapsack problems, to be used by our algorithm for solving the
\cacheprob\ problem. We recall that given a query $x$, $\SetSp \subseteq \SetSs$ denotes the subset of data
stores for which their indicator is positive.
%, i.e., $\SetSp = \set{j \in \SetSs | I_j(x)=1}$.
In the following we let $M=\min\set{\sum\nolimits_{j \in \SetSp} c_j, \beta}$. Clearly, $M$ is an upper bound
on the access cost of any optimal solution for the \cacheprob\ problem. For any $B \in \set{0,1,\ldots,M}$,
consider the $(B,\SetSp,w,c)$-Knapsack problem, i.e., the Knapsack problem with budget $B$ over a collection
of items $\SetSp$, such that each item $j \in \SetSp$ has profit $w_j$ (the log-hit weight of data store $j$)
and cost $c_j$ (the access cost of data store $j$).

Our algorithm named \ppapproxalg, formally defined in
Algorithm~\ref{alg:pseudo_poly_knapsack_based_algorithm}, makes use of a $(1+\epsilon)$-approximation
algorithm $A_{\knap}$ for the knapsack problem, for some $\epsilon\geq 0$. The complexity and performance
guarantee depends upon the value of $\epsilon$. \ppapproxalg\ essentially iterates over all possible values
for the access cost, and solves the associated Knapsack problem using the algorithm $A_{\knap}$ as a
subroutine for each such value. \ppapproxalg{} then selects the subset of data stores $\mL \subseteq \SetSp$
which minimizes $\phi(\mL)$ over all Knapsack solutions calculated by $A_{\knap}$ in all iterations.

\begin{algorithm}[t!]
	\caption{\ppapproxalg($\SetSp$,$c$,$\mr$,$\beta$, $(1+\epsilon)$-approximation algorithm $A_{\knap}$ for Knapsack)}
	\label{alg:pseudo_poly_knapsack_based_algorithm}
	\begin{algorithmic}[1]
		\State $w_j \gets -\log(\mr_j)$ for all $j \in \SetSp$
		\For{$B \in \set{0,1,\ldots,\min\set{\sum\nolimits_{j \in \SetSp} c_j, \beta }}$}
		 \label{alg:ppapprox:for_start}\Comment This algorithm assumes $c_j$ is an integer
		\State $D_B \gets A_{\knap}(B,\SetSp,w,c)$
		\EndFor \label{alg:ppapprox:for_end}
		\State \Return $D = \arg\min_B \set{\phi(D_B)}$
	\end{algorithmic}
\end{algorithm}

We first show that if $A_{\knap}$ finds an {\em optimal} solution to the Knapsack problem in each iteration,
then our algorithm finds an {\em optimal} solution to the \cacheprob\ problem. In terms of running time,
since the best exact algorithm for the Knapsack problem over $n$ items with budget $B$ runs in
pseudo-polynomial time of $O(nB)$~\cite{williamson11design}, our algorithm also runs in pseudo-polynomial
time. These properties are formalized in the following theorem:

\begin{theorem}
	\label{thm:pp_optimal_knapsack_alg_performance}
	When using the pseudo-polynomial algorithm $A_{\knap}$ which finds an optimal solution to the Knapsack problem over $n$ items with budget $B$ in time $O(nB)$, \ppapproxalg\ is a pseudo-polynomial algorithm that finds an optimal solution to the \cacheprob\ problem in time $O(\NumSp M^2)$.
\end{theorem}
\begin{proof}
We first show that \ppapproxalg, defined in Algorithm~\ref{alg:pseudo_poly_knapsack_based_algorithm}, finds
an optimal solution to the \cacheprob\ problem.
	Consider an optimal solution $\Dset \subseteq \SetSp$ for the \cacheprob\ problem, and let $B^*=c_{\Dset}$.
Since by optimality $B^* \leq M$, we are guaranteed that \ppapproxalg\ considers $B=B^*$ in one of the
iterations of the for-loop in lines~\ref{alg:ppapprox:for_start}-\ref{alg:ppapprox:for_end}. Let $D_B$ denote
the solution of the knapsack problem being solved in that iteration, where the knapsack budget is $B$.
	Since algorithm $A_{\knap}$ finds an optimal solution for the knapsack problem in this iteration
	$$
	D_B = \arg\max_{D \subseteq \SetSp | c_D \leq B} \set{\sum\nolimits_{j \in D} w_j}.
	$$
	By the definition of $w_j$ and the monotony of the $\log$ function, such a $D_B$ also satisfies
	\begin{equation}
	\label{eq:rho_minimizer_solution}
	D_B = \arg\min_{D \subseteq \SetSp | c_D \leq B} \set{\mr_D}.
	\end{equation}
Assume by contradiction that $D_B$ is not optimal for the \cacheprob\ problem, i.e., that
$\phi(D_B) = c_{D_B} + \beta \mr_{D_B} > c_{\Dset} + \beta \mr_{\Dset} = \phi(\Dset)$. Since $c_{\Dset} = B^* = B
\geq c_{D_B}$, it must follow that $\mr_{D_B} > \mr_{\Dset}$, for $c_{\Dset} \leq B$, which contradicts
Eq.~\ref{eq:rho_minimizer_solution}.
	
	\paragraph*{Running time} \ppapproxalg\ performs $M$ iterations, where in each iteration it solves a knapsack problem using an algorithm that runs in $O(\NumSp M)$ time.
	It follows that the running time of \ppapproxalg\, in this case, is $O(\NumSp M^2)$, as required.
\end{proof}

In many cases, the value of $M=\min\set{\sum\nolimits_{j \in \SetSp} c_j, \beta }$ is polynomially bounded by
$\NumSp$. The following is an immediate corollary of Theorem~\ref{thm:pp_optimal_knapsack_alg_performance} in
such cases:

\begin{corollary}
If $M=\min\set{\sum\nolimits_{j \in \SetSp} c_j, \beta }$ is polynomially bounded by $\NumSp$, then
\ppapproxalg\ solves the \cacheprob\ problem in polynomial time.
\end{corollary}

We now turn to study the tradeoff between the running time of \ppapproxalg\ and its performance guarantee,
when using a polynomial time approximation algorithm for Knapsack instead of the pseudo-polynomial time exact
algorithm. We first show in Theorem~\ref{thm:pp_approx_knapsack_alg_performance} how the approximation
guarantee of an algorithm for Knapsack translates to an approximation guarantee for the \cacheprob\ problem,
while still in pseudo-polynomial time.

\begin{theorem}
	\label{thm:pp_approx_knapsack_alg_performance} If there exists some constant $\delta$ such that $\mr_j
\leq \delta$ for all $j \in \SetSp$ and algorithm $A_{\knap}$ is a $(1+\epsilon)$-polynomial time
approximation algorithm for Knapsack running in time $O(f(\NumSp,\epsilon))$, then
\ppapproxalg\ is a
pseudo-polynomial algorithm that finds an $O(\mbeta^{\frac{\epsilon}{1+\epsilon}})$-approximate solution for
the \cacheprob\ problem in time $O(f(\NumSp,\epsilon)\cdot M)$.
\end{theorem}

\begin{proof}
First, note that by its definition, the running time of \ppapproxalg\ is as required since it makes $M$
iterations, and in every iteration solves an instance of Knapsack in time $O(f(\NumSp,\epsilon))$. It remains
to bound the approximation ratio of \ppapproxalg.

Consider an optimal solution $\Dset \subseteq \SetSp$ to the \cacheprob\ problem, and let $B^*=c_{\Dset}$ and
$\ell$ be an integer such that
	\begin{equation}\label{eq:ell_upper_lower_bound}
	2^{-(\ell+1)} \leq \mr_{\Dset} \leq 2^{-\ell}.
	\end{equation}
	By our assumption there exists some constant $\delta$ such that for all $j \in \SetSp$ we have $\mr_j
\leq \delta$. We are therefore guaranteed to have $\ell = O(\log \beta)$, since for $\ell > \log_{1/\delta}
\beta$ we have $\mr_{\Dset} \beta < 1$, in which case the optimal solution would not benefit from accessing
more data stores than it currently does.
	By the definition of the log-hit weight, we therefore have $\ell \leq \sum\nolimits_{j \in \Dset} w_j \leq
\ell + 1$.

Consider the iteration of \ppapproxalg\ where $B=B^*$, and let $D_B$ denote the solution obtained by
algorithm $A_{\knap}$ for solving the Knapsack problem in this iteration.
	Since $A_{\knap}$ is a $(1+\epsilon)$-approximation algorithm we are guaranteed to have
	$
	\sum\nolimits_{j \in D_B} w_j
	\geq \frac{1}{1+\epsilon} \sum\nolimits_{j \in \Dset} w_j
	$
	since $\Dset$ is an optimal solution with an access cost of $B^*$, and therefore maximizes the objective function in the Knapsack problem being solved in this iteration.
	It follows that
	\begin{align}
	\prod\nolimits_{j \in D_B} \mr_j
	&\leq \prod\nolimits_{j \in \Dset} \mr_j^{\frac{1}{1+\epsilon}}
	\leq 2^{\frac{-\ell}{1+\epsilon}} \\
	&= 2^{-\ell + \frac{\epsilon \ell}{1+\epsilon}}
	= 2^{-(\ell+1) + (1 + \frac{\epsilon \ell}{1+\epsilon})} \notag \\
	&\leq 2^{1 + \frac{\epsilon \ell}{1+\epsilon}} \prod\nolimits_{j \in \Dset} \mr_j
	\leq O \left( \beta^{\frac{\epsilon}{1+\epsilon}}\right ) \prod\nolimits_{j \in \Dset} \mr_j, \notag
	\end{align}
	where the first inequality follows from our Knapsack approximation guarantee, the following two inequalities follow from Eq.~\ref{eq:ell_upper_lower_bound}, and the last inequality follows from the fact that $\ell = O(\log \mbeta)$.
	For $B=B^*$ we are guaranteed to have $\sum\nolimits_{j \in D_B} c_j \leq B^*$.
	Hence,
	\begin{align}
	\begin{split}
	\phi(D_B)
	&= \sum\nolimits_{j \in D_B} c_j + \beta \prod\nolimits_{j \in D_B} \mr_j \\
	&\leq B^* + O \left( \beta^{\frac{\epsilon}{1+\epsilon}}\right ) \left( \beta \prod\nolimits_{j \in \Dset} \mr_j \right) \\
	&= \sum\nolimits_{j \in \Dset} c_j + O \left( \beta^{\frac{\epsilon}{1+\epsilon}}\right ) \left( \beta \prod\nolimits_{j \in \Dset} \mr_j \right) \\
	&\leq O \left( \beta^{\frac{\epsilon}{1+\epsilon}}\right ) \left( \sum\nolimits_{j \in \Dset} c_j + \beta \prod\nolimits_{j \in \Dset} \mr_j \right) \\
	&= O \left( \beta^{\frac{\epsilon}{1+\epsilon}}\right ) \phi(\Dset)
	\end{split}
	\end{align}
	which completes the proof.
\end{proof}

In what follows, we present a polynomial-time approximation algorithm, \fpapproxalg\, for the problem,
formally defined in Algorithm~\ref{alg:fully_poly_knapsack_based_algorithm}. The algorithm is based on
\ppapproxalg\ but avoids the need to iterate over all possible budgets. In particular, \fpapproxalg\ does not
make use of a general $(1+\epsilon)$-approximation algorithm for solving the Knapsack problem. Instead, \fpapproxalg\ incorporates within its design the specifics of a 2-approximation algorithm for the Knapsack problem, the details of which are presented and discussed in the proof of
Theorem~\ref{thm:fp_approx_knapsack_alg_performance}.

\begin{algorithm}[t!]
	\caption{\fpapproxalg($\SetSp$,$c$,$\mr$,$\beta$)}
	\label{alg:fully_poly_knapsack_based_algorithm}
	\begin{algorithmic}[1]
		\State $w_j \gets -\log(\mr_j)$ for all $j \in \SetSp$
		\For{$u \in \set{c_j | j \in \SetSp}$} \label{alg:prune_start}
		\State $\SetSp^u \gets \set{j \in \SetSp | c_j \leq u}$, let $\NumSp^u=\abs{\SetSp^u}$ \label{alg:prune_mid}
		\State $k_1,\ldots,k_{\NumSp^u} \gets \SetSp^u$ in non-increasing order of $w_j/c_j$
		\For{all $1 \leq t \leq \NumSp^u$} \label{alg:fpapprox:for_start}
		\State $D_t^u \gets \set{k_1,\ldots,k_t}$
		\State $\tilde{D}_t^u \gets \set{k_t}$
		\EndFor \label{alg:fpapprox:for_end}
		\EndFor \label{alg:prune_end}
		\State \Return $D = \arg\min_{D \in \set{D_t^u}\cup\set{\tilde{D}_t^u}\cup\set{\emptyset}} \set{\phi(D)}$
	\end{algorithmic}
\end{algorithm}

\begin{theorem}
	\label{thm:fp_approx_knapsack_alg_performance}
If there exists some constant $\delta$ such that $\mr_j
\leq \delta$ for all $j \in \SetSp$, then Algorithm \fpapproxalg\ is a polynomial
$O(\sqrt{\mbeta})$-approximation algorithm running in time $O(\NumSp^2 \log \NumSp)$.
\end{theorem}
\begin{proof}

The algorithm is based on the 2-approximation algorithm for Knapsack~\cite{williamson11design}, which works
as follows: given budget $B$, prune all elements with a cost greater than $B$. Order all elements in
non-increasing order of their profitability, captured by their profit-to-cost ratio. Greedily add elements to
the solution, starting from the most profitable one, as long as their overall cost does not exceed the given
budget. Once adding an element causes a violation of the budget constraint, pick the best out of
two candidate solutions: the set of elements accumulated which satisfy the budget constraint, and the first
element that caused the violation of the constraint.\footnote{Most common implementations consider the
element with maximum profit instead of the first element causing the violation of the budget constraint.
However, such an amended choice has no effect on the analysis of the algorithm's performance.}

The remainder of the proof draws its intuition from the proof of
Theorem~\ref{thm:pp_approx_knapsack_alg_performance}, combined with the properties of the 2-approximation
algorithm for Knapsack.
	
Given some budget constraint $B$ on the access cost of a solution, consider the 2-approximation algorithm
for knapsack when given $B$ as its budget constraint.

The algorithm first prunes all elements with cost greater than the budget. In particular, there exists some
element $j$ such that $c_j$ is the maximal cost of an element not violating the budget. \fpapproxalg\
simulates the same pruning by iterating over all potential values for this maximal cost, and maintaining
only the data stores with cost not exceeding this maximal cost
(lines~\ref{alg:prune_start}-\ref{alg:prune_mid}). It follows that there is a $u \in \set{c_j | j \in
\SetSp}$ for which
	\begin{equation}
	\label{eq:good_N_x_u_exists}
	\SetSp^u = \set{j \in \SetSp | c_j \leq B}.
	\end{equation}
	
Now that the knapsack approximation algorithm only considers items with cost not violating the budget $B$, it
orders the items in non-increasing order of $w_j/c_j$, and scans the items in this order, starting from the
most profitable, until reaching the first item in this order, $k_{t_B}$, such that $\sum\nolimits_{j=1}^{t_B}
c_{k_j} \leq B$, but $\sum\nolimits_{j=1}^{t_B + 1} c_{k_j} > B$. The algorithm then picks the best between
two possible candidate solutions: the set $\set{1,\ldots,k_{t_B}}$, and the set $\set{k_{t_B + 1}}$.

Our algorithm iterates over {\em all} potential candidates of this form, namely, all sets of data stores
$\set{1,\ldots,k_t}$, and all sets of data stores $\set{k_t}$. 	Consider an optimal solution $\Dset \subseteq
\SetSp$ to the \cacheprob\ problem, and denote by $B^*$ the access cost contributing to the overall cost of
$\Dset$. Consider the iteration of \fpapproxalg\ where $\SetSp^u = \set{j \in \SetSp | c_j \leq B^*}$ (as shown
in the argument leading to Eq.~\ref{eq:good_N_x_u_exists} such a cost $u$ necessarily exists). 	

Consider the items in $\SetSp^u$ ordered in non-increasing order of $w_j/c_j$, and let $t_{B^*}$ be the first
item in the order for which $\sum\nolimits_{j=1}^{t_{B^*}} c_{k_j} \leq B^*$, but
$\sum\nolimits_{j=1}^{t_{B^*} + 1} c_{k_j} > B^*$. The algorithm will choose either
$\set{1,\ldots,k_{t_{B^*}}}$, which is candidate $D_t^u$ in the iteration where $t=t_{B^*}$ of
lines~\ref{alg:fpapprox:for_start}-\ref{alg:ppapprox:for_end}; or it will choose $\set{t_{B^*} + 1}$, which
is candidate $\tilde{D}_t^u$ in the iteration where $t=t_{B^*}+1$ of
lines~\ref{alg:fpapprox:for_start}-\ref{alg:fpapprox:for_end}. By the proof of
Theorem~\ref{thm:pp_approx_knapsack_alg_performance}, the best of these two candidate solutions is an
$O(\mbeta^{\frac{\epsilon}{1+\epsilon}})=O(\sqrt{\beta})$ approximate solution for the \cacheprob\ problem,
since we are using a 2-approximation algorithm for knapsack, implying $\epsilon=1$. 	

Since \fpapproxalg\ picks the candidate solution with the minimal overall cost, the solution returned by the
algorithm is itself an $O(\sqrt{\mbeta})$-approximate solution for the \cacheprob\ problem. The running time
of the algorithm is dominated by the outer for-loop in lines~\ref{alg:prune_start}-\ref{alg:prune_end} which
has $\NumSp$ iterations, where in each iteration we order all elements in $\SetSp^u$, which takes $O(\NumSp
\log \NumSp)$ time.
	Hence, the overall running time of the algorithm is $O(\NumSp^2 \log \NumSp)$, which completes the proof.
\end{proof}

\section{A Partition-and-Merge Algorithmic Framework}\label{sec:pgm}
In this section we develop an alternative algorithm for the \cacheprob\ problem and provide guarantees on its performance. 
We first provide a high-level description of the algorithm, and then turn to a detailed description and analysis of its approximation ratio and run-time.
\subsection{High-level Description of the Algorithm}
As its name indicates, our algorithm, Partition, Generate and Merge (\pgmalg), is built upon three fundamental operations:
\begin{inparaenum}[(i)]
\item {\em Partition} the set of datastores with positive indications $\SetSp$ into disjoint fractions, based on a logarithmic scaling of the access costs.
\item {\em Generate} from each of the fractions candidate sets of datastores with minimal miss ratio.
\item {\em Merge} the candidate sets iteratively, until obtaining a full solution for the \cacheprob\ problem.
\end{inparaenum} 

Figure~\ref{fig:pgm_run_example:init} depicts the partition and generate stage of \pgmalg.
In the {\em partition} stage (Fig.~\ref{fig:pgm_run_example:init}), \pgmalg\ partitions 
$\SetSp$ into $\log \mbeta$ disjoint sets: $\piset{0}{0}, \piset{0}{1}, \dots, \piset{0}{\logbeta-1}$ (we denote $\logbeta = \log (\mbeta)$), where partition $j$ contains all the datastores with positive indications whose access costs fall in the range $[2^j, 2^{j+1})$. Each concrete subset of the set of datastores in some partition is a part of a full solution for the \cacheprob\ problem. In other words, one can compose a full solution for the \cacheprob\ problem by taking the union of $\log \mbeta$ sub-solutions, where each sub-solution is taken exclusively from one of the disjoint sets $\piset{0}{0}, \dots, \piset{0}{\logbeta-1}$.

In the \emph{generate} stage \pgmalg\ sorts the datastores within each partition in a non-decreasing order of the miss ratio, and considers all possible prefixes as its initial candidate sub-solutions. 
% Importantly, the a
We will later show that 
% Intuitively, 
this implementation of the partition and generate stages guarantees that for each partition, \pgmalg\ considers a sub-solution with at most twice the access cost, and at most the same miss ratio, of those obtained by the respective sub-solution of \opt.

Fig.~\ref{fig:pgm_run_example:mrg_no_zoom} depicts a high-level overview of the \emph{merge} stage. In this stage, \pgmalg\ iteratively merges candidate sub-solutions by means of a binary tree, whose leaves are the initial candidate sub-solutions produced by the partition and generate stages. In the final merge step the algorithm obtains a list of full solutions for the \cacheprob\ problem, from which it selects the one minimizing our objective function $\phi(*)$. 

In what follows we present some preliminaries, and then use them to describe \pgmalg\ in details.

\subsection{Preliminaries}\label{sec:pgm_preliminaries}
Denote $\logbeta = \log \mbeta$.
We iteratively partition $\SetSp$ to subsets as follows. All partitions are based on the access cost.
Initially, in level $\mrgline = 0$, we use a simple logarithmic-scaled partitioning, namely $\piset{0}{j} = \set{d \in \SetSp | 2^{j} \leq \mc_d < 2^{j+1}}$, where $j = 0, \dots, \logbeta-1$.
In levels $\mrgline = 1, \dots, \log \logbeta$, each higher-level subset is the union of two adjacent subsets in the lower level, namely $\piset{\mrgline}{j} = \piset{\mrgline-1}{2j} \cup \piset{\mrgline-1}{2j+1} $, where $j = 0, \dots, \frac{\logbeta}{2^{\mrgline}}-1$.

For each $\mrgline = 0, \dots, \log \logbeta$ and $j = 0, \dots, \frac{\logbeta}{2^{\mrgline}}-1$, let $\OOset{\mrgline}{j}$ denote the subset of data stores which an optimal solution $\Oset$ selects out of $\piset{\mrgline}{j}$. That is, $\OOset{\mrgline}{j} = \Oset \cap \piset{\mrgline}{j}$. By the definition of $\piset{\mrgline}{j}$, it follows that $\OOset{\mrgline}{j} = \OOset{\mrgline-1}{2j} \cup \OOset{\mrgline-1}{2j+1}$ for $\mrgline=1,\ldots,\log \logbeta$.

\pgmalg\ organizes candidate subsets of data stores as follows.
For each level $\mrgline = 0, \dots, \log \logbeta$ and $j = 0, \dots, \frac{\logbeta}{2^{\mrgline}}-1$,
$\mvec{\mrgline}{j}$ is a collection of candidate subsets of $\piset{\mrgline}{j}$. 

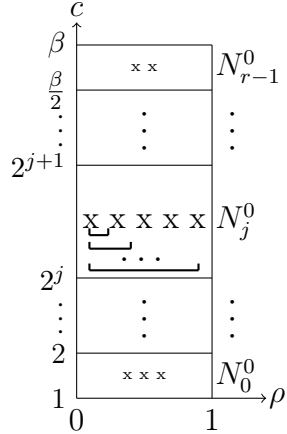
\begin{figure}[h!]
    \centering

    \newcommand {\yone}         {0.2}  % Initial y value
    \newcommand {\ysmall}       {0.6}  % vertical size of a "small" box
    \newcommand {\ylarge}       {1.5}  % vertical size of a "large" box
    \newcommand {\ythreedots}   {1.0} % vertical size of a "three dots" box
    
    \newcommand {\yonetwo}      {\yone +         0.5*\ysmall}
    \newcommand {\ytwo}         {\yone +             \ysmall} 
    \newcommand {\ytwoj}        {\ytwo +     0.5*\ythreedots}
    \newcommand {\yj}           {\ytwo +         \ythreedots}
    \newcommand {\yjjpp}        {\yj +           0.5*\ylarge}
    \newcommand {\yjpp}         {\yj +               \ylarge} 
    \newcommand {\yjpphalfbeta} {\yjpp +     0.5*\ythreedots} 
    \newcommand {\yhalfbeta}    {\yjpp +         \ythreedots}
    \newcommand {\yhalfbetabeta}{\yhalfbeta +    0.5*\ysmall}
    \newcommand {\ybeta}        {\yhalfbeta +        \ysmall}

    \newcommand {\xzero}        {0.2}                   % x value of left edge
    \newcommand {\xwidth}       {1.8}                   % Width of a box
    \newcommand {\xzeroone}     {\xzero + 0.5*\xwidth}  % Horizontal middle 
    \newcommand {\xone}         {\xzero +     \xwidth}  % x value of right edge
    
        \begin{tikzpicture}
    
        % Vertical lines 
        \draw (\xzero, \yone) -- (\xzero, \ybeta);
        \draw (\xone,  \yone) --  (\xone, \ybeta);
        % Vertical arrow with access cost
        \draw [->] (\xzero, \ybeta) -- (\xzero, \ybeta + 0.5*\ysmall); 
        \draw (\xzero, \ybeta + 0.8*\ysmall) node {$\mc$};
        
        % Horizontal lines 
        \draw (\xzero, \yone)       node[left] {\small{1}} node[below] {0}  -- (\xone, \yone);
        \draw (\xzero, \ytwo)       node[left] {\small{2}}                  -- (\xone, \ytwo);
        \draw (\xzero, \yj)         node[left] {\small{$2^j$}}              -- (\xone, \yj);
        \draw (\xzero, \yjpp)       node[left] {\small{$2^{j+1}$}}          -- (\xone, \yjpp);
        \draw (\xzero, \yhalfbeta)  node[left] {\small{$\frac{\mbeta}{2}$}} -- (\xone, \yhalfbeta);
        \draw (\xzero,  \ybeta)     node[left] {\small{$\mbeta$}}           -- (\xone, \ybeta);
         % Horizontal arrow with miss ratio
       \draw [->] (\xone, \yone) -- (\xone + 0.4*\xwidth, \yone);
        \draw (\xone + 0.48*\xwidth, \yone) node {$\mr$};
        
        % numbers around 
        \draw (\xone, \yone)              node[below] {1};
        
        % N_{}^{} on right
        \draw (\xone-0.1, \yonetwo)        node[right] {$\piset{0}{0}$};
        \draw (\xone-0.1, \yjjpp)          node[right] {$\piset{0}{j}$};
        \draw (\xone-0.1, \yhalfbetabeta)  node[right] {$\piset{0}{r-1}$};
        \draw (\xone+0.1, \ytwoj)          node[right] {\large{\rotatebox{90}{$\dots$}}};
        \draw (\xone+0.1, \yjpphalfbeta)   node[right] {\large{\rotatebox{90}{$\dots$}}};
        \draw (\xzero-0.05, \ytwoj)        node[left]  {\rotatebox{90}{$\dots$}};
        \draw (\xzero-0.05, \yjpphalfbeta) node[left]  {\rotatebox{90}{$\dots$}};

        % Data stores within the first level
        \draw (\xzeroone, \yonetwo)         node {\tiny{x x x}}; %{$\bigcdot \ \bigcdot$};
        \draw (\xzeroone, \ytwoj)           node {\large{\rotatebox{90}{$\dots$}}};
        \draw (\xzeroone, \yjjpp)           node {x x x x x};
        \draw (\xzeroone, \yjpphalfbeta)    node {\large{\rotatebox{90}{$\dots$}}};
        \draw (\xzeroone, \yhalfbetabeta)   node {\tiny{x x}};
        
        % Selecting prefixes
        \newcommand {\ygap}        {0.18};
        \newcommand {\yselect}     {0.09};
        \newcommand {\xprefixleft} {\xzero + 0.17}
        \newcommand {\linew}       {0.8}

        % Select prefix of 1
        \draw [line width=\linew] (\xprefixleft, \yjjpp -\ygap)        -- (\xprefixleft + 0.25, \yjjpp - \ygap);
        \draw [line width=\linew] (\xprefixleft, \yjjpp -\ygap)        -- (\xprefixleft,        \yjjpp -\ygap + \yselect);
        \draw [line width=\linew] (\xprefixleft + 0.25, \yjjpp -\ygap) -- (\xprefixleft + 0.25, \yjjpp -\ygap + \yselect);

        % Select prefix of 2
        \draw [line width=\linew] (\xprefixleft,        \yjjpp -2*\ygap) -- (\xprefixleft + 0.55, \yjjpp - 2*\ygap);
        \draw [line width=\linew] (\xprefixleft,        \yjjpp -2*\ygap) -- (\xprefixleft,        \yjjpp - 2*\ygap + \yselect);
        \draw [line width=\linew] (\xprefixleft + 0.55, \yjjpp -2*\ygap) -- (\xprefixleft + 0.55, \yjjpp - 2*\ygap + \yselect);

        % Select prefix of 5
        \draw [line width=\linew] (\xprefixleft,        \yjjpp -3.6*\ygap) -- (\xprefixleft + 1.45, \yjjpp - 3.6*\ygap);
        \draw [line width=\linew] (\xprefixleft,        \yjjpp -3.6*\ygap) -- (\xprefixleft,        \yjjpp - 3.6*\ygap + \yselect);
        \draw [line width=\linew] (\xprefixleft + 1.45, \yjjpp -3.6*\ygap) -- (\xprefixleft + 1.45, \yjjpp - 3.6*\ygap + \yselect);
        
        % 3 dots between the selection of 2 and 5
        \draw (\xzeroone, \yjjpp -2.8*\ygap) node {\large{$\dots$}};
        
        \end{tikzpicture}
	\caption[\pgmalg's partition and generate steps.]{\pgmalg's partition and generate steps. X represents a datastore.}\label{fig:pgm_run_example:init}
\end{figure}

%\itamar{The disclaimer below is used in the proof of Lemma~\ref{lemma:pgm}}
As a datastore with zero hit ratio is useless, \pgmalg\ assumes that for each $j=1, \dots \NumSp$, $\Phitj > 0$. Further, by Eq.~\ref{Eq:\ind_eq_1} $\mr_j < 1$.

\subsection{The \pgmalg\ Algorithm}\label{sec:pgm_details}

\begin{algorithm}[t!]
\caption{$\pmr \left(\SetSp,c,\mr,\beta\right)$}
\label{alg:pgm}
\begin{algorithmic}[1]
\Statex \AlgPhase{Partition and Generate sub-solutions}

\For {$j = 0, 1, \dots, \logbeta - 1$} \label{alg:pgm:partition_for_start}
    \Comment{$\logbeta = \log \mbeta$}
    \State $\piset{0}{j} \gets \set{j \in \SetSp | 2^j \leq \mc_j < 2^{j+1}}$\label{alg:pgm:partition}
    \State sort $\piset{0}{j}$ in non-decreasing order of miss ratio\label{alg:pgm:order} 
    \State $\mvec{0}{j} \gets \{D| D \textrm{\ is a prefix of \ } \piset{0}{j}, 0 \leq |D| \leq \NumSp \}$ \label{alg:pgm:genV0j}
\EndFor
\label{alg:pgm:partition_for_end}

\Statex \AlgPhase{Merge sub-solutions}

\For {$\mrgline  = 1, \dots, \log \logbeta$} \label{alg:pgm:for_line_begin}
    \Comment{level $\mrgline$}
    \For { $j = 0, \dots, \frac{\logbeta}{2^{\mrgline}}-1$ } \label{alg:pgm:for_line_for_j_begin}
        \Comment{node in level $\mrgline$}
        \State $\mvec{\mrgline}{j} \gets \{\emptyset\}$  \label{alg:pgm:init_to_emptyset_in_merge}
     	\For {$t = 1, 2, \dots, \logbeta$} \label{alg:pgm:for_t_merge_begin}
     	            \State $X_t \gets \arg\min \{\mr_D | D = A \cup B, \dots A \in \mvec{\mrgline-1}{2j}, B \in \mvec{\mrgline-1}{2j+1}, \mc_D \in [2^{t-1}, 2^t)\}$ 
     	            \label{alg:pgm:merge_find_argmin}
            \State $\mvec{\mrgline}{j} \gets \mvec{\mrgline}{j} \cup \set{X_t}$\label{alg:pgm:merge_insert_to_vec}
            \label{alg:pgm:merge}
        \EndFor \label{alg:pgm:for_t_merge_end}
    \EndFor \label{alg:pgm:for_line_for_j_end}
\EndFor \Statex \AlgPhase{Pick full solution}\label{alg:pgm:for_line_end}

\State return $\tilde{D}  = \arg\min_{X \in \mvec {\log \logbeta}{0}} \set{\phi(X)}$\label{alg:pgm:last_line}
\end{algorithmic}
\end{algorithm}

We now describe the details of \pgmalg, formally defined in Algorithm~\ref{alg:pgm}.
In lines~\ref{alg:pgm:partition_for_start}-\ref{alg:pgm:partition_for_end} \pgmalg\ partitions the data stores and generates candidate sub-solutions as follows. First, the algorithm partitions the set of data stores with positive indications $\SetSp$ into $\logbeta$ disjoint subsets, $\piset{0}{0}, \piset{0}{1}, \dots, \piset{0}{\logbeta-1}$ based on the access costs, using a logarithmic scale (line~\ref{alg:pgm:partition}). 
Then
the algorithm sorts the elements in each of the partitions by a non-decreasing order of miss ratios (line~\ref{alg:pgm:order}).
Next, the algorithm generates initial candidate sub-solutions by considering all possible prefixes of each partition (line~\ref{alg:pgm:genV0j}).
Recall that the partition and generate stages are illustrated in Fig.~\ref{fig:pgm_run_example:init}.

The second part of the algorithm (lines~\ref{alg:pgm:for_line_begin}-\ref{alg:pgm:for_line_end}) iteratively merges pairs of adjacent candidate solutions into a single candidate sub-solution. The merge process uses a binary tree where the leaves are the initial sets of candidate solutions $\mvec{0}{0}, \dots \mvec{0}{\logbeta-1}$, and the root is a set of full solutions for the \cacheprob\ problem $\mvec{\log \logbeta}{0}$, as depicted in Fig.~\ref{fig:pgm_run_example:mrg_no_zoom}. 

In particular, in each run of lines~\ref{alg:pgm:init_to_emptyset_in_merge}-\ref{alg:pgm:merge_insert_to_vec}
\pgmalg\ merges two sets of candidate sub-solutions, $\mvec{\mrgline-1}{2j}$ and $\mvec{\mrgline-1}{2j+1}$,
into a single set $\mvec{\mrgline}{j}$ 
as follows. First, \pgmalg\ initializes $\mvec{\mrgline}{j}$ to include only the empty set (line~\ref{alg:pgm:init_to_emptyset_in_merge}). 
Then (lines~\ref{alg:pgm:for_t_merge_begin}-\ref{alg:pgm:merge_insert_to_vec}), the algorithm considers all the edges in the full bipartite graph whose vertices in each side are the candidate sub-solutions in each of the two merged sets. Thus each edge in the bipartite graph represents a union of two sub-solutions. For each $t = 1, \dots, \logbeta$, \pgmalg\ selects the edge that minimizes the miss ratio among all the edges with access cost in the range $\left[2^{t-1}, 2^t\right)$. 

After merging all the sub-solutions into a list of candidate solutions, the algorithm finds and returns the best candidate full solution (line~\ref{alg:pgm:last_line}). 

\input{Fig_PGM_run_merge_Thesis.tex}

Fig.~\ref{fig:pgm_run_example:mrg} depicts the merge stages of \pgmalg. In particular, Fig.~\ref{fig:pgm_run_example:mrg_no_zoom} shows the binary merge tree. Observe that each node contains at most one sub-solution per each log-scale range of the access costs. Also note that each node contains the empty set, represented by an empty circle. For instance, consider the node $\mvec{1}{0}$, which is shown zoomed-in. 
The node $\mvec{1}{0}$ contains a single sub-solution in the access-costs' range $[0,1)$ -- namely, the empty set, whose access cost and miss ratio are 0 and 1; a single sub-solution in the range $[1,2)$; and a single sub-solution in the range $[2,4)$. The node $\mvec{1}{0}$  contains no candidate sub-solution in the range $[4, 8)$.

Fig.~\ref{fig:pgm_run_example:mrg_zoom} exemplifies a merge of two sets of candidate sub-solutions into a single set.
The leftmost part of the figure shows the two sets of sub-solutions which \pgmalg\ merges, $\mvec{\mrgline-1}{2j}$ and $\mvec{\mrgline-1}{2j+1}$. The algorithm considers all the edges in the bipartite graphs whose vertices are the candidate sub-solutions.

The middle part of the figure shows the selection of merged sub-solutions for $t=1,2,3$.
When $t=1$, \pgmalg\ considers all the unions of sub-solutions s.t. the access cost of the union is within the range $[1, 2)$. In our case, this translates to a single candidate union - the union of the empty set (represented by the empty circle) and the set of datastores with a total cost of 1. As this is the only candidate in this range, \pgmalg\ inserts it into the set of merged solutions, $\mvec{\mrgline}{j}$, in the rightmost part of the figure. Note that the access cost of the union is the sum of the access costs of the sub-solutions it unifies, while the miss ratio of the union is the multiplication of their miss ratios.

When $t=2$, \pgmalg\ considers all the unions of sub-solutions with access cost within the range $[2, 4)$. In our case, this translates to two candidate unions: one with miss ratio $1 \cdot 0.7 = 0.7$, and another 
with miss ratio $1 \cdot 0.6 = 0.6$. \pgmalg\ selects the latter union, which minimizes the miss ratio, and inserts it into the merged set of sub-solutions.

When $t=3$, \pgmalg\ considers all the unions of the sub-solutions with access cost within the range $[4, 7)$. This translates to two candidate unions: one with miss ratio $0.9 \cdot 0.7 = 0.63$, and another 
with miss ratio $0.6 \cdot 0.7 = 0.42$. \pgmalg\ selects the second option, which minimizes the miss ratio, and inserts it into the merged set of sub-solutions. Note that \pgmalg\ selects the second option so as to minimize the miss ratio, although its aggregate access cost is 5, which is higher than the aggregate access cost of the first option. This exemplifies how \pgmalg\ prioritizes the minimization of the miss ratio, which has a multiplicative impact on the objective function $\phi$, over the access cost, which has only an additive impact on $\phi$.

Below, we analyse the performance and run time of \pgmalg.

\subsection{Performance and Run Time Analysis}

Our performance analysis involves a careful comparison of the access cost, and miss ratio of candidate sub-solutions considered by \pgmalg\ out of 
every partition $\piset{\mrgline}{j}$ with the access cost, and miss ratio of a respective optimal sub-solution.

The following proposition shows that in each level $\ell$, $\set{\piset{\mrgline}{j}}_j$ is a partition of $\SetSp$ into disjoint sets, based on the access costs of the datastores.  

\begin{proposition}\label{prop:Nlj}
In each level $\mrgline$,
$\set{\piset{\mrgline}{j}}_j$
is a partition of $\SetSp$, satisfying
\begin{equation}\label{eq:Nlj}
\piset{\mrgline}{j} = \set{d \in \SetSp | 2^{j \cdot 2^\mrgline} \leq \mc_d < 2^{(j+1) 2^\mrgline}}
\end{equation}
\end{proposition}

\begin{proof}
We prove the claim by induction on $\mrgline$.
Assigning $\mrgline=0$ in Eq.~\ref{eq:Nlj} implies the definition
$\piset{0}{j} = \set{d \in \SetSp | 2^{j} \leq \mc_d < 2^{j+1}}$.
Furthermore, assuming that $\piset{\mrgline}{j}$ satisfies Eq.~\ref{eq:Nlj} up to level $\mrgline$,
\begin{equation}%\label{}
\notag
\begin{split}
& \piset{\mrgline+1}{j} =
  \piset{\mrgline}{2j} \cup \piset{\mrgline}{2j+1} = \\
& \set {d \in \SetSp | 2^{2j \cdot 2^\mrgline} \leq \mc_d < 2^{(2j+1) 2^\mrgline}}  \cup  \\
& \set {d \in \SetSp | 2^{(2j+1) \cdot 2^\mrgline} \leq \mc_d < 2^{(2j+2) 2^\mrgline}}  =  \\
& \set {d \in \SetSp | 2^{j \cdot 2^{\mrgline+1}} \leq \mc_d < 2^{(j+1) 2^{\mrgline+1}} }
\end{split}
\end{equation}
thus showing that level $\mrgline+1$ satisfies Eq.~\ref{eq:Nlj} as well.
\end{proof}

The following corollary shows that the special case of level $\mrgline = \log (\logbeta)$ the partition $\set{\piset{\mrgline}{j}}_j$ is identical to the datastores in $\SetSp$. As a result, an optimal solution for the sub-problem in level $\log (\logbeta)$ is actually an optimal full solution for the \cacheprob\  problem.

\begin{corollary}
$$
\piset{\log \logbeta}{0} = \set{d \in \SetSp | 2^0 \leq \mc_d < 2^{2^{\log \logbeta}}} = \SetSp,
$$
and
\begin{equation}\label{eq:Oset_final_lvl}
\OOset{\log \logbeta}{0} = \Oset \cap \SetSp = \Oset
\end{equation}
\end{corollary}

The following lemma shows that at each level $\mrgline$, \pgmalg\ has a candidate sub-solution with miss ratio, and access cost which are at most $\log \mbeta$, and 1 higher than those of the respective optimal sub-solution.
\begin{lemma}\label{lemma:pgm}
If $\mc_{\Oset} < \frac{\mbeta}{2 \cdot \log \mbeta}$ then for each $\mrgline = 0, \dots, \log \logbeta$ and $j = 0, \dots, \frac{\logbeta}{2^{\mrgline}}-1$,
$\mvec{\mrgline}{j}$ contains a set $X$ s.t.
$\mc_X \leq \mbase^{\mrgline+1} \cdot \mc_{\OOset{\mrgline}{j}}$ and
$\mr_X \leq  \mr_{\OOset{\mrgline}{j}}$.
\end{lemma}

\begin{proof}
Note first that for each $\mrgline$ and $j$ s.t. $\OOset{\mrgline}{j} = \emptyset$ the claim holds true since we can take $X=\emptyset$ which by lines~\ref{alg:pgm:genV0j} and~\ref{alg:pgm:init_to_emptyset_in_merge} is always a member of $\mvec{\mrgline}{j}$.
We may therefore focus our attention only on $\mrgline$ and $j$ for which $\OOset{\mrgline}{j} \neq \emptyset$.

We prove the claim by induction over $\mrgline$.
For the base case ($\mrgline=0$) we have to prove that if $\mc_{\Oset} < \frac{\mbeta}{2 \cdot \log \mbeta}$ then
for each $j = 0, 1, \dots \logbeta-1$ there exists a set of data stores
$X \in \mvec{0}{j}$ s.t.
$\mc_X \leq 2 \cdot \mc_{\OOset{0}{j}}$ and
$\mr_X \leq \mr_{\OOset{0}{j}}$.
Denote $k_j = \left\vert\OOset{0}{j}\right\vert$. As the access cost of each item in $\OOset{0}{j}$ is in $\left[2^j, 2^{j+1}\right)$, we have
\begin{equation}\label{eq:O_init_set_cost}
k_j \cdot 2^j \leq \mc_{\OOset{0}{j}}
\end{equation}
Denote by $D_j$ the $k_j$-size prefix of items in $\piset{0}{j}$, when sorted in non-decreasing order of miss ratio. Note that \pgmalg\ inserts $D_j$ to $\mvec{0}{j}$ (line~\ref{alg:pgm:genV0j}). 
The access cost of each item in $D_j$ is within $\left[2^j, 2^{j+1}\right)$, and hence
\begin{equation}\label{eq:init_prefix_cost}
\mc_{D_j} < k_j \cdot 2^{j+1}.
\end{equation}
Combining Equations~\ref{eq:O_init_set_cost} and~\ref{eq:init_prefix_cost}, we obtain $\mc_{D_j} \leq 2 \cdot \mc_{\OOset{0}{j}}$. 
Furthermore, as $D_j$ is a prefix of $\piset{0}{j}$ when sorted in a non-decreasing order of miss ratio and $|D_j| = k_j = |\OOset{0}{j}|$, we have $\mr_{D_j} \leq \mr_{\OOset{0}{j}}$, thus completing the proof of the induction's base.

For the induction step, we assume that the claim holds for level $\mrgline$, and prove it for level $\mrgline+1$.
Assume $\mc_{\Oset} < \frac{\mbeta}{2 \log \mbeta}$.
We have to show that there exists a set $X \in \mvec{\mrgline+1}{j}$ that satisfies
$\mc_X \leq \mbase^{\mrgline+2} \cdot \mc_{\OOset{\mrgline+1}{j}}$ and
$\mr_X \leq  \mr_{\OOset{\mrgline+1}{j}}$.

By the induction hypothesis, there exist sets $A^* \in \mvec{\mrgline}{2j}$ and $B^* \in \mvec{\mrgline}{2j+1}$ s.t.
$\mc_{A^*} \leq \mbase^{\mrgline+1} \cdot \mc_{\OOset{\mrgline}{2j}}$ and
$\mc_{B^*} \leq \mbase^{\mrgline+1} \cdot \mc_{\OOset{\mrgline}{2j+1}}$.
Consider the set $D^* = A^* \cup B^*$. We first show that $D^* \neq \emptyset$. 
Recall that $\OOset{\mrgline+1}{j} = \OOset{\mrgline}{2j} \cup \OOset{\mrgline}{2j+1}$.
Since it suffices to focus on the case where $\OOset{\mrgline+1}{j} \neq \emptyset$, we have either 
$\OOset{\mrgline}{2j} \neq \emptyset$ or $\OOset{\mrgline}{2j+1} \neq \emptyset$, 
and therefore either 
$\mr_{\OOset{\mrgline}{2j}} < 1$ or $\mr_{\OOset{\mrgline}{2j+1}} < 1$. By the induction hypothesis $\mr_{A^*} \leq \mr_{\OOset{\mrgline}{2j+1}}$ and 
$\mr_{B^*} \leq \mr_{\OOset{\mrgline}{2j}}$. Therefore, either $\mr_{A^*} < 1$ or $\mr_{B^*} < 1$. As a result, either $A^* \neq \emptyset$ or $B^*  \neq \emptyset$, and hence $D^* \neq \emptyset$.

Recall that $D^* = A^* \cup B^*$, and by definition 
$\OOset{\mrgline+1}{j} = \OOset{\mrgline}{2j} \cup \OOset{\mrgline}{2j+1}$.
Therefore, 
\begin{equation}\label{eq:cost_D_pre_step}
\begin{split}
& \mc_{D^*} \leq \mc_{A^*} + \mc_{B^*} \leq
\mbase^{\mrgline+1} \left(\mc_{\OOset{\mrgline}{2j}} + \mc_{\OOset{\mrgline}{2j+1}}\right) = \\
& 
\mbase^{\mrgline+1} \cdot \mc_{\OOset{\mrgline+1}{j}}
\leq
\mbase^{\mrgline+1} \cdot \mc_{\Oset}
< \mbase^{\mrgline+1} \cdot \frac{\mbeta}{2 \cdot \log \mbeta},
\end{split}
\end{equation}
where the second inequality is by the induction hypothesis.

Recalling that $\mrgline \leq \log \logbeta$ and $\logbeta = \log \mbeta$ we also have $\mbase^{\mrgline+1} \leq \mbase \cdot \mbase^{\log (\log \mbeta)} = 2 \cdot \log \mbeta$.
Combining the reasoning above we have $\mc_{D^*} < \mbeta$. 
Therefore (and recalling that $D^* \neq \emptyset$), there exists some $t \in \set{1, 2 \dots, \logbeta}$ s.t.
\begin{equation}\label{eq:cost_D_in_range_step}
2^{t-1} \leq \mc_{D^*} < 2^t.    
\end{equation}
As a result, one of the iterations of the merge loop (lines~\ref{alg:pgm:for_t_merge_begin}-\ref{alg:pgm:for_t_merge_end}) inserts to $\mvec{\mrgline+1}{j}$
a set $X_t$ s.t. 
\begin{equation}\label{eq:cost_X_in_range_step}
2^{t-1} \leq \mc_{X_t} < 2^t.    
\end{equation}
Combining (\ref{eq:cost_D_pre_step}), (\ref{eq:cost_D_in_range_step}) and (\ref{eq:cost_X_in_range_step}), we have 
$\mc_{X_t} <
\mbase \cdot \mc_{D^*}
\leq
\mbase^{\mrgline+2} \cdot \mc_{\OOset{\mrgline+1}{j}}
$.
Furthermore, by lines~\ref{alg:pgm:merge_find_argmin}-\ref{alg:pgm:merge_insert_to_vec}, $X_t$ minimizes the miss ratio; and by 
the induction hypothesis, $\mr_{A^*}
\leq \mr_{\OOset{\mrgline}{2j}}$ and $\mr_{B^*} \leq \mr_{\OOset{\mrgline}{2j+1}}$. We conclude that 
$\mr_{X_t} \leq \mr_{D^*} =
\mr_{A^*} \cdot \mr_{B^*} \leq
\mr_{\OOset{\mrgline}{2j}} \cdot \mr_{\OOset{\mrgline}{2j+1}} =
\mr_{\OOset{\mrgline+1}{j}}
$.
\end{proof}

We are now in a position to upper-bound the approximation ratio and run-time of \pgmalg.
\begin{theorem}
\pgmalg\ is a $2 \cdot \log (\mbeta)$-approximation algorithm, running in time $O\left(\NumSp^2 + \log^3 (\mbeta) \right)$.
\end{theorem}

\begin{proof}
\emph{Approximation ratio.} Consider an optimal solution $\Oset$.
Recall that in every iteration -- and, in particular, in the last iteration -- \pgmalg\ considers using the empty set.
Hence, $\phi (\tilde{D}) \leq \phi (\emptyset) = \mbeta$. Therefore, if $\mc_{\Oset} \geq \frac{\mbeta}{2 \cdot \log \mbeta}$, the claim is true.

If $\mc_{\Oset} < \frac{\mbeta}{2 \cdot \log \mbeta}$, then by Lemma~\ref{lemma:pgm} there exists a set $X \in \mvec{\log \logbeta}{0}$ s.t.
$\mc_X \leq \mbase^{\log \logbeta+1} \cdot \mc_{\OOset{\log \logbeta}{0}} = 2 \cdot \log \mbeta \cdot \mc_{\OOset{\log \logbeta}{0}}$. 
Using Eq.~\ref{eq:Oset_final_lvl}, we obtain
\begin{equation}\label{eq:pgm:costs_approx}
\mc_X \leq 2 \log \mbeta \cdot \mc_{\Oset}.
\end{equation}
Furthermore, by Lemma~\ref{lemma:pgm}
\begin{equation}\label{eq:pgm:mr_approx}
\mr_X \leq \mr_{\Oset}.
\end{equation}

Combining Equations~\ref{eq:pgm:costs_approx} and~\ref{eq:pgm:mr_approx} and recalling that $\tilde{D}$ minimizes $\phi$ over $\mvec{\log \logbeta}{0}$,
we obtain $\phi(\tilde{D}) \leq \phi(X) \leq 2 \log \mbeta \cdot \phi (\Oset)$.

\emph{Run time.} 
Lines~\ref{alg:pgm:partition_for_start}-\ref{alg:pgm:partition_for_end} require only a single sort of the $\NumSp$ data stores with positive indications, which takes $O(\NumSp \log \NumSp)$ time.

When $\mrgline=1$, the worst-case run time of lines~\ref{alg:pgm:for_line_for_j_begin}-\ref{alg:pgm:for_line_for_j_end} occurs when there exists $j$ s.t. $|\piset{0}{2j}| = O(\NumSp)$ and $|\piset{0}{2j+1}| = O(\NumSp)$, in which case considering all the nodes of the full bipartite graph between $\piset{0}{2j}$ and $\piset{0}{2j+1}$ requires $O(\NumSp^2)$ steps.

For each $\mrgline = 2, 3, \dots, \log (\logbeta)$, \pgmalg\ inserts to $\mvec{\mrgline}{j}$ only at most one node for every $t = 1, 2, \dots, \logbeta$. Therefore when $\mrgline > 1$, each 
iteration of the merge block (lines~\ref{alg:pgm:for_t_merge_begin}-\ref{alg:pgm:for_t_merge_end}) requires merging a full bipartite graph where the number of nodes in each side is at most $O(\logbeta)$. Hence, at each iteration of the merge block \pgmalg\ considers $O(\logbeta^2)$ pairs.
As the merge tree contains $O(\logbeta)$ nodes (recall Fig.~\ref{fig:pgm_run_example:mrg_no_zoom}), \pgmalg\ performs $O(\logbeta)$ such merge operations. Thus, the time required to run lines~\ref{alg:pgm:for_line_begin}-\ref{alg:pgm:for_line_end} when $\mrgline>1$ is~$O(\logbeta^3)$.

Summing the complexity expressions above, the total time complexity is $O(\NumSp^2 + \logbeta^3)$.
\end{proof}

Table~\ref{table:algs_approx_n_run_t} compares the approximation guarantees and the run-times of our algorithms. For \ppapproxalg\ we use a recent $(1+\mbeta)$-approximation for the Knapsack problem, whose running time is near $O \left( \NumSp + \epsilon^{-2.4}\right)$~\cite{FPTAS_Knap}.\footnote{The exact running time of this Knapsack FPTAS is: $O \left( 
\NumSp \log \frac{1}{\epsilon} + \left( \frac{1}{\epsilon} \right)^{2.4} / 2^{\Omega \left( \sqrt{\log (1/\epsilon)}\right)}
\right)
$}
Our algorithms suggest various trade-offs between run-time and approximation guarantees. Furthermore, in the general case our algorithms are not comparable to each other in terms of both run-time and  approximation ratio, and the choice of an algorithm is dependent upon the relations between the parameters $\mbeta, \NumSp$ and ${\mc_j}_j$ in a concrete system. 

\begin{table}
		\label{table:algs_approx_n_run_t}
    \begin{tabular}{l || l | l | l}
        Alg & Approx. & Run time & Based on \\
        \hline\hline
		\rule{0pt}{3ex}\ppapproxalg& $O\left(\beta^{\frac{\epsilon}{1+\epsilon}}\right)$ & $O\left( \left(\NumSp + \epsilon^{-2.4}\right)\cdot M\right)$ & Knapsack fully-polynomial time approximation scheme  \\ 
		\rule{0pt}{3ex}\umb& $ O\left(\sqrt{\beta}\right)$ & $O \left(\NumSp^2 \log \NumSp\right)$ & Knapsack 2-approximation \\
        \rule{0pt}{3ex}\pot& $\max_j \{C_j\}$ & $O \left(\NumSp \log \NumSp \right)$ & Potential function \\
        \rule{0pt}{3ex}\pgmalg & $O (\log (\mbeta))$ & $O\left(\NumSp^2 + \log^3 (\mbeta) \right)$ & Divide \& conquer
    \end{tabular}
		\caption{Approximation guarantees and run-time of proposed algorithms}
\end{table}

\makeatletter
\newcommand*\bigcdot{\mathpalette\bigcdot@{.5}}
\newcommand*\bigcdot@[2]{\mathbin{\vcenter{\hbox{\scalebox{#2}{$\m@th#1\bullet$}}}}}
\makeatother

\section{Simulation Study}\label{sec:sim}

This section uses a real access trace and a real content distribution network topology to provide insights into the performance of various access strategies in versatile settings.  

\subsection{System Topology and Costs}\label{sec:sim_settings_topology_costs}
We use the topology of the OVH~\cite{ovh} content distribution network.
The OVH network~\cite{ovh} includes 19 {\em Points of Presence (PoPs)} in Europe and North America along with the available bandwidth between PoPs.
We interpret each PoP as containing both a data store and a co-located client.
Queries are generated at clients and each such query triggers an access to a subset of the data stores according to the prescribed policy.

We assume that clients use the shortest hop-count path between their location and the data store they access.
Ties are broken by picking the path with maximal bottleneck link bandwidth. The cost for a client located at
node $i$ to access a data store at node $j$ is:
\begin{equation}\label{simulation_cost_function}
c_{i,j} = \Biggl\lceil{1 + \alpha \cdot \dist (i, j) + (1-\alpha) \cdot \frac{T}{\BW (i, j)}\Biggr\rceil},
\end{equation}
where
\begin{inparaenum}[(i)]
	\item $\dist (i,j)$ is the hop-count between node $i$ and node $j$, where $\dist(i,i)=0$,
	\item $\BW (i,j)$ is the maximum bottleneck bandwidth of a minimum length path from node $i$ to node $j$, where $\BW(i,i)=\infty$,
\item $T$ is a design parameter satisfying $T \geq \max_{i,j} \BW(i,j)$, that relates the increased cost of
    having a smaller bandwidth with the increased cost due to having a higher hop-count.
Last,
	\item $\alpha$ is a design parameter that helps balance the effects of hop-count distance and bottleneck bandwidth on the cost.
\end{inparaenum}
 In particular, for $\alpha=1$ the cost is
fully dominated by the hop-count distance and for $\alpha=0$ it is fully dominated by the bottleneck
bandwidth, regularized by the parameter $T$. Unless stated otherwise, throughout our simulations we set
$T=\max_{i\neq j}\BW(i,j)$. Specifically, $T=500$ for the OVH network.

Figure~\ref{fig:PDF_of_costs} presents the histogram of the default access cost used in our evaluation between all pairs of clients and data stores in the OVH network.

\setlength{\belowcaptionskip}{-7pt}

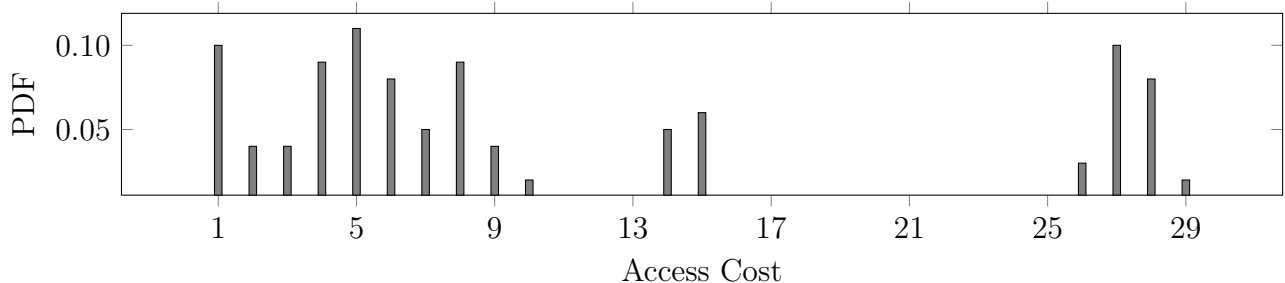
\begin{figure}[h]
	\centering
	\begin{tikzpicture}
	\begin{axis}[
	ybar,
	bar width = 0.1cm,
	width=\columnwidth,
	height=4cm,
	xtick={1,5,9,13,17,21,25,29},
	xlabel near ticks,
    ylabel near ticks,
    xlabel= Access Cost,
	ytick = {0.05, 0.1}, %empty,
	y tick label style={
		/pgf/number format/.cd,
		fixed,
		fixed zerofill,
		precision=2
	},
	ylabel=PDF]
	\addplot[fill=gray] plot coordinates {
		(1, 0.10)
		(2, 0.04)
		(3, 0.04)
		(4, 0.09)
		(5, 0.11)
		(6, 0.08)
		(7, 0.05)
		(8, 0.09)
		(9, 0.04)
		(10, 0.02)
		(14, 0.05)
		(15, 0.06)
		(26, 0.03)
		(27, 0.10)
		(28, 0.08)
		(29, 0.02)
	};
	\end{axis}
	\end{tikzpicture}
%}
	\caption{Histogram of $c_{i,j}$ values for the OVH network, based on Eq.~\ref{simulation_cost_function}, using $\alpha=0.5$ and $T=500$.}
	\label{fig:PDF_of_costs}
	\vspace{\vspacebelowcaption}
\end{figure}

\subsection{Data Store Characteristics}\label{sec:sim_settings_data_store}
Data stores are initially empty, and each can contain a maximum of $S$ data elements.
Once an item is added to a full data store, it evicts an item according to the Least Recently Used (LRU) policy.
%In order to obtain insights with regard to the impact of the cache sharing scheme we consider various cache insertion policies, detailed in the further.
The indicators are implemented using Counting Bloom Filters~\cite{CBF}, each consisting of $B(S)$ 8-bit counters and $5$ hash functions, where $B(S)$ is chosen as the number of counters required to obtain a target false positive ratio of $0.02$~\cite{Survey04}.
For example, in most of our simulations, we set $S=1000$, which implies $B(S)=8181$.
%\itamar{The calculation is: $(1 - (1 - 1/8000)^6000)^6 = 0.0216$}
We assume that up-to-date indicators are available at all time as can be efficiently realized by compressed Bloom filters~\cite{CompressedBF}, or by only transmitting the changes as in~\cite{CDN_OceanStore}.

Each data store estimates its own misindication ratio by evaluating an exponential moving average over epochs
of $R$ requests made to the data store. Formally, let $m_j(s,t)$ denote the number of misses occurring at
data store $j$ during the requests $s+1,\ldots,t$ made to data store $j$\footnote{Recall that we only
access a data store if it has provided a positive indication.}. For any $t \leq R$ we let the estimated
misindication ratio after handling request $t$ be $\rho_j(t)=\frac{m_j(0,t)}{t}$. For $t > R$, we let
$\rho_j(t)$ be the most recent estimate over epochs of $R$ requests, $\rho_j(\floor{t/R}\cdot R)$, where for
every non-negative integer $k$ this estimate is updated after handling request $(k+1)R$ such that
$\rho_j((k+1)R)=\delta \cdot m_j(kR,(k+1)R)/R + (1-\delta) \cdot \rho_j(kR)$. In our simulations, we take
$\delta=0.1$ and $R=100$, as we found this configuration to yield a stable $\rho$ at each data store and
to work well in practice.

We consider a \emph{system-wide request distribution policy} where an item can only be
placed in $k$ data stores that are chosen by a hash function based on the requests' content. Such a policy is
inspired by ideas such as replication and partitioning to increase the hit ratio~\cite{Kaleidoscope}. We increased $k$ up to 5, which make $26\%$ of the 19 datastores in the system. 
%at each {\em client} not change dramatically over short enough horizons, which help bring out the benefits of using cached choices, as described in Section~\ref{*****}.

\subsection{Traffic Trace, Metrics, and Simulated Scenarios}\label{sec:sim_settings_traffic_trace}
We used a publicly available Wikipedia trace~\cite{WikiBench} consisting of 357K read requests to Wikipedia pages during a 5 minute period\footnote{The trace includes requests made on Sep. 22, 2007, from 06:12 to 06:17}.
% \gabi{The trace's URL is: http://www.wikibench.eu/wiki/2007-09/wiki.1190448987.gz}
Each request in this trace is assigned to a random client issuing the request, and requests appear according to their order in the trace.
For handling the requests, we consider the following access policies applied by the clients for choosing the set of data stores to access:
\begin{inparaenum}[(i)]
	\item \cpi,
	\item \epi,
	\item \umb,
	\item \pot, and
	\item \pgmalg.
\end{inparaenum}
The evaluation factors the total cost, where all clients are running the same algorithms.
We also considered the benchmark performance provided by using perfect indicators (PI). This benchmark is used to normalize the costs of the various policies considered.
We measure the \emph{total cost} (TC) incurred by each access strategy for serving the entire trace. 
We normalize the TC of each access strategy by the TC of the perfect indicator PI. This normalization is aimed to compare the performance in various settings, while alleviating some of the exogenous effects specific to the scenario being evaluated.

\subsection{Heterogeneous Case (OVH network)}
\begin{table}[t]
	\begin{center}
		\begin{tabular}{|l|l||c|c||c|c||c|c|}
			\hline
			\multirow{2}{*}{$\mbeta$} & \multirow{2}{*}{Policy} & \multicolumn{2}{c||}{{\bf 1 location}} & \multicolumn{2}{c||}{\textbf{3 locations}} & \multicolumn{2}{c|}{\textbf{5 locations}}\\
			\cline{3-8}
			& & AC & TC & AC & TC & AC & TC\\
			\hline
			\hline
			\multirow{5}{*}{$10^2$}
            & \cpi & 0.16 & 1.20 & 0.10 & 1.11 & 0.08 & 1.08 \\
            & \epi & 0.23 & 1.09 & 0.43 & 1.39 & 0.49 & 1.55 \\
            & \umb & 0.19 & 1.10 & 0.16 & 1.10 & 0.13 & 1.09 \\
            & \pot & 0.20 & 1.11 & 0.18 & 1.13 & 0.20 & 1.16 \\
            & \pgmalg & 0.19 & 1.10 & 0.16 & 1.11 & 0.14 & 1.09 \\
			\hline
			\multirow{5}{*}{$10^3$}
            & \cpi & 0.02 & 1.22 & 0.01 & 1.10 & 0.01 & 1.07 \\
            & \epi & 0.03 & 1.01 & 0.05 & 1.05 & 0.07 & 1.08 \\
            & \umb & 0.03 & 1.01 & 0.03 & 1.04 & 0.02 & 1.03 \\
            & \pot & 0.03 & 1.01 & 0.04 & 1.04 & 0.03 & 1.04 \\			
            & \pgmalg & 0.03 & 1.01 & 0.03 & 1.04 & 0.03 & 1.03 \\
            \hline  
 			\multirow{5}{*}{$10^4$}
            & \cpi & 0.00 & 1.22 & 0.00 & 1.10 & 0.00 & 1.07 \\
            & \epi & 0.00 & 1.00 & 0.00 & 1.02 & 0.01 & 1.02 \\
            & \umb & 0.00 & 1.00 & 0.00 & 1.02 & 0.00 & 1.02 \\
            & \pot & 0.00 & 1.00 & 0.00 & 1.02 & 0.00 & 1.02 \\
            & \pgmalg & 0.00 & 1.00 & 0.00 & 1.02 & 0.00 & 1.02 \\
            \hline
		\end{tabular}
		\caption{OVH network simulation. Results present for every scenario and every policy the Access Cost (AC) and Total Cost (TC). The values are normalized by the TC of the Perfect Indicators policy.}
		\label{tab:fix_sim_results}
	\end{center}
\end{table}

In our first experiment, we compare the performance of various access strategies when varying the number of locations per item $k$ and the miss penalty $\mbeta$. For each configuration, we measure the normalized total cost (TC). 
Recall that the total cost (TC) is the sum of the access cost and the miss cost. Hence, for obtaining better insight of the dominant source for the cost of an access \redtext{strategy}, we show for each access strategy also its \emph{access cost} (AC). 
We normalize both the AC and the TC of each access strategy by the TC of the perfect indicator PI. 
The outcome of this evaluation is provided in Table~\ref{tab:fix_sim_results}, where we present the PI normalized
results for various $\mbeta$ and $k$ values. 

The results show that that \cpi\ has the minimal AC in all scenarios, as could be expected by its definition. The AC obtained by \cpi\ decreases when $k$ is increased, because when an item is found in multiple datastores, \cpi\ is more likely to pick from them a datastore with low access cost.
However, \cpi\ is highly sensitive to false positives, which are translated to high TC. This effect is more prominent when $k$ is low, implying that there exist only a few true positive indications.

\epi, on the other hand, is very effective for $k=1$ but becomes less attractive as we increase $k$, due to the fact
it ends up accessing too many data stores, which is captured by increasing AC. 
This effect is mitigated when $\mbeta$ is high, because a high miss penalty implies that one should better access multiple datastores in aim to minimize the probability of a miss, even at the cost of some unnecessary accesses.

Our three proposed algorithm -- \pot, \umb\ and \pgmalg\ -- outperform the two heuristics (\cpi\ and \epi) in all configurations, with the only of exception of $\beta = 100$ and $k=1$, where \cpi\ obtains slightly lower cost. Focusing on our three algorithms, \pot\ does slightly worse than \umb\ and \pgmalg, which can be explained by the higher AC of \pot. Intuitively, \pot\ optimizes for reducing the miss cost, even at the cost of a slightly higher access cost. 
\umb\ and \pgmalg\ are the best strategies in almost all scenarios, but most importantly, they are never  bad strategies. Thus, even when they underperform compared to some other strategy, the difference is marginal. 

\begin{table}[t]
	\begin{center}
		\begin{tabular}{|l||c|c|c|c||c|c|c|c|}
			\hline
			\multirow{2}{*}{Policy} & \multicolumn{4}{c||}{{\bf 1 location, varying $\fpr$}} &  \multicolumn{4}{c|}{\textbf{5 locations, varying $\fpr$}}\\
			\cline{2-9}
			& 0.01 & 0.02 & 0.03 & 0.04 & 0.01 & 0.02 & 0.03 & 0.04\\
			\hline
			\hline
            \cpi & 1.11 & 1.20 & 1.29 & 1.35 & 1.04 & 1.08 & 1.12 & 1.15 \\
            \epi & 1.04 & 1.09 & 1.13 & 1.17 & 1.51 & 1.55 & 1.60 & 1.64 \\
            \umb & 1.06 & 1.10 & 1.14 & 1.18 & 1.05 & 1.09 & 1.12 & 1.14 \\
            \pot & 1.06 & 1.11 & 1.16 & 1.21 & 1.12 & 1.16 & 1.19 & 1.23 \\
            \pgmalg & 1.06 & 1.10 & 1.14 & 1.18 & 1.06 & 1.09 & 1.12 & 1.15 \\
			\hline
		\end{tabular}
		\caption{OVH network simulation with varying false positive (FP) ratios. The miss penalty is set to $\mbeta=100$.}
		\label{tab:sim_fpr}
	\end{center}
\end{table}

Our next experiment explores the effect of the False Positive (FP) ratio on the performance of different access strategies. 
As our previous experiment implied that \cpi, and \epi\ does best when $k=5$, and $k=1$, we focus hereafter on these two values of $k$.
The results are shown in Table~\ref{tab:sim_fpr}. 
The results show that \cpi\ and \epi\ are highly sensitive to the FP ratio. For instance, when $k=1$ and $FP = 0.04$, \cpi\ incurs an excessive cost of $35\%$ above \opt. When $k=5$ and $FP = 0.04$, \cpi\ incurs an excessive cost of $64\%$ above \opt. In contrast, all our proposed algorithms -- \umb, \pot\ and \pgmalg\ -- show only a mild increase in the cost when incrementing FP. In particular, \umb\ and \pgmalg\ obtain again minimal, or close to minimal, costs across the board. 

\begin{figure*}[ht!]
	\centering
	\subfloat[\label{sim:homogenenous:kloc_1}
	1 location
	]{
		%\resizebox{\columnwidth}{0.5\columnwidth}{
\begin{tikzpicture}
	\begin{axis}[
		width=0.9\columnwidth,
		height=0.4\columnwidth,
		legend style = {at={(0.5,1.23)},anchor=north,legend columns=-1},
		ylabel=Normalized Cost,
		ylabel near ticks,
		xlabel= Data Store Size,
		xlabel near ticks,
		xtick={200,550,900,1250,1600},
		ymin=0.99,
		ymax=1.25,
		ytick={1,1.05,1.10,1.15,1.20, 1.25},
        xmin = 180,
        xmax =1620,
        every axis plot/.append style={very thick},
		]
	\addplot[color=red,mark=triangle, mark size=\marksize] coordinates {
        (200, 1.172458893088127) 
        (400, 1.1891175910905962) 
        (600,1.201419114155504) 
        (800, 1.2038155224461218) 
        (1000, 1.2154734756402474) 
        (1200, 1.2144604632517249) 
        (1400, 1.2130999386197703) 
        (1600, 1.2182128334139808)
	};	\addlegendentry{\cpi}	
	
\addplot[color=purple,mark=o, mark size=\marksize] coordinates {
        (200, 1.015990103384882)
        (400, 1.014576425448219)
        (600, 1.0134752522872053)
        (800, 1.0123316108716378)
        (1000, 1.0118013670890589)
        (1200, 1.0122193854347152)
        (1400, 1.0100786990987487)
        (1600, 1.0102553005476238)
	}; \addlegendentry{\epi}
	
\addplot[color=cyan,mark=x, mark size=\marksize]    coordinates {
        (200, 1.0159879233659628) (400, 1.0145729272122648) (600, 1.0135052171152816) (800, 1.012379430104598) (1000, 1.0118333179850207) (1200, 1.012392496924498) (1400, 1.0102636162354386) (1600, 1.0103483951422827)}; \addlegendentry{\umb, \pot, \pgmalg}
\end{axis}
\end{tikzpicture}
%}
	}
\hfill
	\subfloat[\label{sim:homogenenous:kloc_5}
	5 locations
	]{
		%\resizebox{\columnwidth}{0.5\columnwidth}{
\begin{tikzpicture}
	\begin{axis}[
		width=0.9\columnwidth,
		height=0.4\columnwidth,
		legend style = {at={(0.5,1.23)},anchor=north,legend columns=-1},
		ylabel=Normalized Cost,
		ylabel near ticks,
		xlabel= Data Store Size,
		xlabel near ticks,
		xtick={200,550,900,1250,1600},
		ymin=0.99,
		ymax=1.25,
		ytick={1,1.05,1.10,1.15,1.20, 1.25},
        xmin = 180,
        xmax =1620,
        every axis plot/.append style={very thick},
		]
	\addplot[color=red,mark=triangle, mark size=\marksize] coordinates {
        (200, 1.0571042491358278) (400, 1.0613575857932842) (600, 1.0685655547887776) (800, 1.068415187428904) (1000, 1.0700426105927774) (1200, 1.0693293902841547) (1400, 1.0726655861863812) (1600, 1.0751819133860918)
	};  \addlegendentry{\cpi}	
    \addplot[color=purple,mark=o, mark size=\marksize] coordinates {
        (200, 1.04337064502995) (400, 1.0508111025601496) (600, 1.0579700726622467) (800, 1.0611142341913353) (1000, 1.061730746507961) (1200, 1.0639417949077872) (1400, 1.0661349844349326) (1600, 1.0671175456603568)
	};  \addlegendentry{\epi}

    \addplot[color=cyan,mark=x,mark size=\marksize]    coordinates {
        (200, 1.0252096635675578) (400, 1.0266513608332049) (600, 1.0293537976857638) (800, 1.0289561894306145) (1000, 1.0286853470329471) (1200, 1.028690725243771) (1400, 1.029228787826495) (1600, 1.028819868452436)
    };  \addlegendentry{\umb, \pot, \pgmalg}
\end{axis}
\end{tikzpicture}
%}
	}
	\caption{\label{cacheSize}Homogeneous network with varying data store size. The miss penalty is set to $\mbeta=100$, and the target false positive ratio is $0.02$.}
\end{figure*}
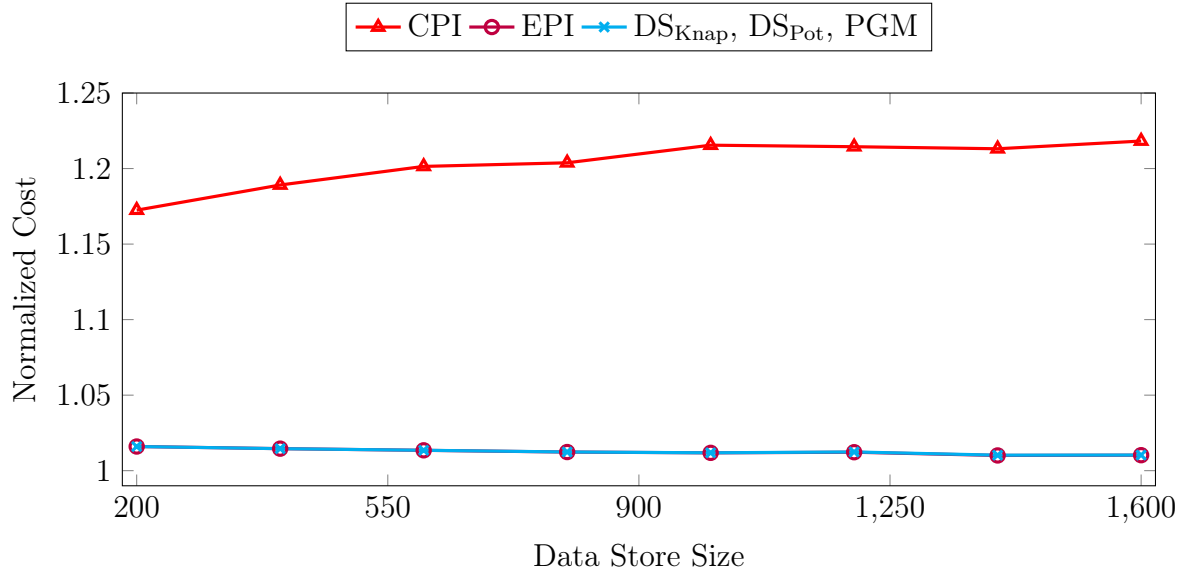
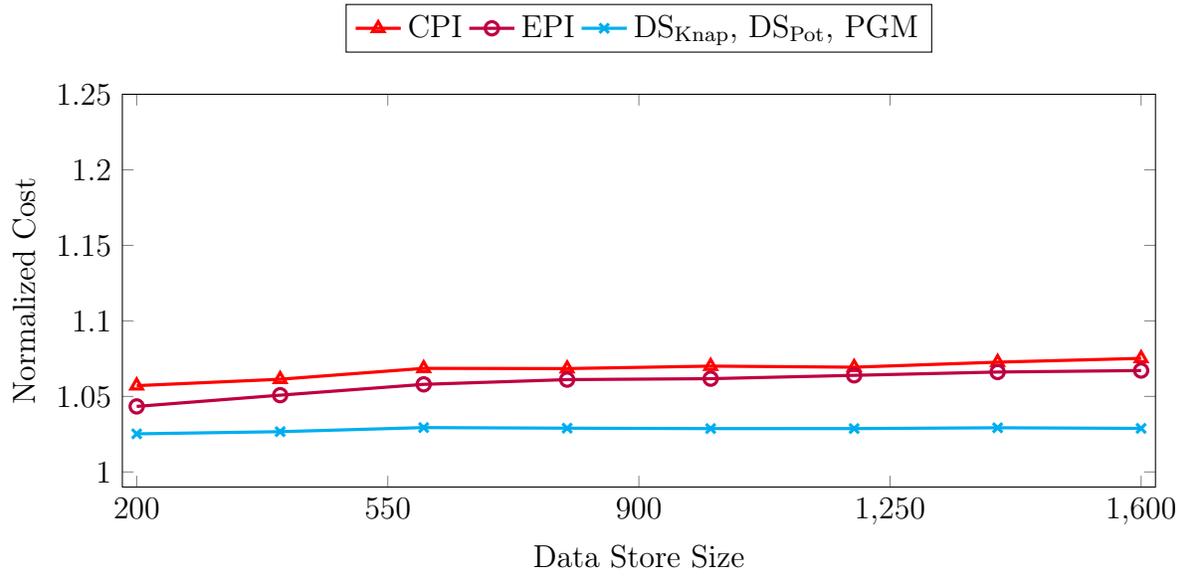

\subsection{Homogeneous Case: Varying Data Store Size}
Our next experiment considers homogeneous settings, where the access costs of all 19 datastores in fixed 1. We aim at studying the effect of the datastore size in these settings. 
Figure~\ref{cacheSize} shows the results for $k=1$ and $k=5$ locations per item, where we vary the size of each datastore from 200 up to 1600. In these homogeneous cost settings, all our algorithms -- namely, \pot, \umb\ and \pgmalg\ -- are equivalent to the scheme which minimizes the expected overall cost, \fpo. Furthermore, their performance is always very close to the one achieved with perfect indicators. 
In contrast to our previous experiments, \cpi\ does not do very well even with $k=5$ locations per item. The reason is that in such homogeneous settings, when there exist multiple positive indications none of them is ``cheapest''. As a result, \cpi\ merely randomly selects a single datastore -- that is, \cpi\ always accesses a datastore which neither minimizes the miss ratio, nor minimizes the access cost. 
The results of the homogeneous case show again that the existing heuristics are too simplistic to fit all system configurations, thus motivating the need for our algorithms.

% ================================================
% Section
% ================================================

\section{Discussion}\label{sec:discussion}

Our work closes an important knowledge gap concerning indicator based caching in network systems. Namely, it
answers the fundamental question of providing a stable access strategy that achieves near-optimal results in
a wide variety of scenarios.

Our work starts by showing that the access strategy problem was roughly ignored until now and that the
existing solutions are only attractive for some system parameters. That is, their effectiveness is determined
by uncontrolled variables that may change throughout the system's life, and may not be known in advance. In
contrast, the algorithms suggested in this work provide provable approximation ratios to the optimal solution
and are shown to be near-optimal in a variety of system settings.

As future work, we aim at studying the trade-off between the bandwidth used by indicators and the miss cost incurred by their false indications.

\newpage
\chapter{Virtual Machine Placement in Virtual Network Functions}\label{sec:APSR}

% PAPER SPECIFIC -------------------------------------------------
\newcommand{\pap}{ParalAlocProc}
\newcommand{\dr}{\eps} % Decline ratio
\newcommand{\nbins}{n}   % Number of rsrcs
\newcommand{\nfreeslots}{f} % (minimal) number of free slots
\newcommand{\nfreebins}{k} % (minimal) number of free bins
\newcommand{\nsched}{s}  % number of schedulers (allocators)
\newcommand{\nhappy}{H_{\nsched}} % (expected) number of happy schedulers, namely, schedulers which successfully allocates a VM
\newcommand{\hk}{H^{k}_{\nsched}} % rv for number of happy agents when there're k free bins
\newcommand{\hkpp}{H^{k+1}_{\nsched}} % rv for number of happy agents when there're k+1 free bins
\newcommand{\phk}{F^{k}_{\nsched}} % rv for number of potentially happy agents when there're k free bins
\newcommand{\phkpp}{F^{k+1}_{\nsched}} % rv for number of potentially happy agents when there're k+1 free bins
\newcommand{\ehk}{E[H^{k}_{\nsched}]} % rv for number of happy agents when there're k free bins
\newcommand{\hkkpp}{E[H^{k+1}_{\nsched}]} % rv for number of happy agents when there're k+1 free bins
\newcommand{\dif}{D} % dif (k, f) = E \left[\hk | \phk = f+1\right] - E [\hk | \phk = f]
\newcommand{\npothappy}{F_{\nsched}} % An rv, representing the # of schedulers which sample at least one free bin
\newcommand{\PrFF}{\sigma} % Prob' for a sched' to Find a Free (i.e., not totally full) bin.
\newcommand{\budget}{B} %sd budget
\newcommand{\estkalg}{EstimateK}
\renewcommand{\alg}{MaximizeParallelism} %Alg' to maximize parallelism in the balls-and-bins model
\newcommand{\algtopKavg}{\algtop$_{avg}$}
\newcommand{\rand}{Rand}
\newcommand{\host}{\vec{h}} 
\newcommand{\hosts}{\mathbf{H}} 
\newcommand{\request}{\vec{r}}
\newcommand{\requests}{\mathbf{R}}
\newcommand{\flavors}{\mathbf{C}}
\newcommand{\flavor}{\vec{c}}
\newcommand{\dqueries}{d}
\newcommand{\availhosts}{\hosts}
\setlength{\textfloatsep}{10pt plus 1.0pt minus 2.0pt}
\setlength{\dbltextfloatsep}{10pt plus 1.0pt minus 2.0pt}

\section{Problem Overview}
\label{sec:apsr:intro}
The {\em Network Function Virtualization (NFV)} paradigm deploys virtual machines for running network functions such as firewalls, deep packet inspection, load balancing and monitoring. NFV enables online deployment of network services and scaling of such services according to the current workload requirements, without relying of concrete physical  middleboxes~\cite{Middleboxes,EASE}.
These features should presumably improve the overall system performance in various perspectives, including throughput and latency.
Unfortunately, these improvements are not manifested in large clouds, as we see and discuss in the sequel.

To understand why these benefits do not scale to large cloud environments, it is instructive to consider the process of VM deployment in the cloud. Once the user issues a request to allocate a new {\em Virtual Machine (VM)}, a resource management algorithm, commonly referred to as a {\em scheduler}, selects a host on which to accommodate the VM which is then deployed on the chosen host. While the deployment time of optimized VMs or containers (e.g., using Kubernetes~\cite{Kubernetes}) can be performed within tens of milliseconds~\cite{ClickOS}, selecting a host on which to place the VM may require hundreds of milliseconds in large clouds~\cite{OSProblem,KubeProblem,ASC}.
It follows that the potential performance boost of using NFV remains largely unfulfilled in large clouds due to bottlenecks in scheduling deployment requests.

The main reason that deciding on which host to place the VM takes so long is that most current 
resource management algorithms~\cite{Power,Power2,Switching,Yaniv2,FaultTolerant,RazPlacement,SHABEERA2017616,Yao2013} require complete information about the availability of resources on the system's hosts.
In a large cloud, gathering the current state from hundreds and sometimes thousands of hosts translates to high communication overhead, resulting in a performance bottleneck~\cite{OSProblem,ASC, KubeProblem}.

Intuitively, one could address this handicap by running multiple schedulers in parallel. However, such an approach may translate to having such multiple schedulers try and place requests simultaneously on the same host, leading to race scenarios~\cite{Omega, Sparrow, host_subset_size}.
In such cases, not all deployments will be successful, and some of the requests placed on the same host may be rejected, or {\em declined}.
However, a provider is typically required to satisfy a {\em Service Level Agreement (SLA)} which bounds the 
\textcolor{red}{the ratio between the number of requests that are declined and the total number of requests (the {\em decline ratio}~\cite{NFV_RA_survey16}).}

An efficient placement algorithm should therefore strive to increase parallelism, while maintaining a low communication overhead and bounded decline ratio.
% higher parallelism may boost performance - but in the cost of higher communication overhead and decline ratio. 
However, to the best of our knowledge, no previous work has studied the interplay between parallelism and decline ratio in VM placement.

\subsection {Related Work}
\label{sec:apsr:rltd_work}

This section provides background on the way scheduling works in the OpenStack platform, and discusses related work addressing various aspects of VM placement.

\paragraph*{OpenStack:} OpenStack is a popular and widely used open source cloud management platform~\cite{6234218,Van2016} that manages compute, storage and network resources. 
It is composed of an ensemble of sub-projects, where scheduling is implemented within the Nova project~\cite{Nova}. 
% Nova facilitates the entire process of VM deployment, starting with having a user issue a request to deploy a VM, continuing with a scheduler selecting a compute node, following up with deploying the VM on the selected node and eventually starting the VM on the host. 
Upon receiving a user request, Nova selects a host for it, and places a VM on the selected host.
The default scheduler is called a \emph{Filter Scheduler}~\cite{filterscheduler}, which goes through the following sequence of stages upon the arrival of each request:
\begin{inparaenum}[(i)]
\item the \emph{State} stage, where the scheduler gathers the state of the available resources in each host, followed by
\item the \emph{Filter} stage, where the scheduler goes over all the hosts reports and filters out the hosts that cannot satisfy the request, and finally
\item the \emph{Weight} stage, where the scheduler selects one of the hosts that can satisfy the request, according to some weight function applied on the hosts' states.
\end{inparaenum} 

\paragraph*{Placement algorithms:} There is a large body of work that deals with placement of user requests in cloud environments~\cite{FaultTolerant,Power,Power2,f4,RazPlacement,VirtualRack,Yao2013,SHABEERA2017616,Ari02whois,Xiao13}.
Most works vary from one another by the nature of the optimized performance metric.
Examples of such metrics include minimizing the number of utilized hosts to save power~\cite{Power2}, minimizing migration overheads~\cite{Power2}, improving
fault tolerance~\cite{FaultTolerant}, minimizing NFV switching overheads~\cite{Switching,Yaniv2}, optimizing host utilization~\cite{RazPlacement,bup}, and studying the impact of network bandwidth on VM placement~\cite{VirtualRack,Yao2013,SHABEERA2017616}.

\paragraph*{Communication overhead:} All the above algorithms implicitly assume full and up-to-date information about the available resources in all hosts, which coincides with the approach of OpenStack's Filter scheduler, which queries all the hosts before considering any specific placement concerns.
However, attaining such a complete state information incurs high communication overhead, resulting in performance bottlenecks~\cite{OSProblem,KubeProblem}. 
The recently introduced \emph{Adaptive Scheduler Cache}~\cite{ASC} aims at decreasing the communication overhead by using a cached system state.
Their method is shown to improve OpenStack's throughput, but they do not provide any guarantees on the system's performance, or on the decline ratio of scheduling decisions. Their proposed approach is essentially orthogonal to those of applying randomness and parallelism (see below), and thus can be deployed alongside with our method. 

\paragraph*{Parallelism:} OpenStack traditionally used a single scheduler, but the community is exploring ways to increase performance, and 
parallelism is suggested as a straightforward technique. 
However, simply running multiple identical independent schedulers may translate to numerous collisions between several schedulers simultaneously trying to place requests on the same host, which results in race conditions~\cite{host_subset_size, Sparrow} and high decline ratio. 
The Omega scheduler~\cite{Omega} mitigates this problem using shared state information and atomic updates. However, in contrast to our work, Omega requires complex synchronization mechanisms and high communication overhead and does not provide provable guarantees on the decline ratio. 

\paragraph*{Randomness:}
A large body of work has considered random approaches to balanced allocations~\cite{P2_azar,parallel_P2,P2}. These works focused on decreasing communication overhead while keeping provable strong guarantees on the {\em maximal load} in the system.
These approaches essentially allow a scheduler to sample the state of but a few of the hosts, picked u.a.r., and place a request on one of the sampled hosts.
Additional works~\cite{Sparrow, Tarcil} proposed practical implementations of this approach in cloud network environments.
However, all these works address a problem that is inherently different from the one studied in our work, which is prevalent in NFV environments. While these works assume an (infinite) buffer for pending requests in each host,  requests for VNF deployment can be either accepted or declined and the goal is to have the VNF operational as soon as possible. Thus, queuing these for later deployment makes little sense. 

In addition,~\cite{P2_azar, parallel_P2,P2,Sparrow, Tarcil} select the ``best'' host to place the request on, among the sampled hosts, using deterministic criteria. Consequently, multiple schedulers sampling the same (best) host are still likely to conflict, making it very hard to provide guarantees on the decline ratio. In contrast, our schedulers select a host u.a.r. among all the available hosts it finds. Such randomness allows for a provable low decline ratio.
It should be noted that for the purpose of using randomness to improve performance, the OpenStack community introduced the parameter \texttt{scheduler\_host\_subset\_size}~\cite{host_subset_size} (denoted $\ell$), which works as follows. 
After ranking the available hosts in the weight stage, the scheduler randomly assigns the request to one of the $\ell$ top ranking hosts. 
However, as the impact of this parameter on the performance and the decline ratio hasn't been analyzed, its value is currently determined using crude estimations and rules of thumb. 
Our work provides insight as to how one should optimize the choice of $\ell$. Specifically, we show that the fully random scheme,
where $\ell$ equals the  number of hosts $\nbins$, minimizes the decline ratio. In Section~\ref{sec:apsr:PlacementStudy}, we show that the common approach of setting $\ell$ to be some small constant still results in poor performance. 

\subsection{System Model}
\label{sec:apsr:model}

We now describe our system model. For ease of reference, the notation used in this chapter is summarized in Table~\ref{tbl:apsr:model}.
We consider a set $\hosts$ of $\nbins$ hosts where each host has some multi-dimensional capacity corresponding to several types of resources, e.g., memory, CPU, or disk space. Formally, we model each $\host \in \hosts$ as a vector whose coordinates correspond to the currently available amount of resources of each type. We refer to this vector as the {\em state} of the host.
We further consider a set $\requests$ of \emph{requests}, each modeled as a vector of demand for each resource. 
% Intuitively, requests can be both VMs and containers.
We assume each request $\request \in \requests$ has its vector drawn from some finite set of {\em flavors}, $\flavors=\set{\flavor_1,\ldots,\flavor_m}$.
A host $\host$ is considered \emph{available} for request $\request$ if it has enough resources of each type, i.e., if $\host > \request$, coordinate-wise.

We assume time is slotted, such that in every time slot some requests arrive at the system, and are queued, pending assignment to hosts.
We denote by $\nsched$ the number of \emph{parallel schedulers} that may perform scheduling decisions simultaneously in any single time slot.
In each time slot $t$, given a queue consisting of some $q$ requests pending at $t$, each of the first $\nsched$ requests in  
the queue is matched to a distinct scheduler, which should proceed in assigning its matched request to one of the hosts.
Each such scheduler may query some subset of hosts for their state, after which it assigns its pending request to one of the available hosts, out of the set of hosts it has queried.
We note that when $\nsched>1$, multiple schedulers may concurrently assign their pending requests to the same host.

Any host $\host \in \hosts$ resolves concurrent requests being assigned to $\host$ at the same time slot in some arbitrary order. The resolution of request $\request$ being assigned by some scheduler to host $\host$ \emph{fails} if the host is no longer available when it resolves $\request$, and is \emph{successful} otherwise. The host updates its available capacity upon a successful resolution by setting $\host = \host - \request$.  Requests live for some time, and the host regains the resources used by completed requests. If request $\request$ placed on host $\host$ is completed we update the resource state of the host by setting $\host = \host + \request$.
The above model implies that a request fails if either \begin{inparaenum}[(i)]
\item the scheduler does not find an available host, or
\item the chosen host is no longer available once it resolves the request.
\end{inparaenum}

In every time slot $t$, and for every request flavor $\flavor \in \flavors$, we let $\nfreebins^{(t)}_{\flavor}$ denote the number of hosts in $\hosts$ that are available for a request of flavor $\flavor$ at time $t$. We further let $\nfreebins^{(t)}$ denote an {\em estimate} of the number of hosts that may accommodate {\em any} request that may arrive at time $t$. We note that $\nfreebins^{(t)}$ may be a pessimistic estimate (e.g., by setting $\nfreebins^{(t)}=\min_{\flavor} \ \nfreebins^{(t)}_{\flavor}$), or it may incorporate some information about the workload distribution, or otherwise the system state.
We will usually be omitting the superscript of $(t)$, and refer to $\nfreebins_{\flavor}$, and $\nfreebins$, when the time slot in question is clear from the context.

The \emph{decline ratio} is the ratio between the number of failed requests and the total number of requests handled by the system. We will use $\delta$ to denote the expected decline ratio of the system (for some set of requests $\requests$). Since we are handling requests independently, $\delta$ denotes the probability of having a declined request.
We assume the system is subject to a \emph{Service Level Agreement (SLA)} which limits the decline ratio to be at most $\dr$, for some $\dr \in [0,1]$.
We further assume we are given some budget $B$ such that the maximal number of queried hosts in every time slot is at most $B$.
In every time slot $t$, we denote by $\dqueries$ the number of hosts queried by any scheduler with a pending request at $t$.
A \emph{valid configuration} of schedulers determines $\nsched$ and $\dqueries$, such that $\nsched \cdot \dqueries \leq \budget$, and the probability of a failed request is at most $\dr$.

Our goal is to find a valid configuration that maximizes the number of parallel schedulers ($\nsched$). We refer to this problem as the {\em Constrained Maximum Parallelism (CMP)} Problem.

\newpage
\LTcapwidth=0.95\textwidth
\begin{longtable}[t!]{| p{.15\textwidth} | p{.68\textwidth} | p{.06\textwidth} |}%{|l|l|l|}
\caption[List of symbols used in Chapter~\ref{sec:APSR}]{List of symbols used in this chapter. The rightmost colon details the section where the symbol is used. A blank entry means that the symbol is used throughout the chapter.}
\label{tbl:apsr:model} \\

\hline \multicolumn{1}{| p{.15\textwidth} |}{\textbf{Symbol}} & \multicolumn{1}{| p{.68\textwidth} |}{\textbf{Meaning}} & \multicolumn{1}{| p{.06\textwidth} |}{\textbf{Section}} \\ \hline 
\endfirsthead

\multicolumn{3}{c}%
{{\bfseries \tablename\ \thetable{} -- continued from previous page}} \\
\hline \multicolumn{1}{|c|}{\textbf{Symbol}} & \multicolumn{1}{c|}{\textbf{Meaning}} & \multicolumn{1}{c|}{\textbf{Section}} \\ \hline 
\endhead

\hline \multicolumn{3}{|r|}{{Continued on next page}} \\ \hline
\endfoot

\hline \hline
\endlastfoot
        \hline
    	\hline
		$\hosts$ & Set of hosts&\tabularnewline
		\hline
		$\nbins$ & Number of hosts (bins) &\tabularnewline
		\hline
		$\host$ & Host in $\hosts$ (resources availability vector)&\tabularnewline
		\hline
		$\requests$ & Set of requests &\tabularnewline
		\hline
		$\request$ & Request in $\requests$ (resources demand vector)&\tabularnewline
		\hline
		$\flavors$ & Set of requests flavors&\tabularnewline
		\hline
		$\flavor$ & Flavor in $\flavors$ of a request &\tabularnewline
		\hline
        $\nsched$ & Number of schedulers (agents)&\tabularnewline
		\hline
        $\delta$ & Actual decline ratio (ratio of failed requests)&\tabularnewline
		\hline
        $\dr$ & Maximum allowed decline ratio by the SLA&\tabularnewline
		\hline
		$\budget$ & Budget for overall number of queries&\tabularnewline 
		\hline
		$\dqueries$ &Number of hosts queried by each scheduler &\tabularnewline
		\hline
        $\nfreebins_{\flavor}$ & Number of available hosts for flavor $\flavor$&\tabularnewline
		\hline
        $\nfreebins$ & Number of available hosts for any request &\tabularnewline
		\hline
        $\npothappy$ & Number of potentially-happy agents&\ref{sec:apsr:APSR}\tabularnewline
		\hline
        $\nhappy$ & Number of happy agents&\ref{sec:apsr:APSR}\tabularnewline
		\hline
        $\PrFF$ & See Eq.~\ref{Eq:def_sigma} 
        &\ref{sec:apsr:APSR}\tabularnewline
		\hline
		$Bin(a,b,c)$ & See Eq.~\ref{Eq:def_bin} 
        &\ref{sec:apsr:APSR}\tabularnewline
		\hline
		$Q_i$ & Set of bins which agent $i$ finds available &\ref{sec:apsr:APSR} \tabularnewline
		\hline
		$q_i$ & Number of bins which agent $i$ finds available: $q_i = \abs{Q_i}$ &\ref{sec:apsr:APSR} \tabularnewline
		\hline
		$\lambda_a$ & Poisson arrival rate &\ref{sec:apsr:Evaluation}\tabularnewline
		\hline
		$\lambda_d$ & Poisson departure rate&\ref{sec:apsr:Evaluation}\tabularnewline
		\hline
\end{longtable}

\subsection{Our Contribution}\label{sec:apsr:our_cont}
We study the problem of virtual machine placement in virtual network functions. We do so by exploring the interplay between throughput, decline ratio, and communication overhead. We focus our attention on large clouds, where maintaining an always-fresh full system’s state in impractical, resulting in a highly uncertain environment.

In Section~\ref{sec:apsr:PlacementStudy} we study the impact of parallelism on the decline ratio of various popular placement algorithms. We show that parallelism may drastically increase the decline ratio, where we attribute this increase to the determinism of most algorithms. Interestingly, we find that randomly placing VMs in suitable hosts allows for a large degree of parallelism without a significant impact on the decline ratio. Our study further shows that the desired level of parallelism depends on the system's utilization. In general, low-utilization environments allow for more schedulers than high-utilization ones. 

In Section~\ref{sec:apsr:APSR} we take advantage of these observations to introduce our proposed algorithm, \algtop, which dynamically adjusts the number of parallel schedulers according to the system's utilization, and incorporates randomness into its decision making. \algtop\ guarantees that the expected decline ratio is always within
a predefined requirement. Furthermore, \algtop\ is inherently optimized to query but a small number of hosts, thus reducing the communication overheads. In Section~\ref{sec:comb_analysis} we formally analyze the performance of \algtop\ and provide guarantees as to its communication overhead, and expected decline ratio in satisfying requests. In Section~\ref{sec:apsr:practical_implementation} we describe a practical implementation of \algtop.

In Section~\ref{sec:apsr:Evaluation} we evaluate the performance of \algtop\ for three real-life datasets and show that it enables a high degree of parallelism (e.g., effectively running 20-100 schedulers) in a variety of realistic scenarios. We further show that \algtop\ reduces the communication overhead by over 85\% compared to state of the art algorithms. 
% Finally, we integrate and implement \algtop\ within the OpenStack framework and show that it matches the throughput of the fastest OpenStack configuration, while significantly reducing the decline ratio and the communication overhead.  

Finally, we conclude in Section~\ref{sec:apsr:conclusions} with a discussion of our results, and several
interesting open questions.

\section{The Impact of Parallelism on Existing Placement Algorithms}
\label{sec:apsr:PlacementStudy}

We begin by evaluating the effect of parallel schedulers on the decline ratio of existing placement algorithms.

%We select algorithms that attempt to optimize the following capacity problem: given placement requests that are modeled as multidimensional demand vectors (CPU, memory, bandwidth, etc.), and given physical hosts with residual capacity in each dimension, successfully assign as many VMs to hosts. 
\subsection{Evaluated Algorithms}
We briefly introduce some common placement algorithms and provide a brief description of their operation (see, e.g., ~\cite{Mills11comparingvmplacement} for further details).

The {\em WorstFit (WF)} algorithm, which serves as OpenStack's default placement algorithm~\cite{filterscheduler}. This algorithm places requests on one of the least loaded hosts, in order to maximize the remaining amount of resources on the hosts.
For the multidimensional settings, we implement a pessimistic variant of WF where we consider the load of a host to be the maximum load over all the possible resources.

% The \emph{Balanced Utilization Placement (BUP)}~\cite{bup},
% which  serves as OpenStack's default placement algorithm, places new requests on
% one of the least loaded hosts, in order to maximize the remaining amount of resources.  It is a multidimensional interpretation of the Worst-Fit strategy~\cite{Epstein2003}.

The \emph{FirstFit (FF)}~\cite{Epstein2003}
algorithm, which assigns a request to the first host that has sufficient resources to accommodate the request (assuming some arbitrary fixed ordering of the hosts). This approach aims at minimizing the number of utilized hosts, thus reducing energy consumption.

The \emph{Adaptive} algorithm~\cite{RazPlacement} combines WF and FF as follows: It begins by distributing the load evenly on all hosts (like WF); once the load passes a threshold, the algorithm switches to an FF regime.
Throughout our evaluation, we used 0.6 as the threshold for the Adaptive algorithm.

The algorithm \emph{DistFromDiag}~\cite{RazPlacement} attempts to balance the resource consumption in the host according to its proportions.
For example, if a host has 100GB disk and 10GB RAM, it aspires for a 10:1 ratio between available disk and RAM. 

We also consider two algorithms that incorporate randomization into WF and FF. These variants, referred to as \emph{WorstFit-Rand (WFR)} and \emph{FirstFit-Rand (FFR)}, respectively, weigh the hosts based on the WF and FF strategies, but instead of selecting the top-ranking host, they randomly select a host from the
$\ell$ top-ranking available hosts (in the spirit of the option available in OpenStack, as described in Section~\ref{sec:apsr:rltd_work}).
In our evaluation of WFR and FFR we set $\ell=5$. 

Finally, we evaluate the \emph{Random} algorithm, which selects an available host uniformly at random among the available hosts (i.e., hosts that pass the filter stage and have sufficient resources to accommodate the request).

\subsection{Datasets}\label{Sec:apsr:Datasets}
We use three datasets that capture requests made in real systems. We evaluate each workload in a cloud environment that has sufficiently many hosts so as to accommodate all the requests (see Section~\ref{sec:apsr:PlacementStudy:experiments} for details of how to choose the number of hosts).

\begin{table}[t]
	\centering
    \footnotesize
    \begin{tabular}{|cc||ccccc|c|}
		\hline
		& & \multicolumn{6}{c|}{$storage$} \\
		& & 0.01 & 0.04 & 0.1 & 0.3 & 0.54 & Total \\ 
		\hline
		\hline
		\parbox[t]{2mm}{\multirow{6}{*}{\rotatebox[origin=c]{90}{$memory$}}}
		& 0.001 & 14 & 22 & 14 & 3 & 13 & 66 \\
		& 0.016 &    7 &     93 &     0    & 2 &    0     & 102\\
		
		& 0.032 &    83 &    165 &    0 &    14    & 0 & 262 \\
		& 0.064 &    1    & 1 &     1 &    0    & 0 &    3 \\
		& 0.19    & 0&    2    &0 &    0&    2&    4\\
		\cline{2-8}
		& Total    & 105    & 283 &    15 &    19 &    15    & 437 \\
		
		\hline
	\end{tabular}
    \caption[Normalized breakdown of request flavors in the NFV dataset]{Normalized breakdown of requests for VM images by memory and storage, obtained from the NFV dataset.}
    \vspace{\vspacebelowcaption} % increase space after table to ensure it doesn't "stick" to following text
    \label{table:proprietary}
\end{table}

\paragraph*{The NFV Dataset} was collected from a proprietary large NFV management and orchestration (MANO) system~\cite{ASC}\footnote{Although it is not publicly available, the authors of~\cite{ASC} cordially agreed to provide us with the dataset.}. 
In this scenario all hosts are identical and the placement requests are for VMs of preset sizes (flavors). Hosts and placement requests are two dimensional tuples of the form $\left<memory, storage\right>$. The host size and requests are normalized such that each host has a capacity of $\left<1,1\right>$ and each VM requires a certain fraction of this capacity.
Table~\ref{table:proprietary} shows the distribution of flavors for this dataset. 

\begin{table}
	\centering
        \begin{tabular}{|cc||ccc|c|}
		\hline
		& & \multicolumn{4}{c|}{$CPU$} \\
		& & 0.25 & 0.5 & 1.0 & Total \\ 
		\hline
		\hline
		\parbox[t]{2mm}{\multirow{6}{*}{\rotatebox[origin=c]{90}{$memory$}}}
		& 0.125 & 0 & 60 & 0 & 60 \\
		& 0.25 & 123 & 3,835 & 0 & 3,958 \\
		& 0.5 & 0 & 6,672 & 3 & 6,675 \\
		& 0.75 & 0 & 992 & 0 & 992 \\
		& 1.0 & 0 & 4 & 788 & 792 \\
		\cline{2-6}
		& Total & 123 & 11,563 & 791 & 12,477 \\
		\hline
	\end{tabular}
    \caption[Normalized breakdown of request flavors in the Google dataset]{Breakdown of the number of placement request sizes by CPU and memory, obtained from the Google dataset.}
    \label{table:google}
    \vspace{\vspacebelowcaption}
\end{table}

\paragraph*{The Google Dataset} was recorded in one of Google's clusters~\cite{clusterdata:Reiss2011}. It holds data from 12,477 virtual machines that are characterized by tuples of $\left<CPU, memory\right>$. 
As summarized by~\cite{liu}, the normalized CPU values vary
between $0.25$, $0.5$, and $1$, while the memory values can be grouped around
five levels:  $0.125$, $0.25$, $0.5$, $0.75$, and $1$. The hosts capacities are either  $\left<1,2\right>$ or $\left<2,1\right>$ in equal proportions~\cite{RazPlacement}.
Table~\ref{table:google} provides the breakdown of flavors for this dataset. 

\begin{table*}[t]
	\centering
	\small
    \resizebox{\textwidth}{!}{
	
	\begin{tabular}{|c||c|c|c|c|c|c|c|c|c||c|c|c|c|c|c|c|c|c|}
	    \cline{2-16}
	    \multicolumn{1}{c||}{} &
	    \multicolumn{9}{c||}{Small} &
	    \multicolumn{6}{c|}{Large} \\
		\hline
		$CPU$ &0.035&0.07&0.083&0.1& 0.142 & 0.167 & 0.2 & 0.333 & 0.354 &    
		0.4 & 0.5 & 0.5 & 0.8 & 0.833 & 1
		\\ 
		$memory$ &0.008&0.016&0.031&0.008&0.031&    0.063      & 0.016 & 0.125 & 0.062 &    
		0.031   &  0.125 &    0.5 &   0.063  &    0.25 &     0.25  \\
		\hline
	\end{tabular}
	}
	\vspace{0.2cm} % increase space after table to ensure it doesn't "stick" to following text
    \caption[Normalized breakdown of request flavors in the Amazon EC2 dataset]{Breakdown of placement request flavors of $\left<CPU,memory\right>$ obtained from the Amazon EC2 dataset. Flavors are sorted by $CPU$.}
    \label{table:amazon}
    \vspace{\vspacebelowcaption}
\end{table*}

\paragraph*{The Amazon Dataset} is based on data from Amazon EC2 hosts and VM flavors~\cite{Mills11comparingvmplacement,RazPlacement}. Table~\ref{table:amazon} depicts the
flavors of the normalized $\left<CPU, memory\right>$ in this dataset, 
where each column represents one possible flavor of requests. We partition requests flavors into two types: {\em small} flavors, which have a CPU requirement below $0.4$, and {\em large} flavors, which consist of all remaining flavors.
We generate a sequence of $1000$ small requests and $100$ large
ones (i.e., a total of $1100$ requests), where for each request of one of these types we pick its specific flavor uniformly at random from the various flavors of the type.
In evaluating the performance for this dataset we consider hosts with capacities of either $\left<1, 2\right>$ or $\left<2,1\right>$ in equal proportions (similarly to the host setup used in the Google dataset).

\begin{figure}[h!]
    \centering
    \includegraphics[width=\columnwidth] {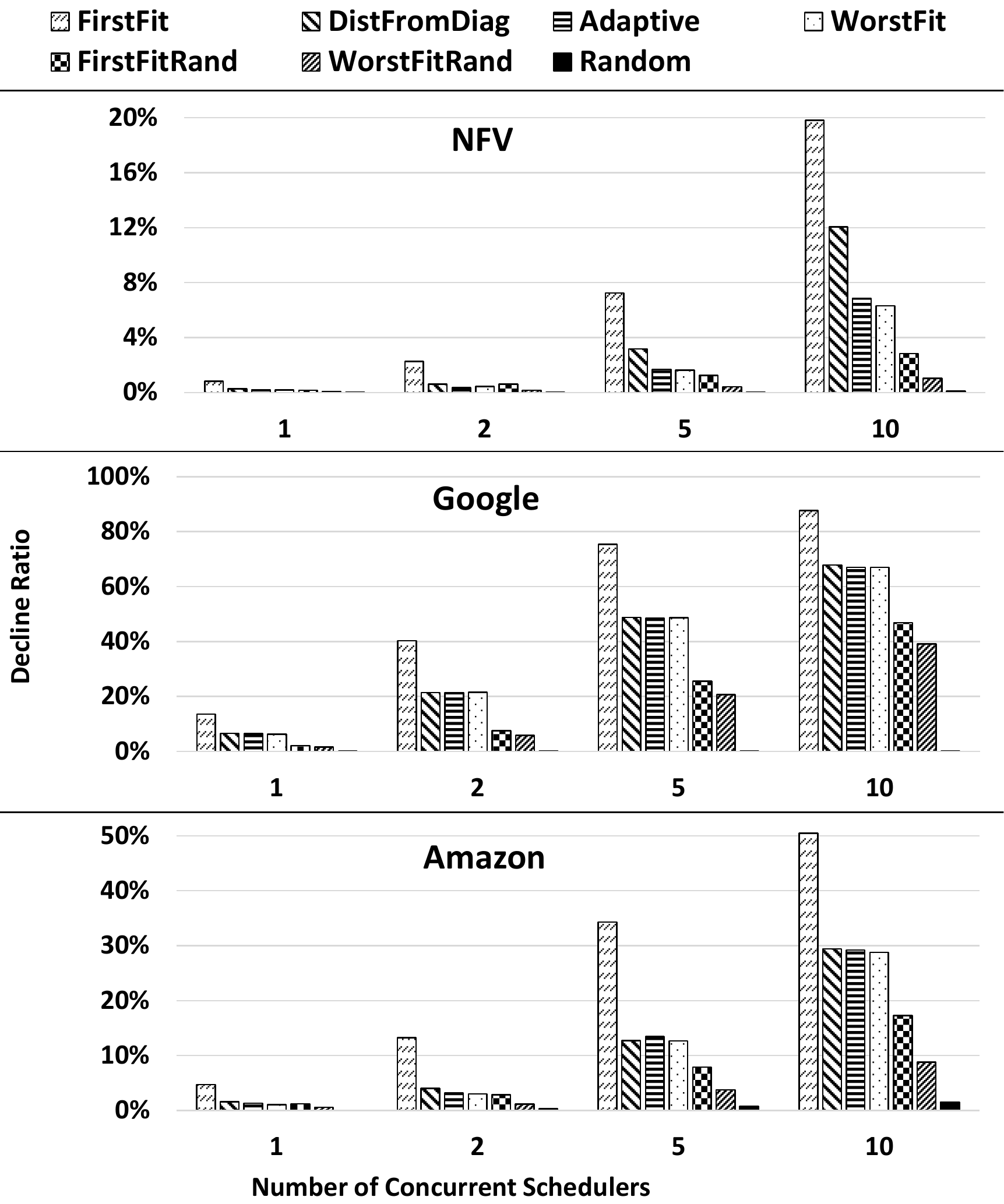}
    \caption[Decline ratios for different placement algorithms and varying number of parallel schedulers]{Decline ratios for different placement algorithms and varying number of parallel schedulers on the NFV, Google and Amazon datasets. Note that the decline ratio (y-axis) ranges corresponding to the various datasets are distinct.} 
    \vspace{\vspacebelowcaption}
    \label{fig:sim_random_rules}
\end{figure}

\subsection{Experiments}\label{sec:apsr:PlacementStudy:experiments}
Our goal in this section is to understand the effect of running multiple parallel schedulers with existing algorithms.

The number of hosts is selected so that it is possible to place all requests at once (by some algorithm). 
Since evaluating the required number of hosts to accommodate all the requests in a given trace is equivalent to the multi-dimensional bin packing problem which is NP-hard~\cite{garey79computers}, we use the approximation suggested in~\cite{RazPlacement}. Briefly, this approximation runs the trace for each algorithm multiple times, each time with a randomly-generated order of requests. Whenever the placement algorithm does not succeed in accommodating a request with the currently available resources, the approximation opens a new host. The approximated value for the required number of hosts is the minimum obtained over all algorithms and orders of requests.

In order to simulate large clouds, we replicated the NFV dataset to have 4730 requests with 279 hosts. The Amazon dataset is evaluated with 126 hosts, and the Google dataset with 5989 hosts. 

Our results are illustrated in Fig.~\ref{fig:sim_random_rules}. Notice that when using a single scheduler there are very few failures in all the policies.
Yet, the decline ratio in Random remains low also for higher levels of parallelism. This result is intuitive as randomly allocating requests to hosts minimizes the probability of having many schedulers select the same host, concurrently.
In contrast, the FirstFit algorithm is the worst, as all the schedulers select the same host even if it is close to being full. In other algorithms like WorstFit, once a host is near full it is less attractive, and thus the schedulers distribute their placement decisions upon a larger number of hosts.  In general, notice that even when running only $10$ schedulers we experience a noticeable increase in the number of failed requests. 

Also notice that OpenStack's solution of introducing small randomization into traditional algorithms improves the behavior, and yet statically setting it at $\ell=5$ is sufficient when we only have two schedulers but insufficient when we have ten schedulers. Still, these results show that the OpenStack community correctly identified the problems with parallelism and introduced an effective workaround. However, the question of understanding the interplay between parallelism and SLA compliance was left open. 

Our work builds upon the insights drawn from the above results, and makes a claim that one should use as much randomness as possible to maximize the parallelism in resource management.
In particular, our goal is to study the scaling laws of parallelism when combined with random VM placement. 

\section{The Adaptive Partial State Random (APSR) Algorithm}
\label{sec:apsr:APSR}

This section presents an analysis our algorithm \emph{Adaptive Partial State Random (\algtop)}. 

Motivated by our observations from Section~\ref{sec:apsr:PlacementStudy}, \algtop\ implements an efficient random policy that dynamically adjusts the number of schedulers ($s$) according to the system's perceived utilization state as captured by the estimate $\nfreebins$ of available hosts. Whenever \algtop\ uses parallel schedulers ($s>1$) it is guaranteed to satisfy the SLA and budget constraints.  

Each \algtop\ scheduler does the following upon receiving a placement request:
\begin{inparaenum}[(i)]
\item queries $d$ hosts (for some value $d$), 
\item filters out hosts that cannot accommodate the request,
\item randomly selects an available host out of the remaining set of hosts, and
\item sends the request to the chosen host.
\end{inparaenum}

\algtop\ relies on a centralized controller called the \emph{\algtop\ Controller} to periodically do the following:
\begin{inparaenum}[(i)]
\item estimate the system's utilization, captured by the estimate $\nfreebins$ of the number of available hosts,
\item determine the number $\nsched$ of parallel schedulers, and
\item determine the number $d$ of hosts each scheduler queries per request. 
\end{inparaenum}
The controller determines the above parameters to ensure the validity of the configuration.

Algorithm~\ref{alg:algtop} illustrates the \algtop\ controller algorithm. 
The procedure GenerateSchedulers($\nsched, d$) adjusts the number of schedulers to $\nsched$ and sets the number of hosts $d$ queried by each scheduler.
The method EstimateK provides an estimate of the number of hosts $\nfreebins$ that can accommodate a request. We note that we do not specify the arguments for this method, since it can be implemented in a variety of ways (see details for several such approaches in~Section~\ref{sec:apsr:practical_implementation}).
Finally, the heart of the controller lays in the procedure \alg\ that takes into account the system state and the SLA constraints, and calculates the number of schedulers $\nsched$ and the number of hosts $d$ to be queried by each scheduler.  

\begin{algorithm}[tbh!]
\caption {\algtop \ Controller ($\nbins, \dr, \budget, T$)} \label{alg:algtop}
\begin{algorithmic}[1]
\State $\nsched \gets 1$, $k \gets n$ \label{alg:algtop:nsched_gets_1}
\State GenerateSchedulers (1, \budget)\label{alg:algtop:Generate_1_sched}
\For{every time slot $t=T,2T,3T,\ldots$}\label{alg:algtop:While_loop_start}
    \State \label{alg:APSR:estimateK} $\nfreebins \gets$ EstimateK($\ldots$)
    \State ($\nsched, \dqueries$) $\gets$ \alg($\nbins, \dr, \budget, \nfreebins$)
    \label{alg:algtop:call_max_paral}
        \State GenerateSchedulers ($\nsched, \dqueries$)
\EndFor
\label{alg:algtop:While_loop_end}
\end{algorithmic}
\end{algorithm}

In what follows we  shed light on the interplay between the number of schedulers and the decline ratio and provide a detailed description of the \alg\ method.
We begin our analysis with a simplified {\em balls-and-bins} model, where we view hosts as {\em bins}, and requests as {\em balls}. In this simplified model each bin has the capacity to store a single ball, and all balls are of the same size.
Each scheduler is viewed as an {\em agent} which may assign balls to bins.
We show sufficient conditions for guaranteeing the SLA requirement in this simplified model, where our sufficient conditions provide a lower bound on the number of agents that may be employed to perform assignments in parallel. We further show that the decline ratio in this simplified model serves as an upper bound on the decline ratio in the original model. These results imply that when \algtop\ utilizes parallelism the decline ratio is at most $\dr$, and the total number of queries performed by all agents is at most $B$. 

\section{Analysis}\label{sec:comb_analysis}

In this section, our goal is to shed light on the interplay between the number of schedulers and the decline ratio.
While studying this interplay, we provide a detailed description of the \alg\ method, and several variants of the EstimateK method, that are used in Algorithm~\ref{alg:algtop}.

We begin our analysis with a simplified {\em balls-and-bins} model, where we view hosts as {\em bins}, and requests as {\em balls}. In this simplified model each bin has the capacity to store a single ball, and all balls are of the same size.
Each scheduler is viewed as an {\em agent} which may assign balls to bins.
We show sufficient conditions for guaranteeing the SLA requirement in this simplified model, where our sufficient conditions provide a lower bound on the number of agents that may be employed to perform assignments in parallel. We further show that the decline ratio in this simplified model serves as an upper bound on the decline ratio in the original model. We do this while ensuring that the total number of queries performed by all agents is at most $B$. These results imply that \algtop\ always maintains a valid configuration.

\subsection{Balls-and-bins Model}\label{SubSec:balls_n_bins_model}
We employ $\nsched$ identical agents which try to place balls in available bins, and they all act in parallel.
Each agent makes $d$ queries of  random bins (with replacements) for their state and possibly finds some of them available.
If the agent does not find any available bins, the ball assignment fails. Otherwise, the agent selects an available bin uniformly at random and places its ball in the selected bin.

Agents are unaware of the decisions made by other agents which may cause multiple agents to select the same available bin. 
In such a case, one of the agents succeeds, and the rest of them fail.
We use the term {\em potentially-happy agent} to refer to an agent that finds an available bin.
Similarly, the term {\em happy agent} refers to an agent that successfully places a ball in an available bin. Finally, we use the term {\em unhappy agent} to refer to an agent that fails to place its ball. 
An agent may become unhappy if she either (i) does not find an available bin, or (ii) finds at least one available bin, but places her ball in a bin which is already occupied by another agent. In the latter case, we say that the assignment failed due to \emph{collision}.

We use the random variables  $\npothappy$ and $\nhappy$ to denote the number of potentially-happy agents and happy agents, respectively.
We further use $\nfreebins$ to denote a lower bound on the number of available bins in some time slot where agents contend for assigning balls into bins.

We view the SLA requirement of having a decline ratio of at most $\dr$ as a lower bound on the probability that an arbitrary agent attempting to assign a ball to some bin is happy. Formally, this requirement translates to ensuring that:  
\begin{equation}\label{Eq:SLA_1}
\frac{E[\nhappy]}{\nsched} \geq 1 - \dr.
\end{equation}
We also require that the total number of bins queried by our agents is no more than a prescribed budget ($B$), which translates to requiring that:
\begin{equation}\label{Eq:SLA_2}
\nsched \cdot d \leq \budget.
\end{equation}

Given $\nbins, \nfreebins, \dr$ and $\budget$, our goal is to find the largest number of agents $\nsched$, and a value of bin queries per agent ($d$) that satisfy both Equation~\ref{Eq:SLA_1}, and Equation~\ref{Eq:SLA_2}. 

To guarantee that Equation~\ref{Eq:SLA_1} is satisfied, we strive to calculate the expected number of happy agents $E[\nhappy]$. 
Observe that  $E[\nhappy]$ can be expressed by conditioning the number of happy agents $\nhappy$ on the number of potentially-happy agents $\npothappy$. I.e., 
\begin{equation}\label{Eq:Ehs_by_Ehs_cond_f}
  E[\nhappy] = 
  \sum_{f=1}^{s} \bigg[\Pr (\npothappy = f) \cdot E[\nhappy | \npothappy = f] \bigg].
\end{equation}

We now turn to evaluate the probability distribution of 
$\npothappy$, and then calculate the conditional expectation $E[\nhappy | \npothappy = f]$. 

To evaluate the distribution of the number of potentially-happy agents  $\npothappy$, observe that an agent fails to find an available bin with probability
$\left (\frac{\nbins-\nfreebins}{\nbins} \right)^d$. Therefore the probability that an agent is potentially-happy is: 
\begin{equation}\label{Eq:def_sigma}
\PrFF = 1 - \left (\frac{\nbins-\nfreebins}{\nbins} \right)^d.     
\end{equation}

One can interpret 
$\npothappy$ as the result of $\nsched$ independent Bernoulli trials with success probability $\PrFF$. Therefore:
\begin{equation}\label{Eq:Fs}
Pr (\npothappy = f) = Bin (f, \nsched, \PrFF),
\end{equation}

where 
\begin{equation}\label{Eq:def_bin}
Bin (f, \nsched, \PrFF) \equiv 
\binom {\nsched}{f}  \PrFF^f (1-\PrFF)^{\nsched-f}.
\end{equation}

For calculating $E[\nhappy | \npothappy = f]$, we examine the process of the potentially-happy agents placing their balls from the point of view of the $\nfreebins$ free bins. For ease of presentation, we associate 
each potentially-happy agents with a sequence number $1, \dots, f$, 
and each available bin with a sequence number $1, \dots, \nfreebins$. 

The following proposition shows that the probability that an arbitrary potentially-happy agent selects an arbitrary available bin is uniform for all potentially-happy agents $1, \dots, f$ and for all available bins $1, \dots, \nfreebins$. 

\begin{proposition}\label{Prop:one_over_k}
If agent $i$ is potentially happy, then the probability that agent $i$ places its ball on available bin $j$ is $\frac{1}{k}$.
\end{proposition}

\begin{proof}
Let $i$ be an agent -- not necessarily a potentially-happy agent.
Denote by $Q_i$ the set of bins which agent $i$ queries and finds available. 
Let $q_i$ denote the random variable for the number of bins which agent $i$ finds available, namely, $\abs{Q_i} = q_i$. Then 
$\Pr (q_i = x)$ captures the probability that agent $i$ finds $x$ distinct available bins in his overall $d$ samples. 

For each available bin $\ell$, we let  $B_\ell$ denote a binary random variable which indicates whether agent $i$ samples bin $\ell$. Namely, $B_\ell=1$ iff $\ell \in Q_i$. Then we have
\begin{equation}\label{Eq:indicator_rv}
\begin{split}
\sum_{\ell=1}^k \Pr \left(\ell \in Q_i | q_i = x \right) & = 
\sum_{\ell=1}^k \Pr(B_\ell = 1 | q_i = x) \\
& = 
\sum_{\ell=1}^k E\left[B_\ell | q_i = x\right] \\
& = 
E \left[ \sum_{\ell=1}^k B_\ell | q_i = x \right] \\
& = 
E [q_i | q_i=x] = x.
\end{split}
\end{equation}

Since agent $i$ samples the bins i.i.d., we have for each $\ell, \ell'$ that $\Pr (\ell \in Q_i | q_i = x) = \Pr (\ell' \in Q_i | q_i = x)$. 
By Equation~\ref{Eq:indicator_rv} it follows that 
for every $\ell = 1, \dots, \nfreebins$, 
$\Pr (\ell \in Q_i | q_i = x) = \frac{x}{\nfreebins}$.

Hence, the probability that agent $i$ selects available bin $j$ is
\begin{equation}\label{Eq:i_sel_j}
\begin{split}
\Pr (j \in Q_i) 
& = 
\sum_{x=1}^{\nfreebins} \Pr (q_i = x) \cdot
\Pr (j \in Q_i | q_i = x) \\
& =
\frac{1}{\nfreebins}
\sum_{x=1}^{\nfreebins} x \cdot \Pr (q_i = x) 
\end{split}
\end{equation}

If agent $i$ samples available bin $j$, then she selects $j$ w.p. $\frac{1}{x}$. It follows that
\begin{equation}\label{Eq:i_smpls_j}
\begin{split}
\Pr (i \ \textrm{selects} \ j) 
& =
\frac{1}{\nfreebins}
\sum_{x=1}^{\nfreebins} x \cdot \Pr (q_i = x) \cdot 
\frac{1}{x} = 
\frac{1}{\nfreebins}
\sum_{x=1}^{\nfreebins} 
\Pr (q_i = x) 
\end{split}
\end{equation}

Observe that agent $i$ is potentially happy iff she samples at least one available bin, that is, if $q_i > 0$. This happens w.p. 
$\sum_{x=1}^{\nfreebins} \Pr (q_i = x)$. Combining this observation with Equation~\ref{Eq:i_smpls_j}, the probability that agent $i$ samples available bin $j$ given that $i$ is potentially-happy is $\frac{1}{k}$.
\end{proof}

By Proposition~\ref{Prop:one_over_k}, the probability that agent $i$ does \emph{not} place its ball in bin $j$ is $1 - \frac{1}{\nfreebins} = \frac{k-1}{k}$. As the agents are mutually independent, the probability that none of the $f$ potentially-happy agents places its ball in bin $j$ is $\left( \frac{k-1}{k} \right)^f$. The probability that at least one of the $f$ potentially-happy agent tries to place its ball in bin $j$ is $1 - \left( \frac{k-1}{k} \right)^f$. 
From the point of view of bin $j$, this process is equivalent to a Bernoulli trial, which succeeds iff at least one agent places its ball in bin $j$. If the Bernoulli trial succeeds, bin $j$ is exclusively associated with a single happy agent.

Applying the analysis above for each of the $\nfreebins$ free bins, we obtain that $E[\nhappy | \npothappy = f]$ is equivalent to the expected number of successes in $\nfreebins$ independent Bernoulli trials, with probability of success $1 - \left( \frac{k-1}{k} \right)^f$ each. Hence, 

\begin{equation}\label{Eq:EHs_cond_f_correct}
E[\nhappy | \npothappy = f] = \nfreebins \left[
1 - \left( \frac{k-1}{k} \right)^f
\right]
\end{equation}

Combining Equation~\ref{Eq:Ehs_by_Ehs_cond_f} with Equations~\ref{Eq:Fs} and~\ref{Eq:EHs_cond_f_correct}, we obtain the following  expression for 
$E[\nhappy]$
\begin{equation}\label{Eq:EHs}
E[\nhappy]  = 
\nfreebins 
\sum_{f=1}^{\nsched} 
\left[
1 - \left( \frac{k-1}{k} \right)^f
\right]
\cdot
Bin (f, \nsched, \PrFF)
\end{equation}

Combining the SLA requirement (Equation~\ref{Eq:SLA_1}) and Equation~\ref{Eq:EHs}, we show how to guarantee that the decline ratio is at most $\dr$ in the following corollary.
\begin{corollary}\label{cor:balls_n_bins_SLA}
If 
\begin{equation}
\label{eq:cor:balls_n_bins_SLA}
\nfreebins 
\sum_{f=1}^{\nsched} 
\left[
1 - \left( \frac{k-1}{k} \right)^f
\right]
\cdot
Bin (f, \nsched, \PrFF)
\geq \nsched (1 - \dr)
\end{equation}
then the expected decline ratio with
$\nsched$ agents, where each agent queries $\dqueries$ bins, is at most $\dr$.
\end{corollary}

Based on Corollary~\ref{cor:balls_n_bins_SLA}, we now describe the details of the \alg\ method, which maximizes the parallelism while satisfying the SLA and budget constraints. The method is detailed in Algorithm~\ref{alg:max_paral}.
After initializing the number of schedulers to $\nsched=1$, the algorithm repeatedly increases the value of $\nsched$ (and adjusts the number of queries performed by each scheduler to satisfy the budget constraint), as long as Equation~\ref{eq:cor:balls_n_bins_SLA} is satisfied.
Specifically, procedure SatisfySLA validates that Equation~\ref{eq:cor:balls_n_bins_SLA} is satisfied for the given configuration.

\begin{algorithm}[th]
\caption {\alg \ ($\nbins, \dr, \budget, \nfreebins$)} \label{alg:max_paral}
\begin{algorithmic}[1]
\State $\nsched \gets 1$\label{alg:max_paral:line_s_eq_1}\Comment{initialization} 
\While {SatisfySLA $\left( \nbins, \dr, \nfreebins, \nsched+1, \bigl\lfloor{\frac{\budget}{\nsched+1}}  \bigr\rfloor \right)$}\label{alg:max_paral:while_s_begin}
    \State $\nsched \gets \nsched + 1$
\EndWhile \label{alg:max_paral:while_s_end}
\State return $\nsched, \big\lfloor{\frac{\budget}{\nsched}} \big\rfloor$ \label{alg:max_paral:while_d_end}
\end{algorithmic}
\end{algorithm}

\subsection{SLA Guarantees with Availability Lower Bounds}

We now show how to apply our results to the original model, assuming that the value $\nfreebins$ is guaranteed to be a {\em lower bound} on the number of hosts that are available for any request. This assumption effectively translates to an upper bound on the system's utilization.
We begin by considering the case where $\nfreebins$ is the precise number of available bins for any request.
The following theorem shows that if we know $\nfreebins$, then \alg\ indeed generates a valid configuration.

\begin{theorem}\label{thm:practic_alg_SLA}
Assume $k$ is the number of available hosts that may accommodate any request flavor.
If \alg($\nbins, \dr, \budget, \nfreebins) = (\nsched, d)$ and $\nsched>1$ then employing $\nsched$ schedulers, each querying $\dqueries$ hosts,
% during time interval $T$
guarantees an expected decline ratio of at most $\dr$.
\itamar{The condition $\nsched>1$ is because  $\nsched=1$ is the default value. We don't use $\nsched=0$ as default, because this will force the upper-level alg' to check the value returned by \alg, to ensure at least 1 scheduler is generated.}
\end{theorem}
\begin{proof}
Let $\hosts_{\flavor}$ denote the set of hosts with enough resources for accommodating a request of flavor $\flavor$. Using our notation, it follows that $\abs{\hosts_{\flavor}} = \nfreebins_{\flavor}$.
Let $\flavor^* = \arg\min_{\flavor}\set{\nfreebins_{\flavor}}$.

Consider the following {\em condensation} process:
\begin{enumerate}
\itemsep-0.2em 
    \item \label{condense:1} Make all the hosts in $\hosts_{\flavor^*}$ available for all flavors.
    \item \label{condense:2} Make the rest of the hosts unavailable for any request.
    \item \label{condense:3} Determine that once a scheduler allocates a request in a host, this host becomes unavailable for any further requests (in this time slot). 
\end{enumerate}   

We claim that condensing the system can only increase its decline ratio from the following reasons:
First, as for each $\flavor \in \flavors$ we have $\nfreebins_{\flavor^*} \leq \nfreebins_{\flavor}$, steps~\ref{condense:1} and~\ref{condense:2} can only decrease the number of hosts available for each flavor. This decrease reduces the expected number of available hosts found by each scheduler.
Second, steps~\ref{condense:1} and~\ref{condense:2} concentrate the available hosts of all flavors to be exactly $\hosts_{\flavor^*}$. This concentration may only increase the probability that multiple schedulers will end up assigning their requests to the same host.
Finally, a host may accommodate multiple parallel requests providing it has enough resources while step~\ref{condense:3} disallows it which implies a potential increase in the decline ratio.
That is, we showed that an algorithm that satisfies the SLA on a condensed system also satisfies it on the original system.

We now note that the condensed system is equivalent to our balls-and-bins model. To see this, observe that once the sets of available hosts for every request become identical (due to steps~\ref{condense:1} and~\ref{condense:2}), the requests themselves are also virtually identical, and thus become equivalent to the identical balls in our balls-and-bins model. Furthermore, as every host can accommodate only a single request (due to step~\ref{condense:3}), the hosts can be modeled as identical bins, where each available bin can accommodate merely a single ball.

By Corollary~\ref{cor:balls_n_bins_SLA}, 
\alg\ satisfies the SLA requirement in the balls-and-bins model, which is equivalent to guaranteeing SLA also in the condensed system. As the decline ratio in the condensed system serves as an upper bound on the decline ratio ($\dr$) the result follows.
\end{proof}

We stress that the proof of Theorem~\ref{thm:practic_alg_SLA} implicitly suggests that all the requests handled in a time slot belonging to the flavor with the minimum amount of available hosts. Furthermore, it suggests that no two requests can be placed in parallel on the same host. Therefore, we expect better decline ratios in practice. 

The following corollary shows that for providing performance guarantees it is sufficient to know only a lower bound on the number of hosts available for every request flavor.

\begin{corollary}
The consequence of Theorem~\ref{thm:practic_alg_SLA} holds whenever $k$ is a lower bound on the number of hosts available for every request flavor.
\end{corollary}

\begin{proof}
We have to show that increasing 
the number of hosts available for every request flavor, 
while keeping the number schedulers $\nsched$ and the sample size $d$ unchanged, can only decrease the decline ratio. 
We do so by checking the effect of increasing 
the number of available bins $\nfreebins$ 
on  our balls-and-bins analysis.
As we now vary 
$\nfreebins$,
% the number of available hosts, 
we add to the notation of our random variables a superscript indicating its value.
That is, $\phk$ and $\hk$ denote the random variable for the number of potentially-happy and happy agents, respectively,  when  there are $\nfreebins$ available bins. 
Recalling the SLA requirement in Equation~\ref{Eq:SLA_1}, it suffices to show that $E\left[\hkpp \right] \geq E\left[ \hk \right]$.

Using our modified notation, we rewrite Equation~\ref{Eq:Ehs_by_Ehs_cond_f} as
\begin{equation}\label{Eq:Ehs_by_Ehs_cond_f_k}
  E\left[\hk\right] = 
    \sum_{f=1}^{s} \Pr \left(\phk = f \right) \cdot E \left[\hk | \phk = f \right] 
\end{equation}

We now handle separately each of the components in the product appearing in the right hand side of Equation~\ref{Eq:Ehs_by_Ehs_cond_f_k}, namely
\begin{inparaenum}[(i)]
\item the probability distribution of the number potentially-happy agents, and
\item the expected number of happy agents, given that there are $f$ potentially happy agents.
\end{inparaenum} 

Intuitively, the probability of having more than $f$ potentially-happy agents is non-decreasing in the number of free bins $\nfreebins$. Indeed, 
combining  Equations~\ref{Eq:def_sigma},~\ref{Eq:Fs} and~\ref{Eq:def_bin} shows that 
\begin{equation}\label{Eq:Fs_inc_in_k}
\Pr \left(\phkpp > f\right) \geq \Pr \left(\phk > f\right).
\end{equation}

To quantify the impact of the number of potentially-happy agents $f$ on the expected number of happy agents  we
let $\dif(k, f)$ denote the difference function
\begin{equation}\label{Eq:def_d}
\dif (k, f) = E \left[\hk | \phk = f+1\right] - 
E \left[\hk | \phk = f\right].
\end{equation}

$\dif(k, f)$ captures the contribution of adding one potentially-happy agent to the expected number of happy agents. 
As $E [\hk | \phk = 0] = 0$, we have $\dif (k, 0) = E [\hk | \phk = 1]$. We can therefore rewrite Equation~\ref{Eq:Ehs_by_Ehs_cond_f_k} as follows:
\begin{equation}\label{Eq:EHs_by_f_and_dif}
\begin{split}
E \left[\hk \right] 
= & \sum_{f=1}^s \Pr(\phk = f) \cdot E \left[\hk | \phk = f \right] \\
= & \Pr(\phk = 1) \cdot \dif (k, 0) + \\
& \Pr(\phk = 2) \cdot \left[\dif (k, 0) + \dif (k, 1) \right] + \\
& \dots + \\
& \Pr(\phk = s) \cdot \left[\dif (k, 0) + \dif (k, 1) + \dots + \dif (k, s-1) \right]\\
= & \sum_{f=0}^{s-1} \left[\Pr (\phk > f) \cdot \dif (k, f) \right]
\end{split}    
\end{equation}

By combining Equations~\ref{Eq:Fs_inc_in_k} and~\ref{Eq:EHs_by_f_and_dif}, it suffices to show that
\begin{equation}\label{Eq:d_inc_in_k}
\dif(k+1, f) > \dif (k, f)    
\end{equation}

For proving  that Equation~\ref{Eq:d_inc_in_k} is satisfied, we assign Equation~\ref{Eq:EHs_cond_f_correct} in the definition of $\dif(*)$ in Equation~\ref{Eq:def_d}, and obtain:
\begin{equation}\label{Eq:dif_k_correct}
\begin{split}
\dif (k, f) & = 
\nfreebins \left[
\left( \frac{k-1}{k} \right)^{f}
- 
\left( \frac{k-1}{k} \right)^{f+1} 
\right]  \\
& = 
\left( \frac{k-1}{k} \right)^f
\end{split}
\end{equation}

Hence,
\begin{equation}
\begin{split}
\dif (k+1, f) - \dif (k, f) = 
\left( \frac{k}{k+1} \right)^f -
\left( \frac{k-1}{k} \right)^f
\geq 0
\end{split}
\end{equation}
where the last inequality is satisfied for every $\nfreebins > 0$.
\end{proof}

\section{Practical Implementation of \algtop}
\label{sec:apsr:practical_implementation}

We now discuss practical aspects of implementing \algtop.
The main caveat in implementing \algtop\ such that the conditions in Theorem~\ref{thm:practic_alg_SLA} are met revolves around the estimation of the number of available hosts for any request flavor ($k$).

A straightforward option is to compute $\nfreebins$ explicitly, by running a centralized periodic task, which gathers state from all hosts. We note that such a task may be executed by the \algtop\ Controller (in line~\ref{alg:APSR:estimateK}). When the task is performed in every time step (i.e., by setting $T=1$ in \algtop), then the guarantees of Theorem~\ref{thm:practic_alg_SLA} hold.
However, this approach incurs the communication overhead of querying all the hosts.\footnote{Note that all the alternative placement algorithms considered in section~\ref{sec:apsr:PlacementStudy} have {\em each scheduler} query all the hosts in every time slot.}

Therefore we suggest an alternative approach, which does not incur any additional communication overhead. We observe that while the {\em identities} of the available concrete hosts for a
request of some flavor $\flavor$ are dynamic, the total {\em number}
of these hosts is expected to change at a slower rate, and the {\em minimal} number of available hosts (overall request flavors) is
likely to change slower still.
Using this observation, we propose to estimate $\nfreebins$ by relying on the statistics which the schedulers gather during their regular operation. We now describe our proposed algorithm for estimating $\nfreebins$, EstimateK($k$), which is formally defined in Algorithm~\ref{alg:estimateK}.

\begin{algorithm}[ht]
\caption {EstimateK ($k$)} \label{alg:estimateK}
\begin{algorithmic}[1]
    \ForAll{$\flavor \in \flavors$} \Comment{for each flavor}
        \State $\nbins_{\flavor}^{(tot)} \gets \sum_{i=1}^{\nsched} \nbins_{\flavor}^{(i)}$ 
        \label{alg:algtop_n_i_tot}
        \State $\nfreebins_{\flavor}^{(tot)} \gets \sum_{i=1}^{\nsched} \nfreebins_{\flavor}^{(i)}$
        \label{alg:algtop_k_i_tot}
    \EndFor
    \State $\tilde{\nfreebins} \gets
    \nbins \cdot \min_{\flavor \in \flavors}
    \left[ \frac{\nfreebins_{\flavor}^{(tot)}}{\nbins_{\flavor}^{(tot)}} \right]$ \label{alg:algtop:assign_k_n}
    \State return $\alpha \cdot  \tilde{\nfreebins} + (1 - \alpha) \cdot \nfreebins$ \label{alg:algtop:exp_moving_average}
\end{algorithmic}
\end{algorithm}

Our algorithm assumes that each scheduler $i$ maintains counters $\nbins_{\flavor}^{(i)}$ and $\nfreebins_{\flavor}^{(i)}$, which keep track of the overall number of hosts queried, and the total number of available hosts of flavor $\flavor$, respectively,  when handling its requests.
These counters are reset after each call to algorithm EstimateK.
The algorithm uses these counters to obtain an estimation of the {\em overall} number of hosts queried, and the overall number of available hosts for each flavor.
These values are then used to obtain an estimate of the {\em percentage} of hosts that were available for each request flavor.
The normalized minimum of all flavors is chosen as the pessimistic estimate of $k$.
The current estimate is then averaged using exponential averaging with the previous estimate, to produce an updated estimate of $k$, the number of hosts available for all request types.
The intuition underlying this estimation is the fact that schedulers operate {\em independently}, and each scheduler queries a relatively small number of hosts in each time slot.
% The estimation is then combined with the concepts used for proving Theorem~\ref{thm:practic_alg_SLA}.

We note that our proposed algorithm for estimating the value of $k$ does not require any additional querying of hosts. In the following section, we demonstrate the effectiveness of this estimation when incorporated within our \algtop\ Algorithm.

\section{Evaluation of the \algtop\ Algorithm}
\label{sec:apsr:Evaluation}
This section positions \algtop\ with respect to known placement algorithms, and evaluates
the interplay between parallelism, utilization, decline ratio, and throughput.

\subsection{Trace-based Simulation}
We model the arrival of requests using a Poisson process with parameter $\lambda_a$. Unless stated otherwise, we set $\lambda_a$ to 20, and $\dr$ (\algtop's target decline ratio) to 5\%. 
We set \algtop's query budget to be $\budget = \nbins$. That is, the {\em overall} number of queries made by all of our parallel schedulers is the same as the number of hosts queried by {\em a single} OpenStack scheduler. 

We start with infinite lifetime requests as it is a common (though somewhat unrealistic) benchmark for placement algorithms, as it demonstrates the relationship between utilization and parallelism in a clean manner.
We set \algtop's time interval for estimating the state of the cloud to be $T=10$, and also set $\alpha=0.1$ for the EstimateK method. 

We use the workloads described in Section~\ref{Sec:apsr:Datasets}, and simulate large clouds with 30 replicas of the NFV dataset, 7 replicas of the Amazon dataset, and 1 replica of the Google dataset attaining a total of 13110, 7700 and 12477 requests, respectively.
We determine the number of hosts as the number of hosts needed for successfully placing all the requests at once (by some offline algorithm), as described in Section~\ref{sec:apsr:PlacementStudy:experiments}; we use 837 hosts for NFV, 876 hosts for Amazon and 5989 hosts for Google. 
\\~\\
\textbf{Comparing \algtop\ to Known Algorithms.}
We now study the impact of parallelism on the decline ratios in \algtop\ and existing placement algorithms. We artificially cap the maximal number of schedulers used by \algtop\ to $100$, as the request arrival rate in our experiments is unlikely to exceed $100$ in any given time slot.  For the competing algorithms, we  vary the (fixed) number of schedulers.

Table~\ref{table:aspr_quality} summarizes the results. 
The algorithms DistFromDiag and Adaptive are abbreviated to Diag and Adapt.
The average number of schedulers used by \algtop\ is shown below its decline ratios. First notice that  \algtop's decline ratios are always within SLA ($\dr =5\%$), and that \algtop\ uses at least 38 schedulers on average. 
Random and \algtop\ yield the lowest decline ratio. The two policies indeed implement Random placement, but \algtop\ only queries a small subset of the hosts, while each of the schedulers employed by Random queries all of them. This less accurate view of the system state causes \algtop's decline ratio to sometimes be slightly higher than that of Random (although always within the SLA).

\begin{table}[h!]
	\centering
	\renewcommand\arraystretch{1.3}
	%\small
	\begin{tabular}{|c|c|c|c|c|c|c|c|c|c|}
		\hline
	    Dataset & \textit{s} &  APSR & \rand & FirstFit & WorstFit & Diag & Adapt\\ 
		\hline
		\hline
		\multirow{3}{*} {NFV} & 1 &   & 0.3\% & 0.0\% & 0.3\% & 0.7\% & 0.3\%  \\
		& 5 &  0.4\% &  0.4\% & 11.1\% & 4.0\% & 5.3\% & 2.2\%   \\
		& 10 & $\bar{s} = 38$ &  0.5\% & 23.3\% & 7.9\% & 7.7\% & 3.2\%  \\
		 & 20 & & 0.7\% & 35.7\% & 11.8\% & 11.7\% & 12.8\%\\
		& 50 &  & 0.8\% & 39.0\% & 15.6\% & 16.1\% & 16.0\% \\
		\hline
		\multirow{3}{*} {Google} & 1  &   & 2.3\% & 0.4\% & 8.7\% & 2.2\% & 8.7\%  \\
		& 5  & 3.1\% & 2.4\% & 56.2\% & 42.0\% & 42.7\% & 42.0\% \\
		& 10 & $\bar{s} = 92.2$   & 2.4\% & 77.8\% & 64.2\% & 62.8\% & 64.6\% \\
		& 20 & &  2.4\% & 87.8\% & 79.5\% & 74.4\% & 79.5\%  \\
		& 50 & & 2.4\% & 88.9\% & 81.4\% & 75.2\% & 81.6\% \\
        \hline
		\multirow{3}{*} {Amazon} & 1  & & 0.5\% & 0.2\% & 0.4\% & 1.3\% & 0.2\%  \\
		& 5  & 0.8\% & 0.6\% & 18.2\% & 6.5\% & 7.2\% & 6.3\% \\
		& 10 & $\bar{s} = 49.7$   & 1.0\% & 33.6\% & 20.4\% & 15.5\% & 20.6\%  \\
		& 20 &  & 1.2\% & 49.1\% & 61.1\% & 31.3\% & 61.7\% \\
		& 50 & & 1.4\% & 52.8\% & 64.6\% & 36.6\% & 64.7\% \\
		\hline
	\end{tabular}
    \caption[Decline ratios of \algtop\ and other placement algorithms when varying the number of schedulers]{Decline ratios of \algtop\ and other placement algorithms when varying the (fixed) number of schedulers ($\nsched$). \algtop's\  average number of schedulers ($\bar{\nsched}$) is listed below the decline ratio.}
    \vspace{\vspacebelowcaption}
    \label{table:aspr_quality}
\end{table}

\definecolor{LightCyan}{rgb}{0.88,1,1}
\begin{table}[h!]
	\centering
	\renewcommand\arraystretch{1.3}
    %\small
	\begin{tabular}{c|c|c|c||c|c|c|}
	\cline{2-7}
		&\multicolumn{3}{c||}{ \algtop} &\multicolumn{3}{c|}{Random} \\ \cline{2-7}
		&\multicolumn{3}{c||}{Target Decline Ratio ($\dr$)} &\multicolumn{3}{c|}{Number of Schedulers}\\
		 \cline{2-7}
		 
		 & 3\% & 5\% & 10\% & 1 &  10 & 100 \\ 
		  \cline{2-7}
	\rowcolor{LightCyan}	 &\multicolumn{6}{c|}{ NFV}\\
		 \cline{2-7}
		\hline	\multicolumn{1}{|c|}{Number of Queries} & 1553K  & 811K & 578K & \multicolumn{3}{c|}{11000K}\\
		\hline  \multicolumn{1}{|c|}{Avg \# of Schedulers ($\bar{s}$)} & 12  & 38 & 79 & 1 & 10 & 100\\
		\hline  \multicolumn{1}{|c|}{Throughput [req./slot]} & 7.2 & 14 & 19.6 & 1  & 10 & 19.8 \\
		\hline  \multicolumn{1}{|c|}{Decline Ratio ($\delta$)} & 0.4\%  & 0.4\% & 0.6\% & 0.3\% & 0.5\%& 0.8\%\\
		\hline
    	\rowcolor{LightCyan}	 &\multicolumn{6}{c|}{ Google}\\
		\hline  \multicolumn{1}{|c|}{Number of Queries} & 3920K  & 3860K & 3823K & \multicolumn{3}{c|}{74724K}\\
		\hline  \multicolumn{1}{|c|}{Avg \# of Schedulers ($\bar{s}$)} & 84.1  & 92.2 & 98.3 & 1 & 10 & 100\\
		\hline  \multicolumn{1}{|c|}{Throughput [req./slot]} & 19.8 & 19.9 & 19.9 & 1  & 10 & 19.9 \\
		\hline  \multicolumn{1}{|c|}{Decline Ratio ($\delta$)} & 3.0\%  & 3.1\% & 2.9\% & 1.9\% & 2.3\%& 2.4\%\\
		\hline
		\rowcolor{LightCyan} &\multicolumn{6}{c|}{ Amazon}\\
		\hline  \multicolumn{1}{|c|}{Number of Queries} & 469k  & 370k & 354k & \multicolumn{3}{c|}{6745.2k}\\
		\hline  \multicolumn{1}{|c|}{Avg \# of Schedulers ($\bar{s}$)} & 24.8  & 49.7  & 80.9 & 1 & 10 & 100\\
		\hline  \multicolumn{1}{|c|}{Throughput [req./slot]} & 15.3 & 19.3 & 19.9 & 1  & 10 & 19.9 \\
		\hline  \multicolumn{1}{|c|}{Decline Ratio ($\delta$)} & 0.7\%  & 0.8\% & 1.0\% & 0.5\% & 1.0\% & 1.4\%\\
		\hline
	\end{tabular}
    \caption[Total number of queries, throughput and decline ratios of \algtop\ versus Random]{Total number of queries, throughput and actual decline ratios of \algtop\ versus Random.}
    \vspace{\vspacebelowcaption}
    \label{table:aspr_price}
\end{table}

\normalsize
Table~\ref{table:aspr_price} compares the throughput, decline ratios and the total number of queries of \algtop\ and Random.
Note that \algtop\ reduces the total number of queries by at least 85\%. Increasing \algtop's target decline ratio increases its parallelism which in turn increases the throughput. This highlights the tension between decline ratio and parallelism. 
The best attainable throughput is 20 as it is the average arrival rate.  Indeed, \algtop\ and Random with fixed 20 schedulers are very close to the maximal throughput. Also recall that unlike Random, \algtop\ may fail due to not finding an available host in the queried hosts, thus its decline ratio is sometimes higher.
\newcommand {\subfloatwidth}  {0.7*\columnwidth}
\newcommand {\subfloatheight} {4 cm}
\begin{figure}[h!]
     \subfloat [\label{fig:algtop_paral_Vs_Util_inifinite_LT_Kmin}\algtop \ (actual decline ratio is 0.4\%).] { % Fig.4a in NSDI 
    \begin{tikzpicture}
        \pgfplotsset{
            scale only axis,
            xmin=0, xmax=1000
        }
        
        \begin{axis}[
            width  = \subfloatwidth,
            height = \subfloatheight,
            axis y line*=left,
            ymin=0, ymax=1,
            xlabel = Time,
            ylabel = Utilization,
            ]
            \addplot[smooth,mark=triangle,mark size = \marksize] coordinates{
                (0, 0.00)
                (20, 0.04)
                (40, 0.07)
                (60, 0.10)
                (80, 0.13)
                (100, 0.16)
                (120, 0.19)
                (140, 0.22)
                (160, 0.25)
                (180, 0.28)
                (200, 0.31)
                (220, 0.34)
                (240, 0.38)
                (260, 0.41)
                (280, 0.44)
                (300, 0.47)
                (320, 0.49)
                (340, 0.53)
                (360, 0.55)
                (380, 0.57)
                (400, 0.60)
                (420, 0.63)
                (440, 0.66)
                (460, 0.68)
                (480, 0.72)
                (500, 0.75)
                (520, 0.77)
                (540, 0.80)
                (560, 0.84)
                (580, 0.86)
                (600, 0.89)
                (620, 0.90)
                (640, 0.91)
                (660, 0.92)
                (680, 0.93)
                (700, 0.94)
                (720, 0.94)
                (740, 0.94)
                (760, 0.95)
                (780, 0.95)
                (800, 0.95)
                (820, 0.95)
                (840, 0.95)
                (860, 0.96)
                (880, 0.96)
                (900, 0.96)
                (920, 0.96)
                }; \label{plot_util1}
        \end{axis}
        
        \begin{axis}[legend columns=1,
            width  = \subfloatwidth,
            height = \subfloatheight,
            axis y line*=right,
            axis x line=none,
            ymin=0, ymax=100,
            ylabel = Number of Schedulers,
            ytick={20, 40, 60, 80, 100},
            legend style     = {at={(0.78,0.65)},anchor=north,legend columns=-1},
            ]
            
        \addlegendimage{/pgfplots/refstyle=plot_util1}\addlegendentry{Utilization}
            \addplot[smooth,mark=*,mark size = \marksize] coordinates{
                (0, 87)
                (20, 87)
                (40, 86)
                (60, 85)
                (80, 84)
                (100, 83)
                (120, 82)
                (140, 80)
                (160, 78)
                (180, 76)
                (200, 74)
                (220, 71)
                (240, 69)
                (260, 66)
                (280, 63)
                (300, 60)
                (320, 57)
                (340, 54)
                (360, 51)
                (380, 48)
                (400, 45)
                (420, 42)
                (440, 39)
                (460, 36)
                (480, 33)
                (500, 29)
                (520, 26)
                (540, 23)
                (560, 20)
                (580, 16)
                (600, 14)
                (620, 11)
                (640, 9)
                (660, 7)
                (680, 6)
                (700, 5)
                (720, 4)
                (740, 3)
                (760, 2)
                (780, 2)
                (800, 2)
                (820, 1)
                (840, 1)
                (860, 1)
                (880, 1)
                (900, 1)
                (920, 1)
            }; \addlegendentry {\# Schedulers}
        \end{axis}
    \end{tikzpicture}
    } % End of subfloat
    
    \subfloat [\label{fig:algtop_paral_Vs_Util_inifinite_LT_Kavg}\algtopKavg \ (actual decline ratio is 0.8\%).] {% Fig.4b in NSDI
    \begin{tikzpicture}
        \pgfplotsset{
            scale only axis,
            xmin=0, xmax=1000
        }
        
        \begin{axis}[
            width  = \subfloatwidth,
            height = \subfloatheight,
            axis y line*=left,
            ymin=0, ymax=1,
            xlabel = Time,
            ylabel = Utilization,
            ]
            \addplot[smooth,mark=triangle,mark size = \marksize] coordinates{
                (0, 0.00)
                (20, 0.04)
                (40, 0.07)
                (60, 0.10)
                (80, 0.13)
                (100, 0.16)
                (120, 0.19)
                (140, 0.22)
                (160, 0.25)
                (180, 0.28)
                (200, 0.31)
                (220, 0.34)
                (240, 0.38)
                (260, 0.41)
                (280, 0.43)
                (300, 0.47)
                (320, 0.49)
                (340, 0.52)
                (360, 0.55)
                (380, 0.57)
                (400, 0.60)
                (420, 0.63)
                (440, 0.66)
                (460, 0.68)
                (480, 0.71)
                (500, 0.74)
                (520, 0.77)
                (540, 0.80)
                (560, 0.83)
                (580, 0.86)
                (600, 0.88)
                (620, 0.90)
                (640, 0.92)
                (660, 0.94)
                }; \label{plot_util2}
        \end{axis}
        
        \begin{axis}[
            width  = \subfloatwidth,
            height = \subfloatheight,
            axis y line*=right,
            axis x line=none,
            ymin=0, ymax=100,
            ylabel = Number of Schedulers,
            ytick={20, 40, 60, 80, 100},
            legend style     = {at={(0.6,0.3)},anchor=north,legend columns=-1},
            ]
            
        \addlegendimage{/pgfplots/refstyle=plot_util2}\addlegendentry{Utilization}
            \addplot[smooth,mark=*,mark size = \marksize] coordinates{
                (0, 87)
                (20, 87)
                (40, 87)
                (60, 86)
                (80, 86)
                (100, 86)
                (120, 86)
                (140, 85)
                (160, 85)
                (180, 84)
                (200, 84)
                (220, 83)
                (240, 83)
                (260, 82)
                (280, 81)
                (300, 80)
                (320, 79)
                (340, 77)
                (360, 76)
                (380, 75)
                (400, 73)
                (420, 72)
                (440, 70)
                (460, 68)
                (480, 66)
                (500, 64)
                (520, 62)
                (540, 59)
                (560, 56)
                (580, 53)
                (600, 50)
                (620, 47)
                (640, 43)
                (660, 40)
            }; \addlegendentry {\# Schedulers}
        \end{axis}
    \end{tikzpicture}
    } % End of subfloat
	\caption[Resource utilization and number of schedulers in \algtop\ (requests have infinite lifetime)]{Cloud resource utilization and the number of schedulers in \algtop\ for the NFV dataset under Poisson arrivals (requests have infinite lifetime).}
	\label{fig:algtop_paral_Vs_Util_infinite_LT}
	\vspace{\vspacebelowcaption}
\end{figure}
\\~\\
\textbf{Under the Hood of \algtop.}
Our next set of experiments studies the interplay between the system's utilization and the level of parallelism offered by \algtop. For these experiments we use solely the NFV dataset. 

Fig.~\ref{fig:algtop_paral_Vs_Util_inifinite_LT_Kmin} depicts the number of schedulers and the system utilization of \algtop. Initially, \algtop\ employs many schedulers as there are many available hosts for any flavor. As the utilization increases and the number of available hosts decreases, \algtop\ gradually reduces the number of schedulers. Intuitively, reducing the number of schedulers serves two goals: First, it allows each scheduler to query more hosts while still complying with the budget constraint. This increases the probability of finding an available host.  Second, having fewer schedulers reduces the collision probability. % mitigates the probability of collisions.

Recall that \algtop\ uses a conservative approach in estimating the number of available hosts ($\nfreebins$). This conservative approach indeed yields a very low decline ratio ($0.4\%$) -- but at the cost of throttling parallelism when utilization ramps up.
We therefore consider a variant of \algtop, which we dub \algtopKavg. As its name suggests, this variant differs from Algorithm~\ref{alg:estimateK} in Line~\ref{alg:algtop:assign_k_n}, where it updates $\nfreebins$ according to the {\em average} number of available hosts taken over all flavors.

Fig.~\ref{fig:algtop_paral_Vs_Util_inifinite_LT_Kavg} shows that \algtopKavg\ runs a significantly higher number of schedulers than \algtop, for any given level of utilization. As a result, \algtopKavg\ finishes allocating all requests much faster than \algtop, implying a higher throughput. Indeed, switching from \algtop\ to \algtopKavg\ doubles the actual decline ratio to $0.8\%$ -- but this value is still way below the target decline ratio of $\dr = 5\%$. 

\input{tikz_sim_finite_LT.tex}

Our next experiment aims at exploring how either \algtop\ and \algtopKavg\ dynamically adjust the number of schedulers when the utilization fluctuates. To generate fluctuations in the utilization, we modeled the request arrival process as a variant of a \emph{Markov Modulated Poisson Process (MMPP)}~\cite{cite_MMPP}. Specifically, the number of requests arriving per slot is a Poisson process with mean $\lambda_a$ throughout the experiment. However, for the first $20\%$ of the requests we use $\lambda_a=20$, to fill up the system; while for the rest of the requests we fix $\lambda_a=5$. Furthermore, in this experiment requests have a finite lifetime. That is, the number of allocated requests leaving per time slot follows a Poisson process with mean $\lambda_d = 4$. Finally, we use here 100 replicas of the NFV dataset 
(instead of 30 used in the rest of this section) so that even when requests leave their hosts, the hosts become utilized again with more arriving requests. These settings are intended to let utilization first build-up, and then stay at some (high) level, with mild fluctuations. 
The results of this experiment are shown in Fig.~\ref{fig:algtop_paral_Vs_Util_finite_LT}. Both \algtop\ (Fig.~\ref{fig:algtop_paral_Vs_Util_finite_LT_Kmin}) and \algtopKavg\
(Fig.~\ref{fig:algtop_paral_Vs_Util_finite_LT_Kavg})
dynamically adapt the number of schedulers  to the utilization. However, \algtopKavg\ employs more schedulers than \algtop. It obtains shorter total run-time but experiences a higher decline ratio ($1.6\%$ for \algtopKavg\ versus $0.01\%$ for \algtop). Note that both algorithms are below the maximum allowed decline ratio ($5\%$).
 
We now investigate the effect of the query budget $\budget$ on the number of schedulers. We use the same settings of the MMPP model as in the previous experiment. However, this time we vary the budget from $20\%$ to $100\%$ of the number of hosts, and measure the level of parallelism, captured by the average number of schedulers which \algtop\ employs along the run. Table~\ref{table:aspr_budget} illustrates the results. Indeed, parallelism is proportional to the given budget. However, the increase in parallelism when increasing the budget is very mild. In particular, decreasing the budget from 
$\nbins$ to $0.6 \cdot \nbins$ does not decrease the average number of schedulers used by \algtop. That is, \algtop\ runs 17 schedulers with less communication overhead than a single OpenStack scheduler. However, running above 17 schedulers increases the probability of collisions. 

\begin{table}[h!]
	\centering
	\begin{tabular}{c|c|c|c|c|c|}
	\cline{2-6}
		&\multicolumn{5}{c|}{ Budget } \\
		 \cline{2-6}
		 & 20\% & 40\% & 60\% & 80\% &  100\% \\ 
		 	\hline
			\multicolumn{1}{|c|}{ $\bar{s}$} & 12 & 16 & 17 & 17 &  17\\

		\hline
	\end{tabular}
    \caption[Average number of schedulers per budget]{Average number of schedulers per budget \textit{B}, given as a percentage of total number of hosts ($\dr = 5\%$).} 
    \vspace{\vspacebelowcaption}
    \label{table:aspr_budget}
\end{table}
    \vspace{\vspacebelowcaption}

\section {Discussion}\label{sec:apsr:conclusions}
In this chapter, we address the problem of fast placement of virtual machines of virtualized network functions in large clouds. We focus our attention in the interplay between throughput, communication overhead and decline ratio. 
We study the performance of existing placement algorithms using three real-world workloads. Our study shows that randomization is a key component in boosting throughput by employing multiple parallel schedulers, while keeping a low decline ratio. Based on this observation we develop the algorithm \algtop, which efficiently implements random placement while minimizing the communication overhead, and dynamically adjusts the degree of parallelism to ensure that decline ratios are kept at their SLA. Using a balls-and-bins combinatorial analysis, we formally prove the correctness of \algtop\ and provide insights into the possibilities and limitations of parallel resource management. We evaluate \algtop\ on three real workloads and demonstrate its capability to provide high degrees of parallelism with small decline ratios, and low communication overheads. 
% We then integrate \algtop\ into the OpenStack cloud management platform. 
We show that \algtop\ matches the best attainable throughput of OpenStack's default Filter Scheduler, while reducing the decline ratio from up to $13.6\%$ to $\approx 1\%$, and the communication overheads by $\approx 20x$. 

However, our evaluation shows that OpenStack only gains up to $3x$ speedup from parallelism, whereas \algtop\ easily supports many parallel schedulers. Therefore, it is appealing to carefully benchmark OpenStack, identify its current bottlenecks and unleash its full potential for parallel resource management. In addition, \algtop\ uses a central controller which gathers the state from all the parallel schedulers, and adjust the number of schedulers and the number of hosts queried by each scheduler accordingly. It is of interest to consider a distributed settings, in which each scheduler adapts its level of functionality and the number of hosts it queries dynamically, based on the system's state it collects and on the ratio of declines it experiences.  

\begin{appendices}
\chapter{Running Example of \SAM}\label{sec:buf:run_example}
For better understanding of \SAM\ and exploration of its characteristics, we now provide a running example of it.
Note that we use here the notation of Chapter~\ref{sec:buf}, which is summarized in Table~\ref{tbl:buf:notations}.

Figure~\ref{fig:algorithm-example} exemplifies a running of \SAM\ equipped with a 3-slots buffer. Each packet is represented by a square. If it is a known $(w, v)$-packet, then $(w,v)$ (namely, its work, profit values, resp.) appears within the square representing the packet. If the packet is unknown, however, the (unknown) work and profit values do not appear, and the packet's color is dark gray.

Known packets that belong to the selected class ($\CsK$-packets) are marked in light gray.
The figure assumes that the (randomly-chosen) selected class is the class of packets with work- and profit- values within the range [3, 4]. Recall, that this range refers to the
characteristics of a packet \emph{upon arrival}.
For instance, a $(3, 3)$-packet always belongs to the selected class, although after being processed its residual work decreases, and it becomes a $(2, 3)$-packet, and later a $(1, 3)$-packet, and so on.

\begin{figure}
    \centering
    \subfloat[\label{fig:algorithm-example_part1}] {
        \includegraphics[width=1.0\textwidth]{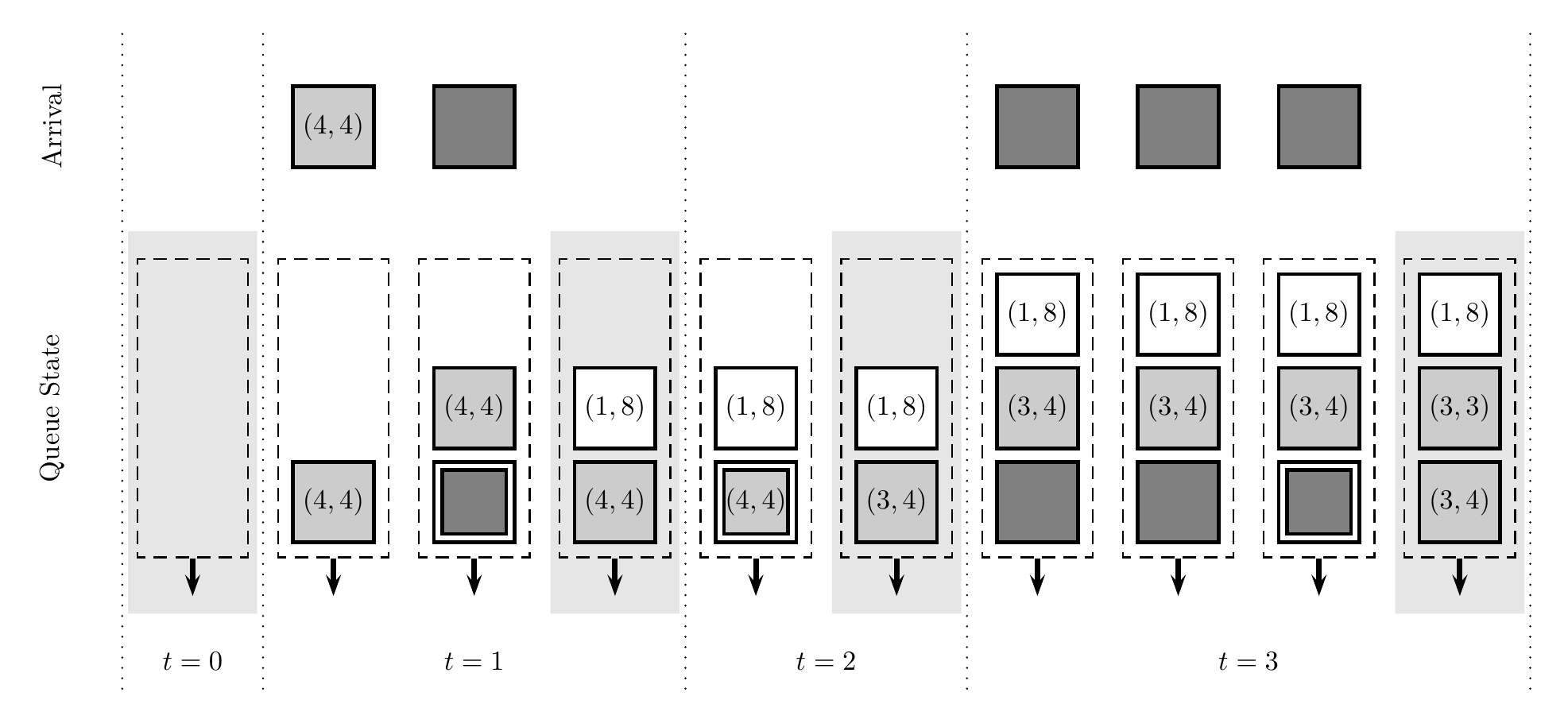}
    } 
    
    \subfloat[\label{fig:algorithm-example_part2}] {
        \includegraphics[width=1.0\textwidth]{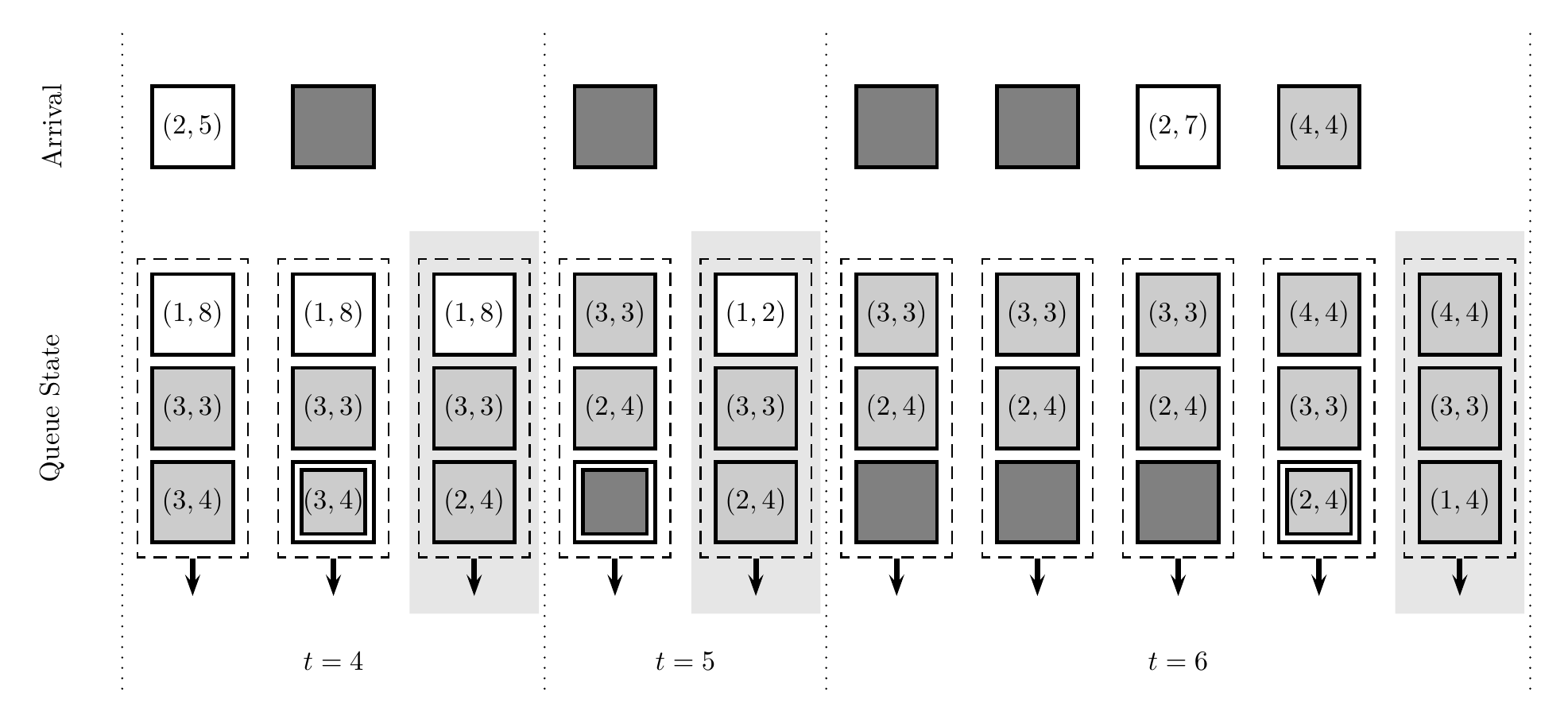}
    } 
     \caption[Running Example of \SAM]{
        Running Example of \SAM, equipped with a 3-slots buffer. Each known packet is labeled $(w(p),v(p))$, where $w$ is the remaining work and $v$ is the profit. $\CsK$-packets are marked by light gray. $U$-packets are colored by dark gray. \newline
        Each cycle begins with a transmission step, in which a fully-processed packet, if such exists, is transmitted. Next comes the arrival step, where arriving packets are handled by the algorithm one by one. For each arriving packet, the buffer below the arrival depicts the state of the buffer after handling the packet's arrival.
        The packet in the queue's head-of-line (HoL) at the end of the arrival step is emphasized by an extra, internal, square. This packet is the one processed in the processing step. The state of the buffer at the end of each cycle is highlighted with a light-gray background.
     }
    \label{fig:algorithm-example}
\end{figure}

Each cycle begins with the transmission step, in which a fully processed packet, if such exists, leaves the queue.
In our example there is no packet transmitted since we focus our attention on handling arrivals and determining priorities which are the core components of our algorithm.
This step is followed by the arrival step,
where arriving packets are handled by the algorithm. Finally,
the cycle ends with a processing step, where the head-of-line (HoL) packet is processed. This packet is emphasized by an extra internal square.
The state of the queue at the end of each cycle is depicted by a light-gray background.
At each cycle, the algorithm tosses a coin, and assigns the cycle as an \emph{admittance cycle} w.p. $r$. In this example, we assume that cycles $1,3,5, 6$ are admittance cycles.
We now turn to explain the scenario depicted in Figure~\ref{fig:algorithm-example} cycle by cycle.
\paragraph{\textbf{$t=0$.}} Begin with an empty buffer.
\paragraph{\textbf{$t=1$.}} A known $(4,4)$-packet arrives. As both its work- and profit- values belong to the ranges [3,4], it is a $\CsK$-packet, and therefore it is retained by the algorithm (recall that $\CsK$-packets are never dropped during the fill phase, as shown in Proposition~\ref{prop:transmits_all_G_K}).

Next, a $U$-packet arrives. As this is an admittance cycle, this $U$-packet is \emph{admitted}, that is, accepted into the buffer, and assigned to the HoL.
Since this is the last packet to arrive in this cycle, and being the HoL-packet, this packet is processed in the processing step. We refer to this packet as being \emph{parsed} in this cycle, as this is the first processing cycle of this packet.

After parsing, the characteristics of
the HoL packet become known: it is now a known $(1,8)$-packet.
Namely, when it arrived, it was a $(2, 8)$-packet which has received one cycle of processing.
By these values, this packet does not belong to the selected class. Therefore, it is pushed down to the buffer's tail. Instead, the $\CsK$-packet, with values $(4,4)$ is assigned to be the HoL packet.
It should be noted that although the parsed $(1,8)$-packet is superior to any $\CsK$-packet currently in the buffer (since it carries a profit value of 8 while requiring just one more cycle of work), \SAM\ still prefers $\CsK$-packets over this packet. 
We note that the improved \SAO\ algorithm would re-assign such a packet to be a $\CsK$-packet by considering the selective class closure.

\paragraph{\textbf{$t=2$.}} No packets arrive. The HoL-packet, $(4,4)$, is processed, and becomes a $(3,4)$-packet.

\paragraph{\textbf{$t=3$.}}
This is an admittance cycle. Therefore, the first arriving $U$-packet is admitted. In particular, this cycle well exemplifies the buffer's ordering: at top-priority is the admitted packet; at a second priority is the $\CsK$-packet, $(3,4)$; the remaining packet in the buffer, $(1,8)$, is of \redtext{the} lowest priority.

When a second $U$-packet arrives, \SAM\ tosses a coin, and replaces the previously-admitted packet with the new arriving $U$-packet w.p. $1/2$. When a third $U$-packet arrives, \SAM\ tosses a coin again, and replaces the previously-admitted packet with the new arriving $U$-packet w.p. $1/3$.

In the processing step, \SAM\ parses the admitted packet, unraveling it as a $(3,3)$-packet. Namely, upon arrival its characteristics were $(4,3)$, ascribing it to the selected class. As there already exists another $\CsK$-packet in the buffer (the $(3,4)$-packet) \SAM\ breaks the tie between the two $\CsK$-packets in its buffer by FIFO order.
We note that the improved \SAO\ algorithm would transition to the flush phase at this point, since it would have been full of $\CsK$-packets.

\paragraph{\textbf{$t=4$.}} First, we have an arriving known $(2,5)$-packet. By its characteristics, it is not a $\CsK$-packet. Therefore, it is assigned the lowest priority. In particular, as the buffer is full, this packet is discarded. Next, a $U$-packet arrives. However, as this is a non-admittance cycle, the $U$-packet is discarded as well. Finally, during the processing step, the HoL packet is processed, decreasing its remaining work to 2.

\paragraph{\textbf{$t=5$.}} We have a single arriving $U$-packet. As it is an admitted cycle, this $U$-packet is admitted, hence, accepted and parsed. In order to make room for this admitted packet, the $(1,8)$-packet in the tail is pushed-out and dropped.
 After parsing, the $U$-packet is uncovered as a $(1, 2)$-packet. Namely, upon arrival it was a $(2,2)$-packet. By these characteristics, this packet does not belong to the selected class, and therefore has the lowest priority, and downgraded to the tail.

\paragraph{\textbf{$t=6$.}} This is an admittance cycle. Therefore the first arriving $U$-packet is admitted, pushing-out from the buffer the $(1,2)$-packet, which was in the tail. When a second $U$-packet arrives, it replaces the previously-admitted packet w.p. $  1/2$. Then, a $(2,7)$-packet arrives. By its characteristics, it is neither an admitted packet (as it is a $K$-packet), nor does it belong to the selected class. As a result, the $(2,7)$-packet is assigned the lowest priority, and is therefore discarded.
The last arrival in this cycle is a known $(4,4)$-packet. By its characteristics, it is a $\CsK$-packet. Since the buffer already contains $B-1=2$ $\CsK$-packets, the $U$-packet at the HoL is dropped, and the newly arriving $CsK$-packet is accepted to the queue (see lines~\ref{alg:SAM:if_is_Gk}-\ref{alg:SAM:accept_Gk} in Algorithm~\ref{alg:SAM}).
The queue therefore becomes \hfull, i.e., the buffer is full with $\CsK$-packets. \SAM\ then switches to the \emph{flush} state, and it will merely process all the packets in its buffer in a run-to-completion manner and transmit all the fully-processed packets, until the buffer is empty again.
\end{appendices}

\bibliographystyle{IEEEtran}
\newpage
\bibliography{Refs}
\end{document}